\documentclass[aos]{imsart}

\RequirePackage{amsthm,amsmath,amsfonts,amssymb}
\RequirePackage[authoryear]{natbib}%% uncomment this for author-year citations
\RequirePackage[colorlinks,citecolor=blue,urlcolor=blue]{hyperref}
\RequirePackage{graphicx}

\usepackage{graphicx}
\usepackage{mathrsfs}
\usepackage{multirow}
\usepackage{amsfonts}
\usepackage{amsmath}
\usepackage{comment}
\usepackage{amsthm}
\usepackage{color}
\usepackage{mathtools}
\usepackage{algorithm}
\usepackage{algorithmic}
\usepackage{sectsty}
\usepackage[colorlinks]{hyperref}
\usepackage{amsmath}
\usepackage{amssymb}
\usepackage[toc,page]{appendix}
\usepackage[mathscr]{euscript} 
\usepackage[toc]{appendix}

\usepackage{chngcntr}
\usepackage{apptools}
\usepackage{booktabs}
\usepackage{lscape}
\usepackage{arydshln}
\usepackage{soul}
\usepackage{epstopdf}
\usepackage{xr}
\usepackage{accents}
\newcommand{\ubar}[1]{\underaccent{\bar}{#1}}
\usepackage{tikz}
\usetikzlibrary{positioning}
\usepackage{mathdots}

\DeclareFontFamily{U}{mathx}{}
\DeclareFontShape{U}{mathx}{m}{n}{<-> mathx10}{}
\DeclareSymbolFont{mathx}{U}{mathx}{m}{n}
\DeclareMathAccent{\widecheck}{0}{mathx}{"71}

\mathtoolsset{showonlyrefs}

\makeatletter
\newcommand*{\addFileDependency}[1]{% argument=file name and extension
  \typeout{(#1)}
  \@addtofilelist{#1}
  \IfFileExists{#1}{}{\typeout{No file #1.}}
}
\makeatother

\newcommand*{\myexternaldocument}[1]{%
    \externaldocument{#1}%
    \addFileDependency{#1.tex}%
    \addFileDependency{#1.aux}%
}

\myexternaldocument{Supp}

\startlocaldefs
\numberwithin{equation}{section}
\theoremstyle{plain}

\newtheorem{theorem}{Theorem}[section]
\newtheorem{lemma}[theorem]{Lemma}
\newtheorem{proposition}[theorem]{Proposition}
\newtheorem{remark}[theorem]{Remark}
\newtheorem{corollary}[theorem]{Corollary}
\newtheorem{assumption}[theorem]{Assumption}
\theoremstyle{remark}

\newcommand{\bm}{\mathbf}

\newcommand{\cm}[1]{\mbox{\boldmath$\mathscr{#1}$}}

\endlocaldefs

\begin{document}

\begin{frontmatter}
\title{HIGH-DIMENSIONAL LOW-RANK MATRIX REGRESSION WITH UNKNOWN LATENT STRUCTURES}
\runtitle{HOMOGENEITY PURSUIT IN MATRIX REGRESSION}

\begin{aug}
%%%%%%%%%%%%%%%%%%%%%%%%%%%%%%%%%%%%%%%%%%%%%%%
%% Only one address is permitted per author. %%
%% Only division, organization and e-mail is %%
%% included in the address.                  %%
%% Additional information can be included in %%
\author[A]{\fnms{Di}~\snm{Wang}},
\author[B]{\fnms{Xiaoyu}~\snm{Zhang}}\footnote{Di Wang and Xiaoyu Zhang are joint first authors.},
\author[C]{\fnms{Guodong}~\snm{Li}}
\and
\author[D]{\fnms{Wenyang}~\snm{Zhang}}

%%%%%%%%%%%%%%%%%%%%%%%%%%%%%%%%%%%%%%%%%%%%%%
%% Addresses                                %%
%%%%%%%%%%%%%%%%%%%%%%%%%%%%%%%%%%%%%%%%%%%%%%
\address[A]{School of Mathematical Sciences,
Shanghai Jiao Tong University, {\color{blue}di.wang@sjtu.edu.cn}}

\address[B]{School of Mathematical Sciences,
Tongji University, {\color{blue}xzhangck@tongji.edu.cn}}

\address[C]{Department of Statistics and Actuarial Science,
University of Hong Kong, {\color{blue}gdli@hku.hk}}

\address[D]{Faculty of Business Administration, University of Macau, {\color{blue}wenyangzhang@um.edu.mo}}
\end{aug}

\begin{abstract}

We study low-rank matrix regression in settings where matrix-valued predictors and scalar responses are observed across multiple individuals. Rather than assuming a fully homogeneous coefficient matrices across individuals, we accommodate shared low-dimensional structure alongside individual-specific deviations. To this end, we introduce a tensor-structured homogeneity pursuit framework, wherein each coefficient matrix is represented as a product of shared low-rank subspaces and individualized low-rank loadings. We propose a scalable estimation procedure based on scaled gradient descent, and establish non-asymptotic bounds demonstrating that the proposed estimator attains improved convergence rates by leveraging shared information while preserving individual-specific signals. The framework is further extended to incorporate scaled hard thresholding for recovering sparse latent structures, with theoretical guarantees in both linear and generalized linear model settings. Our approach provides a principled middle ground between fully pooled and fully separate analyses, achieving strong theoretical performance, computational tractability, and interpretability in high-dimensional multi-individual matrix regression problems.

\end{abstract}

\begin{keyword}[class=MSC]
\kwd[Primary ]{62J05}
\kwd[; secondary ]{62H12}
\kwd{62H25}
\end{keyword}

\begin{keyword}
\kwd{Dimensional reduction}
\kwd{Homogeneity pursuit}
\kwd{Matrix data}
\kwd{Subgroup analysis}
\kwd{Subspace}
\end{keyword}

\end{frontmatter}
%%%%%%%%%%%%%%%%%%%%%%%%%%%%%%%%%%%%%%%%%%%%%%
%% Please use \tableofcontents for articles %%
%% with 50 pages and more                   %%
%%%%%%%%%%%%%%%%%%%%%%%%%%%%%%%%%%%%%%%%%%%%%%
%\tableofcontents

\section{Introduction}

\subsection{Heterogeneity and Homogeneity in Matrix Regression}

Matrix-valued data are ubiquitous in modern scientific applications, including neuroimaging \citep[e.g., functional connectivity matrices;][]{abraham2020deriving}, genomics \citep[e.g., gene co-expression networks;][]{aibar2022singlecell}, recommender systems \citep[user-item interaction matrices;][]{zhang2021tensor}, and quantitative finance \citep[asset return dependency matrices;][]{bollerslev2021volatility}. In many of these domains, data are collected from multiple individuals, experimental units, or observational contexts, giving rise to multi-individual matrix regression problems with matrix-valued predictors and scalar responses. The primary goal is to model the relationship between matrix predictors and scalar outcomes while addressing high dimensionality, complex structure, and limited sample sizes per individual.

A central challenge in such settings is that regression relationships may exhibit both shared structure across individuals and individual-specific deviations. Fully pooled models, which assume a single coefficient matrix for all individuals, risk bias due to ignored heterogeneity. Conversely, fully separate models, which fit independent regressions per individual, often incur high variance due to limited data per unit. These extremes motivate the development of methods that judiciously balance shared information and individual adaptations. This problem, commonly referred to as \textit{homogeneity pursuit}, involves identifying and leveraging latent common structures across individuals while appropriately accommodating individual variation.

In high-dimensional settings, it is typical to assume that the coefficient matrix is low rank, reflecting the influence of a small number of latent factors. Low-rank matrix regression enhances interpretability and computational efficiency, and has been extensively studied in the literature \citep{negahban2011estimation,koltchinskii2011nuclear}. However, in multi-individual settings, low-rank assumptions alone are insufficient. The same low-rank matrix may arise from many diverse factorizations, and existing methods often fail to disentangle shared structures, such as common row or column spaces, from individual heterogeneity. Moreover, conventional low-rank formulations do not explicitly model dependence across individuals, limiting their ability to pool information effectively.

To address these challenges, we propose a framework for homogeneity pursuit, in which shared low-dimensional subspaces (e.g., column and row spaces of the coefficient matrices) are explicitly modeled as common across individuals, while individual-specific deviations are captured via structured low-rank loadings. This leads to a tensor-structured representation of the coefficient matrices, where shared structure is encoded via factor matrices and individual heterogeneity is captured via a low-rank core tensor encoding individualized loadings. Our approach provides a structured and interpretable decomposition that distinguishes between common mechanisms and individual variation, offering a natural middle ground between fully pooled and fully separate analyses.

We develop a computationally scalable estimation procedure based on scaled gradient descent, tailored to the proposed tensor decomposition. The algorithm is designed to handle inherent scaling differences between shared and individual components, and we establish non-asymptotic error bounds demonstrating the estimation improvements when leveraging shared information across individuals. Specifically, convergence rates for the shared subspaces benefit from the larger number of individuals, while individual-specific components retain their distinct signals. To further enhance interpretability and reduce model complexity, we apply hard thresholding to the rows of the shared factor matrices for the recovery of sparse patterns.

The proposed methodology is applicable to both linear models, such as matrix trace regression, and generalized linear models, including matrix logistic regression, making it suitable for a broad range of response types. Overall, our approach provides a flexible, principled, and theoretically grounded framework for multi-individual matrix regression that effectively balances shared and individual information. Below, we summarize our key contributions.
\begin{itemize}
  \item[1.] We introduce a tensor-structured model for multi-individual matrix regression, in which shared low-dimensional subspaces are encoded via factor matrices and individual deviations are captured by a low-rank core tensor. This decomposition enables structured homogeneity pursuit, distinguishing common mechanisms from individual variation.
  \item[2.] We formalize the notion of shared subspaces and introduce a distance metric that is invariant to orthogonal transformations. This metric supports robust estimation and provides a basis for comparing decompositions while accounting for inherent non-identifiability in tensor representations.
  \item[3.] We develop a scaled gradient descent algorithm tailored to the tensor decomposition, with adaptive scaling to balance updates across shared and individual parameters. We establish minimax optimal error bounds demonstrating improved estimation through information pooling.
  \item[4.] We extend the framework to incorporate additional sparsity in shared components to further enhance interpretability and improve efficiency. The proposed approach is applicable to both linear and generalized linear models, accommodating a wide range of scientific applications.
\end{itemize}

\subsection{Related Literature}

Regression with matrix-valued predictors has been extensively studied in both low- and high-dimensional settings. Classical approaches include trace regression models\citep{rodriguez2011unifying, Zhou13}, as well as more general linear or nonlinear models acting on the matrix structure. In high-dimensional regimes where matrix dimensions dominate the sample size, low-rank regularization has emerged as a powerful tool for achieving estimation efficiency and structural interpretability \citep{negahban2011estimation, koltchinskii2011nuclear,yu2020recovery}.

Tensor decomposition methods have also been employed to model multi-way data structures, including those arising in multi-individual or multi-task regression \citep{yu2016learning,lock2018tensor,xu2019spatio}. While tensor-based approaches are natural for representing hierarchical or grouped structures, many existing methods lack rigorous theoretical guarantees under high-dimensional scaling. Regularized estimation methods have been developed to address identifiability issues in low-rank matrix and tensor problems \citep{tu2016low,wang2017unified,han2022optimal}. More recently, \citet{tong2021accelerating} and \citet{tong2022accelerating} proposed scaled gradient descent algorithms for low-rank matrix and tensor problems. In comparison, our work incorporates additional sparsity structures for improved interpretability and estalibshes statistical guarantees in a more complex setting involving both homogeneous and heterogeneous components.

Closer to our setting, recent advances have explored homogeneity structures in heterogeneous data analysis. Homogeneity pursuit was introduced by \citet{ke2015homogeneity} and was extended to panel and longitudinal data models with latent structures \citep{,wang2018homogeneity,li2019spatial,li2019subgroup,lian2021homogeneity,zhang2023nonparametric}. Subgroup analysis has also seen substantial development \citep{ke2016structure,ma2017concave,zhang2019robust}. Our work contributes to this literature by introducing the homogeneity at the level of subspaces, rather than element-wise. We explicitly model shared subspaces and individual-specific loading factors, developing an estimation procedure with provable guarantees. Our theoretical results quantify the benefits of information pooling while preserving individual heterogeneity.

\subsection{Notation}

Throught this paper, we denote vectors by boldface small letters (e.g. $\bm{x}$), matrices by boldface capital letters (e.g. $\bm{X}$), and tensors by boldface Euler letters (e.g. $\cm{X}$). Tensor notations and operations are introduced below and formally defined in Appendix \ref{append:A}. For a tensor $\cm{X}\in\mathbb{R}^{p_1\times\cdots\times p_d}$, $\cm{Y}\in\mathbb{R}^{p_1\times\cdots\times p_{d}}$, and matrices $\bm{Y}_k\in\mathbb{R}^{q_k\times p_k}$ for $k=1,\dots,d$, we denote the mode-$k$ matricization of $\cm{X}$ as $\cm{X}_{(k)}$, the generalized inner product of $\cm{X}$ and $\cm{Y}$ as $\langle\cm{X},\cm{Y}\rangle$, and the mode-$k$ product of $\cm{X}$ and $\bm{Y}$ as $\cm{X}\times_k\bm{Y}_k$. The tensor outer product is denoted by $\cm{X}\circ\cm{Y}$.

For a generic matrix $\bm{X}$, we use $\bm{X}^\top$, $\|\bm{X}\|_\text{F}$, $\|\bm{X}\|$, $\text{vec}(\bm{X})$, and $\sigma_j(\bm{X})$ to denote its transpose, Frobenius norm, operator norm, vectorization, and the $j$-th largest singular value, respectively. For symmetric matrices, $\lambda_{\min}(\bm{X})$ and $\lambda_{\max}(\bm{X})$ denote the minimum and maximum eigenvalues. We dnote the general linear group of degree $n$ by $\text{GL}(n)$. A generic positive constant is denoted by $C$. For sequences $x_k$ and $y_k$, we write $x_k\gtrsim y_k$ if there exists a constant $C>0$ such that $x_k\geq Cy_k$ for all $k$, and $x_k\asymp y_k$ if $x_k\gtrsim y_k$ and $y_k\gtrsim x_k$.

\subsection{Outline of the Paper}

The remainder of the paper is organized as follows. Section \ref{sec:model_framework} introduces the homogeneity pursuit model and its tensor decomposition, including a discussion of identifiability. Section \ref{sec:estimation_methodology} presents the scaled gradient descent algorithm and its computational properties. Section \ref{sec:sparsity} extends the framework to incorporate scaled hard thresholding for sparse pattern recovery. Section \ref{sec:statistical_models} details applications to matrix linear and generalized linear models. Section \ref{sec:numerical} reports simulation results and presents a real data application. Section \ref{sec:conclusion} concludes with a discussion. Technical proofs and implementation details are provided in Appendices.

\section{Model Framework}\label{sec:model_framework}

\subsection{Homogeneity Pursuit via Subspace Sharing}

We consider a dataset with $n$ individuals indexed by $i=1,\dots,n$. For each individual $i$, we observe $m_i$ independent samples, each consisting of a matrix-valued covariate $\bm{X}_{ij}\in\mathbb{R}^{p_1\times p_2}$ and a scalar response $Y_{ij}$. We aim to model the relationship between the matrix predictors and scalar responses through a regression framework that accommodates high dimensionality, heterogeneity across individuals, and structured dependencies. We adopt a generalized linear model formulation, where the conditional mean of the response depends linearly on the matrix covariate via an individual-specific coefficient matrix $\bm{B}_i\in\mathbb{R}^{p_1\times p_2}$:
\begin{equation}
  f_{Y_{ij}}(y_{ij}|\bm{X}_{ij}) \propto \exp\left\{\frac{y_{ij}\langle\bm{X}_{ij},\bm{B}_i\rangle - g(\langle\bm{X}_{ij},\bm{B}_i\rangle)}{c(\sigma)}\right\},
\end{equation}
where $c(\sigma)$ is a scale nuisance parameter and $g(\cdot)$ is a differentiable link function.

The negative log-likelihood loss is given by
\begin{equation}
    \mathcal{L}(\bm{B}_i;\bm{X}_{ij},Y_{ij}) = g(\langle\bm{X}_{ij},\bm{B}_i\rangle)-Y_{ij}\langle\bm{X}_{ij},\bm{B}_i\rangle.
\end{equation}
This framework encompasses a variety of standard regression models. For exmaple, when $g(t)=t^2/2$, the model reduces to matrix trace regression \citep{negahban2011estimation}; when $g(t)=\log(1+\exp(t))$, it corresponds to matrix logistic regression \citep{Zhou13}. 

A common assumption in high-dimensional matrix regression is that $\bm{B}_i$ admits a low-rank decomposition $\bm{B}_i=\bm{C}_i\bm{R}_i^\top$, where $\bm{C}_i\in\mathbb{R}^{p_1\times r}$ and $\bm{R}_i\in\mathbb{R}^{p_2\times r}$ are matrices of rank $r$, typically much smaller than $p_1$ and $p_2$. This factorization reflects the influence of latent factors and facilitates dimension reduction. The inner product can be equivalently expressed as $\langle\bm{X}_{ij},\bm{C}_i\bm{R}_i^\top\rangle = \text{tr}(\bm{C}_i^\top\bm{X}_{ij}\bm{R}_i)$, which reparametrizes the regression in terms of transformations of the covariates induced by the low-rank factors.

However, such a decomposition is not unique: for any orthogonal matrix $\bm{Q}\in\mathbb{R}^{r\times r}$, the pair $(\bm{C}_i,\bm{R}_i)$ is indistinguishable from $(\bm{C}_i\bm{Q},\bm{R}_i\bm{Q}^{-\top})$. This non-identifiability does not extend to the column spaces of $\bm{C}_i$ and $\bm{R}_i$, which are uniquely determined up to orthogonal transformations. These subspaces encode the dominant directions of covariate variation relevant to the response, and it is at this level that meaningful homogeneity may emerge.

Motivated by this, we focus on the subspaces spanned by $\bm{C}_i$ and $\bm{R}_i$, assuming that the column spaces of $\{\bm{C}_i\}_{i=1}^n$ share a common $K_1$-dimensional subspace, and similarly for $\{\bm{R}_i\}_{i=1}^n$ with dimension $K_2$. Formally, letting $\bar{\bm{C}}=[\bm{C}_1,\dots,\bm{C}_n]\in\mathbb{R}^{p_1\times nr}$ and $\bar{\bm{R}}=[\bm{R}_1,\dots,\bm{R}_n]\in\mathbb{R}^{p_2\times nr}$, we assume the existence of matrices $\bm{C}\in\mathbb{R}^{p_1\times K_1}$ and $\bm{R}\in\mathbb{R}^{p_2\times K_2}$ such that $\bar{\bm{C}}=\bm{C}\bm{L}_1$ and $\bar{\bm{R}}=\bm{R}\bm{L}_2$, where $\bm{L}_1=[\bm{L}_{11},\dots,\bm{L}_{1n}]$ and $\bm{L}_2=[\bm{L}_{21},\dots,\bm{L}_{2n}]$ are individual-specific loading matrices. Consequently, for each individual $i$, $\bm{C}_i=\bm{C}\bm{L}_{1i}$ and $\bm{R}_i=\bm{R}\bm{L}_{2i}$. The integers $K_1$ and $K_2$ are assumed to be small relative to $p_1$ and $p_2$, reflecting low-dimensional nature of the shared structure.

Under this formulation, the inner product $\langle\bm{X}_{ij},\bm{B}_i\rangle$ becomes
\begin{equation}
    \langle\bm{X}_{ij},\bm{B}_i\rangle = \text{tr}(\bm{L}_{1i}^\top\bm{C}^\top\bm{X}_{ij}\bm{R}\bm{L}_{2j}),
\end{equation}
where $\bm{C}^\top\bm{X}_{ij}\bm{R}$ is a bilinear transformation shared across individuals, modulated by individual-specific matrices $\bm{L}_{1i}$ and $\bm{L}_{2i}$. When the loading matrices are identical across individuals, i.e., $\bm{L}_{1i}=\bm{L}_1$ and $\bm{L}_{2i}=\bm{L}_2$ for all $i$, the model reduces to a fully homogeneous low-rank regression. More generally, our formulation allows for homogeneity in the underlying subspaces, while accommodating individual-level heterogeneity through the low-rank loadings.

This approach offers several advantages. First, by focusing on the shared subspaces rather than the full low-rank factors, we reduce the number of parameters required to describe the heterogeneity across individuals, leading to improved estimation efficiency, particularly in settings where the number of observations per individual is small. Second, the decomposition is interpretable, as the shared subspaces capture the common structure across individuals, while the loading matrices encode individual deviations. Third, the non-uniqueness issues associated with low-rank matrix decompositions are mitigated by working directly with the subspaces, which are uniquely defined up to orthogonal transformations.

\subsection{Tensor Parameterization}

To formalize the structured relationship between individuals and the shared and individualized components, we adopt a tensor-based parameterization. Specifically, we organize the individual coefficient matrices $\bm{B}_1,\dots,\bm{B}_n\in\mathbb{R}^{p_1\times p_2}$ into a third-order tensor $\cm{B}\in\mathbb{R}^{p_1\times p_2\times n}$, where the $i$-th frontal slice is given by $\cm{B}_{(3)}(:,:,i)=\bm{B}_i$. 

The homogeneity in the underlying subspaces is captured via a Tucker decomposition of $\cm{B}$ \citep{tucker1966some,delathauwer2000multilinear}, decomposing it into a set of shared factor matrices and an individualized core tensor. Formally, we express the decomposition as
\begin{equation}\label{eq:tensor_decomp}
    \cm{B} = \cm{G} \times_1 \bm{C} \times_2 \bm{R}\quad\text{or equivalently}\quad\bm{B}_i=\bm{C}\bm{G}_i\bm{R}^\top,\quad\text{with }\bm{G}_i=\bm{L}_{1i}\bm{L}_{2i}^\top,
\end{equation}
where $\bm{C}\in\mathbb{R}^{p_1\times K_1}$ and $\bm{R}\in\mathbb{R}^{p_2\times K_2}$ are the shared factor matrices, $\cm{G}\in\mathbb{R}^{K_1\times K_2\times n}$ is the core tensor, with frontal slices $\bm{G}_i=\bm{L}_{1i}\bm{L}_{2i}^\top$, and $\bm{L}_{1i}\in\mathbb{R}^{K_1\times r}$ and $\bm{L}_{2i}\in\mathbb{R}^{K_2\times r}$ are individualized low-rank loading matrices, encoding individual deviations in the row and column projections. See Figure \ref{fig:tensor_parameterization} for an illustration.

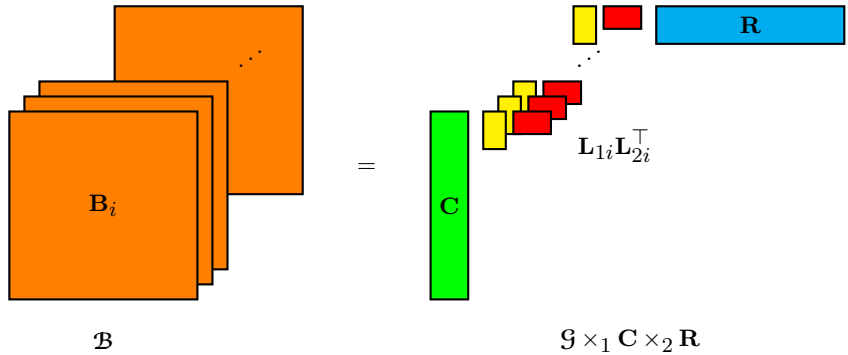
\begin{figure}[!htp]
  \begin{center}
    \begin{tikzpicture}

% 你的绘图代码

\node (rect) at (1.4,1.4) [draw, thick, fill=orange, fill opacity=1,
minimum width=2.5cm,minimum height=2.5cm] {};
\node (rect) at (0.4,0.4) [draw, thick, fill=orange, fill opacity=1,
minimum width=2.5cm,minimum height=2.5cm] {};
\node (rect) at (0.2,0.2) [draw, thick, fill=orange, fill opacity=1,
minimum width=2.5cm,minimum height=2.5cm] {};
\node (rect) at (0,0) [draw, thick, fill=orange, fill opacity=1,
minimum width=2.5cm,minimum height=2.5cm] {$\bm{B}_i$};

\node at (1.95,2.05) {$\iddots$};

\node at (3.5,0.5) {$=$};

\node at (0,-1.8) {$\cm{B}$};

\node (rect) at (4.6,0) [draw, thick, fill=green, fill opacity=1,
minimum width=0.5cm,minimum height=2.5cm] {$\bm{C}$};

\node (rect) at (6.4,2.4) [draw, thick, fill=yellow, fill opacity=1,
minimum width=0.3cm,minimum height=0.5cm] {};

\node (rect) at (6.9,2.5) [draw, thick, fill=red, fill opacity=1,
minimum width=0.5cm,minimum height=0.3cm] {};

\node (rect) at (5.6,1.4) [draw, thick, fill=yellow, fill opacity=1,
minimum width=0.3cm,minimum height=0.5cm] {};

\node (rect) at (6.1,1.5) [draw, thick, fill=red, fill opacity=1,
minimum width=0.5cm,minimum height=0.3cm] {};

\node (rect) at (5.4,1.2) [draw, thick, fill=yellow, fill opacity=1,
minimum width=0.3cm,minimum height=0.5cm] {};

\node (rect) at (5.9,1.3) [draw, thick, fill=red, fill opacity=1,
minimum width=0.5cm,minimum height=0.3cm] {};

\node (rect) at (5.2,1) [draw, thick, fill=yellow, fill opacity=1,
minimum width=0.3cm,minimum height=0.5cm] {};

\node (rect) at (5.7,1.1) [draw, thick, fill=red, fill opacity=1,
minimum width=0.5cm,minimum height=0.3cm] {};

\node at (6.45,2.05) {$\iddots$};

\node at (6.8,0.8) {$\bm{L}_{1i}\bm{L}_{2i}^\top$};

\node (rect) at (8.6,2.4) [draw, thick, fill=cyan, fill opacity=1,
minimum width=2.5cm,minimum height=0.5cm] {$\bm{R}$};

\node at (7,-1.8) {$\cm{G}\times_1\bm{C}\times_2\bm{R}$};

\end{tikzpicture}
    \vspace{-0.2cm}
    \caption{Tensor parameterization and tensor decomposition.}
    \label{fig:tensor_parameterization}
  \end{center}
\end{figure}

\begin{remark}[Rank Properties]
    By the definition of Tucker ranks, we have $\textup{rank}(\cm{B}_{(1)})=K_1$ and $\textup{rank}(\cm{B}_{(2)})=K_2$, where $\cm{B}_{(1)}$ and $\cm{B}_{(2)}$ are the mode-1 and mode-2 matricizations of $\cm{B}$, respectively. Consequently, the mode-3 matricization satisfies $\text{rank}(\cm{B}_{(3)})\leq K_1K_2$. When the number of individuals $n>K_1K_2$, an additional low-rank structure may exist along the third (individual) mode. In such cases, one may consider an extended decomposition of the form
    \begin{equation}\label{eq:tensor_decomp2}
        \cm{B} = \widetilde{\cm{G}}\times_1\bm{C}\times_2\bm{R}\times_3\bm{W},
    \end{equation}
    where $\widetilde{\cm{G}}\in\mathbb{R}^{K_1\times K_2\times K_1K_2}$ and $\bm{W}\in \mathbb{R}^{n\times K_1K_2}$. However, such an extension does not yield additional structural or statistical benefits in our setting, and we therefore focus on the two-mode shared decomposition in \eqref{eq:tensor_decomp} throughout the paper.
\end{remark}

This tensor formulation provides a unified and interpretable framework that explicitly separates:
\begin{itemize}
  \item the \textit{shared low-dimensional subspaces} encoded by $\bm{C}$ and $\bm{R}$,
  \item the \textit{individualized loading structure} encoded by the low-rank core tensor $\cm{G}$, and
  \item the overall \textit{heterogeneity in individual responses} through individual-specific projections.
\end{itemize}
It balances model parsimony with flexibility, enabling joint estimation of common mechanisms and individual deviations.

\subsection{Identification}

A central challenge in tensor-structured or low-rank matrix decomposition is identifiability: the decomposition is inherently non-unique due to rotational and scaling ambiguities. In our setting, the factors $\bm{C}$ and $\bm{R}$ and the core tensor slices $\bm{G}_i=\bm{L}_{1i}\bm{L}_{2i}^\top$ are not uniquely determined without further constraints. We now formalize this non-uniqueness and define an appropriate notion of equivalence and distance for estimation purposes.

The factor matrices $\bm{C}$ and $\bm{R}$ are identifiable only up to orthogonal transformations. Specifically, for any nonsingular $\bm{Q}_1\in\text{GL}(K_1)$ and $\bm{Q}_2\in\text{GL}(K_2)$, the parameter sets
\begin{equation}
  (\bm{C},\bm{R},\{\bm{L}_{1i},\bm{L}_{2i}\}_{i=1}^n)\quad\text{and}\quad(\bm{C}\bm{Q}_1,\bm{R}\bm{Q}_2,\{\bm{Q}_1^{-1}\bm{L}_{1i},\bm{Q}_2^{-1}\bm{L}_{2i}\}_{i=1}^n)
\end{equation}
yields exactly the same tensor $\cm{B}$ and thus lead to identical predictions. Likewise, the individual loading matrices $\bm{L}_{1i}$ and $\bm{L}_{2i}$ are not uniquely identified, since for any invertible $\bm{P}_i\in\text{GL}(r)$, replacing $\bm{L}_{1i}$ with $\bm{L}_{1i}\bm{P}_i$ and $\bm{L}_{2i}$ with $\bm{L}_{2i}\bm{P}_i^{-\top}$ leaves the product $\bm{L}_{1i}\bm{L}_{2i}^\top$ unchanged.

Due to this inherent non-identifiability, we do not impose explicit constraints (e.g., orthogonality or normalization) on the factors in estimation. Instead, we work with \textit{equivalent classes of decompositions} and define a distance metric that is invariant to these ambiguities. This allows us to focus on estimating the underlying structural components while being agnostic to the specific rotation or scaling.

Let $\cm{B}^*$ denote the true coefficient tensor, with corresponding decomposition defined by true parameters $\bm{C}^*\in\mathbb{R}^{p_1\times K_1}$, $\bm{R}^*\in\mathbb{R}^{p_2\times K_2}$, $\bm{L}_{1i}^*\in\mathbb{R}^{K_1\times r}$ and $\bm{L}_{2i}^*\in\mathbb{R}^{K_2\times r}$ for $i=1,\dots,n$. We define the true factor matrices as the leading left singular vectors of $\cm{B}_{(1)}^*$ and $\cm{B}_{(2)}^*$, respectively, with $\bm{C}^{*\top}\bm{C}^*=\bm{I}_{K_1}$ and $\bm{R}^{*\top}\bm{R}^*=\bm{I}_{K_2}$. Then, the true value of core tensor is
\begin{equation}
  \cm{G}^*=\cm{B}^*\times_1\bm{C}^{*\top}\times_2\bm{R}^{*\top}.
\end{equation}
Let $\bm{\Sigma}_C=\text{diag}(\sigma_1(\cm{B}_{(1)}^*),\dots,\sigma_{K_1}(\cm{B}_{(1)}^*))$ and $\bm{\Sigma}_R = \text{diag}(\sigma_1(\cm{B}_{(2)}^*),\dots,\sigma_{K_2}(\cm{B}_{(2)}^*))$ be diagonal matrices containing the singular values of the tensor matricizations along modes 1 and 2, respectively. Let $\ubar{\sigma}=\min(\sigma_{K_1}(\cm{B}_{(1)}^*),\sigma_{K_2}(\cm{B}_{(1)}^*))$ denote the smallest singular value across these modes. For each fronal slice $\bm{G}_i^*$, we perform an SVD: $\bm{G}_i^*=\bm{U}_i\bm{\Sigma}_i\bm{V}_i^\top$, and define the true individual loading factors as
\begin{equation}
  \bm{L}_{1i}^*=\bm{U}_i\bm{\Sigma}_i^{1/2},\quad\bm{L}_{2i}^*=\bm{V}_i\bm{\Sigma}_i^{1/2},
\end{equation}
ensuring consistency with the structure $\bm{G}_i^*=\bm{L}_{1i}^*\bm{L}_{2i}^{*\top}$.

To compare decompositions up to orthogonal and scaling transformations, we define the following distance between an estimated parameter set $\bm{\Theta}=(\bm{C},\bm{R},\{\bm{L}_{1i}^*,\bm{L}_{2i}^*\}_{i=1}^n)$ and the true parameter set $\bm{\Theta}^*$:
\begin{equation}
    \begin{split}
        \text{dist}(\bm{\Theta},\bm{\Theta}^*)^2 = & \inf_{\substack{\bm{Q}_1\in\text{GL}(K_1),\\\bm{Q}_2\in\text{GL}(K_2),\\\bm{P}_1,\dots,\bm{P}_n\in\text{GL}(r)}}\Bigg\{ \|(\bm{C}\bm{Q}_1-\bm{C}^*)\bm{\Sigma}_C\|_\text{F}^2 + \|(\bm{R}\bm{Q}_2-\bm{R}^*)\bm{\Sigma}_R\|_\text{F}^2 \\
        & + \sum_{i=1}^n\|(\bm{Q}_1^{-1}\bm{L}_{1i}\bm{P}_i - \bm{L}_{1i}^*)\bm{\Sigma}_i^{1/2}\|_\text{F}^2 + \sum_{i=1}^n\|(\bm{Q}_2^{-1}\bm{L}_{2i}\bm{P}_i^{-\top} - \bm{L}_{2i}^*)\bm{\Sigma}_i^{1/2}\|_\text{F}^2 \Bigg\}.
    \end{split}
\end{equation}
where $\bm{Q}_1$, $\bm{Q}_2$, and $\bm{P}_i$'s are invertible matrices in the general linear groups that account for the inherent non-uniqueness of the tensor decomposition. This distance metric explicitly removes the effects of transformations $(\bm{Q}_1,\bm{Q}_2,\bm{P}_i)$ and measures the discrepancy between essential structural components of the estimated and true decompositions. It is invariant to rotations and scaling, and serves as the foundation for establishing estimation guarantees.

We now establish a key property of the proposed distance metric: it tightly controls the Frobenius norm error of the underlying coefficient tensor $\cm{B}$, thereby linking the geometric notion of identifiability to the statistical target of estimation.

\begin{proposition}[Distance-Error Equivalence]\label{prop:1}
    For $\bm{\Theta}$ and $\bm{\Theta^*}$, denote the corresponding parameter tensors as $\cm{B}$ and $\cm{B}^*$. If $\textup{dist}(\bm{\Theta},\bm{\Theta}^*)\leq \ubar{\sigma}$,
    \begin{equation}
        C_1\|\cm{B}-\cm{B}^*\|_\textup{F}^2 \leq \textup{dist}(\bm{\Theta},\bm{\Theta}^*)^2 \leq C_2\|\cm{B}-\cm{B}^*\|_\textup{F}^2,
    \end{equation}
    where $C_1$ and $C_2$ are positive constants.
\end{proposition}

\section{Estimation Methodology}\label{sec:estimation_methodology}

\subsection{Scaled Gradient Descent Algorithm}

We propose a scaled gradient descent algorithm tailored to the tensor-structured homogeneity pursuit model introduced in Section \ref{sec:model_framework}. The objective is to estimate the full parameter set $\bm{\Theta}=(\bm{C},\bm{R},\{\bm{L}_{1i},\bm{L}_{2i}\}_{i=1}^n)$, which comprises the shared low-rank factor matrices $\bm{C}\in\mathbb{R}^{p_1\times K_1}$ and $\bm{R}\in\mathbb{R}^{p_2\times K_2}$, and the individualized loading matrices $\bm{L}_{1i}\in\mathbb{R}^{K_1\times r}$ and $\bm{L}_{2i}\in\mathbb{R}^{K_2\times r}$ for $i=1,\dots,n$.

The estimation is based on minimizing the empirical negative log-likelihood loss
\begin{equation}
    \mathcal{L}_n(\bm{\Theta}) = \sum_{i=1}^n\sum_{j=1}^{m_i}\mathcal{L}(\bm{C}\bm{L}_{1i}\bm{L}_{2i}^\top\bm{R}^\top;\bm{X}_{ij},Y_{ij}),\quad\text{where }\bm{B}_i=\bm{C}\bm{L}_{1i}\bm{L}_{2i}^\top\bm{R}^\top.
\end{equation}
Here, $\mathcal{L}(\cdot)$ is a differential loss function associated with the generalized linear model framework introduced in Section \ref{sec:model_framework}.

From the Tucker decomposition in \eqref{eq:tensor_decomp}, we define the composite matrices 
\begin{equation}
  \widetilde{\bm{C}}=[\bm{L}_{11}\bm{L}_{21}^\top\bm{R}^\top,\dots,\bm{L}_{1n}\bm{L}_{2n}^\top\bm{R}^\top],\quad \widetilde{\bm{R}}=[\bm{L}_{21}\bm{L}_{11}^\top\bm{C}^\top,\dots,\bm{L}_{2n}\bm{L}_{1n}^\top\bm{C}^\top],
\end{equation}
so that the mode-1 and mode-2 matricizations of $\cm{B}$ satisfy 
\begin{equation}
  \cm{B}_{(1)}=\bm{C}\widetilde{\bm{C}}^\top,\quad\cm{B}_{(2)}=\bm{R}\widetilde{\bm{R}}^\top.
\end{equation}
The partial gradients of the empirical loss with respect to the parameters are given by
\begin{equation}
    \begin{split}
        \nabla_{\bm{C}}\mathcal{L}_n(\bm{\Theta}) & = \sum_{i=1}^n\sum_{j=1}^{m_i}\left[\bm{e}_i(n)^\top\otimes\nabla\mathcal{L}(\bm{B}_i;\bm{X}_{ij},Y_{ij})\right]\widetilde{\bm{C}},\\
        \nabla_{\bm{R}}\mathcal{L}_n(\bm{\Theta}) & = \sum_{i=1}^n\sum_{j=1}^{m_i}\left[\bm{e}_i(n)^\top\otimes\nabla\mathcal{L}(\bm{B}_i;\bm{X}_{ij},Y_{ij})^\top\right]\widetilde{\bm{R}},\\
        \nabla_{\bm{L}_{1i}}\mathcal{L}_n(\bm{\Theta}) & = \sum_{j=1}^{m_i}\bm{C}^\top\nabla\mathcal{L}(\bm{B}_i;\bm{X}_{ij},Y_{ij})\bm{R}\bm{L}_{2i},\\
        \nabla_{\bm{L}_{2i}}\mathcal{L}_n(\bm{\Theta}) & = \sum_{j=1}^{m_i}\bm{R}^\top\nabla\mathcal{L}(\bm{B}_i;\bm{X}_{ij},Y_{ij})^\top\bm{C}\bm{L}_{1i},
    \end{split}
\end{equation}
where $\bm{e}_i(n)$ is the $i$-th canonical basis vector in $\mathbb{R}^{n}$.

A major challenge in optimizing the above objective lies in the \textit{inherent scaling imbalance} between the shared parameters ($\bm{C}$, $\bm{R}$), which operate in high-dimensional spaces $\mathbb{R}^{p_1\times K_1}$ and $\mathbb{R}^{p_1\times K_1}$, and the individual parameters $(\bm{L}_{1i},\bm{L}_{2i})$, which lie in much lower-dimensional spaces $\mathbb{R}^{K_1\times r}$ and $\mathbb{R}^{K_2\times r}$. This is further complicated by the non-identifiability of the decomposition under orthogonal transformations, and the aggregation scope of gradients, where shared gradients pool information across all individuals, while individual gradients depend only on local data.

To address these issues, we introduce a scaled gradient descent approach that adaptively rescales the gradient for each parameter class. Specifically, for $\bm{C}$ and $\bm{R}$, we precondition the gradients by $(\widetilde{\bm{C}}^\top\widetilde{\bm{C}})^{-1}$ and $(\widetilde{\bm{R}}^\top\widetilde{\bm{R}})^{-1}$, respectively. For each $\bm{L}_{1i}$ (or $\bm{L}_{2i}$), we scale by $(\bm{C}^\top\bm{C})^{-1}$ and $(\bm{L}_{2i}^\top\bm{R}^\top\bm{R}\bm{L}_{2i})^{-1}$ (or their counterparts for $\bm{L}_{2i}$), ensuring that updates are numerically stable and dimensionally balanced. This leads to the Scaled Gradient Descent Algorithm (Algorithm \ref{alg:1}), detailed below.

\begin{algorithm}[!htp]
\caption{Scaled Gradient Descent for Homogeneity Pursuit}
\label{alg:1}
\begin{flushleft}
\linespread{1.75}\selectfont
\textbf{Input}: Initial values $\bm{C}^{(0)}$, $\bm{R}^{(0)}$, $\bm{L}_{1i}^{(0)}$, $\bm{L}_{2i}^{(0)}$ for $i=1,\dots,n$, step size $\eta$, number of iterations $T$ \\
\textbf{For} $t=0,\dots,T-1$\\
\hspace*{1cm}$\bm{C}^{(t+1)}\leftarrow \bm{C}^{(t)} - \eta\nabla_{\bm{C}}\mathcal{L}_n(\bm{\Theta}^{(t)})\left(\widetilde{\bm{C}}^{(t)\top}\widetilde{\bm{C}}^{(t)}\right)^{-1}$\\
\hspace*{1cm}$\bm{R}^{(t+1)}\leftarrow \bm{R}^{(t)} - \eta\nabla_{\bm{R}}\mathcal{L}_n(\bm{\Theta}^{(t)})\left(\widetilde{\bm{R}}^{(t)\top}\widetilde{\bm{R}}^{(t)}\right)^{-1}$\\
\hspace*{1cm}\textbf{For} $i=1,\dots,n$ (in parallel) \textbf{do}\\
\hspace*{2cm}$\bm{L}_{1i}^{(t+1)}\leftarrow\bm{L}_{1i}^{(t)}-\eta\left(\bm{C}^{(t)\top}\bm{C}^{(t)}\right)^{-1}\nabla_{\bm{L}_{1i}}\mathcal{L}_n(\bm{\Theta}^{(t)})\left(\bm{L}_{2i}^{(t)\top}\bm{R}^{(t)\top}\bm{R}^{(t)}\bm{L}_{2i}^{(t)}\right)^{-1}$\\
\hspace*{2cm}$\bm{L}_{2i}^{(t+1)}\leftarrow\bm{L}_{2i}^{(t)}-\eta\left(\bm{R}^{(t)\top}\bm{R}^{(t)}\right)^{-1}\nabla_{\bm{L}_{2i}}\mathcal{L}_n(\bm{\Theta}^{(t)})\left(\bm{L}_{1i}^{(t)\top}\bm{C}^{(t)\top}\bm{C}^{(t)}\bm{L}_{1i}^{(t)}\right)^{-1}$\\
\hspace*{1cm}\textbf{End for}\\
\textbf{End for}\\
\textbf{Output}: $\widehat{\bm{\Theta}}=\left(\bm{C}^{(T)},\bm{R}^{(T)},\bm{L}_{11}^{(T)},\bm{L}_{21}^{(T)},\dots,\bm{L}_{1n}^{(T)},\bm{L}_{2n}^{(T)}\right)$
\end{flushleft}\vspace{-0.1cm}
\end{algorithm}

Notably, the parallel updates for $\bm{L}_{1i}$ and $\bm{L}_{2i}$ enhance computational scalability, and all scaling operations involve only inverses of low-dimensional matrices, ensuring efficiency even in high-dimensional settings.

\subsection{Algorithm Computational Guarantees}

We now establish the general theoretical guarantees for the scaled gradient descent algorithm. Let $\bm{B}_i^* = \bm{C}^* \bm{L}_{1i}^* \bm{L}_{2i}^{*\top} \bm{R}^{*\top}$ denote the true coefficient matrices, and let $\bm{\Theta}^*=(\bm{C}^*,\bm{R}^*,\{\bm{L}_{1i}^*,\bm{L}_{2i}^*\}_{i=1}^n)$ be the true parameter set.

\begin{assumption}[Restricted Correlated Gradient Condition]\label{def:RCG}
    For each individual $i$, the loss satisfies the Restricted Correlated Gradient (RCG) condition: for $\bm{B}_i=\bm{C}\bm{L}_{1i}\bm{L}_{2i}^\top\bm{R}^\top$,
    \begin{equation}
        \left\langle\mathbb{E}[\nabla\mathcal{L}(\bm{B}_i;\bm{X}_{ij},Y_{ij})],\bm{B}_i-\bm{B}_i^*\right\rangle \geq \frac{\alpha_i}{2}\|\bm{B}_i-\bm{B}_i^*\|_\textup{F}^2 + \frac{1}{2\beta_i}\|\mathbb{E}\nabla\mathcal{L}(\bm{B}_i;\bm{X}_{ij},Y_{ij})\|_\textup{F}^2,
    \end{equation}
    where $\alpha_i,\beta_i>0$ are individual-specific parameters satisfy $0<\alpha_i\leq\beta_i$.
\end{assumption}
The RCG condition in Assumption \ref{def:RCG} implies that, for each individual $i$, the expectation of the gradient $\nabla\mathcal{L}(\bm{B}_i;\bm{X}_{ij},Y_{ij})$ is postively correlated with the optimal descent direction $\bm{B}_i-\bm{B}_i^*$. This condition is implied by the restricted strong convexity with parameter $\alpha_i$ and restricted strong smoothness with parameter $\beta_i$ of the population risk $\mathbb{E}[\mathcal{L}(\bm{B}_i;\bm{X}_{ij},Y_{ij})]$, which are widely adopted in non-convex algorithm convergence analysis. In addition, the distributions of $(\bm{X}_{ij},Y_{ij})$ are allowed to be heterogeneous across individuals, with various $\alpha_i$ and $\beta_i$.

\begin{assumption}[Gradient Stability]
\label{def:stability}
For any $\bm{\Theta}$, the gradients are stable in the sense that there exist some constants $\phi>0$, $\xi_C,\xi_R,\xi_{1i},\xi_{2i}>0$ such that
\begin{equation}
    \begin{split}
        \left\|\{\nabla_{\bm{C}}\mathcal{L}_n(\bm{\Theta}) - \mathbb{E}[\nabla_{\bm{C}}\mathcal{L}_n(\bm{\Theta})]\}(\widetilde{\bm{C}}^\top\widetilde{\bm{C}})^{-1/2}\right\|_\textup{F}^2 & \leq \phi \|\cm{B}-\cm{B}^*\|_\textup{F}^2 + \xi_C,\\
        \left\|\{\nabla_{\bm{R}}\mathcal{L}_n(\bm{\Theta}) - \mathbb{E}[\nabla_{\bm{R}}\mathcal{L}_n(\bm{\Theta})]\}(\widetilde{\bm{R}}^\top\widetilde{\bm{R}})^{-1/2}\right\|_\textup{F}^2 & \leq \phi \|\cm{B}-\cm{B}^*\|_\textup{F}^2 + \xi_R,\\
        \left\|(\bm{C}^\top\bm{C})^{-1/2}\{\nabla_{\bm{L}_{1i}}\mathcal{L}_n(\bm{\Theta})-\mathbb{E}[\nabla_{\bm{L}_{1i}}\mathcal{L}_n(\bm{\Theta})]\}(\bm{L}_{2i}^\top\bm{R}^\top\bm{R}\bm{L}_{2i})^{-1/2} \right\|_\textup{F}^2 & \leq \phi\|\bm{B}_i-\bm{B}_i^*\|_\textup{F}^2 + \xi_{1i},\\
        \left\|(\bm{R}^\top\bm{R})^{-1/2}\{\nabla_{\bm{L}_{2i}}\mathcal{L}_n(\bm{\Theta})-\mathbb{E}[\nabla_{\bm{L}_{2i}}\mathcal{L}_n(\bm{\Theta})]\}(\bm{L}_{1i}^\top\bm{C}^\top\bm{C}\bm{L}_{1i})^{-1/2} \right\|_\textup{F}^2 & \leq \phi\|\bm{B}_i-\bm{B}_i^*\|_\textup{F}^2 + \xi_{2i}.
    \end{split}
\end{equation}
\end{assumption}
The parameters $\phi$ controls the performance of all partial gradients for any $\bm{\Theta}$, while the parameters $\xi_C$, $\xi_R$, $\xi_{1i}$'s, and $\xi_{2i}$'s represent the estimation accuracy of the partial sample gradients around their expectations. Similar definitions for the gradient stability have been considered for robust tensor estimation in \citet{zhang2025robust}. Under mild conditions, the gradient stability will be further verified for specific statistical models in Section \ref{sec:statistical_models}.

\begin{theorem}[Local Linear Convergence]
    \label{thm:1}
    Under Assumptions \ref{def:RCG}--\ref{def:stability}, if the step size satisfies $\eta\asymp\ubar{\alpha}\bar{\beta}^{-1}$, $\phi\lesssim\ubar{\alpha}^2m^2$, and the initialization satisfies $\textup{dist}(\bm{\Theta}^{(0)},\bm{\Theta}^*) \lesssim\ubar{\alpha}^{1/2}\bar{\beta}^{-1/2}\ubar{\sigma}$, then for all $t=1,2,\dots,T$,
    \begin{equation}
      \textup{dist}(\bm{\Theta}^{(t)},\bm{\Theta}^{*})^2 \leq (1-C\ubar{\alpha}/\bar{\beta})^t \cdot\textup{dist}(\bm{\Theta}^{(0)},\bm{\Theta}^{*})^2 + C\ubar{\alpha}^{-2}m^{-2}\xi,
    \end{equation}
    and the tensor estimation error satisfies
    \begin{equation}
      \|\cm{B}^{(t)}-\cm{B}^*\|_\textup{F}^2 \leq C(1-C
        \ubar{\alpha}/\bar{\beta})^t \cdot\textup{dist}(\bm{\Theta}^{(0)},\bm{\Theta}^*)^2 + C\ubar{\alpha}^{-2}m^{-2}\xi,
    \end{equation}
    where $m=n^{-1}\sum_{i=1}^nm_i$, $\ubar{\alpha}=m^{-1}\min_{1\leq i\leq n}m_i\alpha_i$, $\bar{\beta}=m^{-1}\max_{1\leq i\leq n}m_i\beta_i$, and $\xi=\xi_C+\xi_R+\sum_{i=1}^n(\xi_{1i}+\xi_{2i})$.
\end{theorem}

Theorem \ref{thm:1} states the local convergence of the proposed scaled gradient descent iterates given that the initial value $\bm{\Theta}^{(0)}$ lies in a local region around the ground truth. In both upper bounds, the first term captures the exponential decay of the optimization error within the basin of attraction around the true parameter. The second term represents the statistical error, which diminishes as the sample size $m$ increases or as the gradient noise diminishes. Notably, the convergence rate is independent of the condition number of the underlying parameter matrices, i.e., $\max(\sigma_1(\bm{\Sigma}_C),\sigma_1(\bm{\Sigma}_R))/\min(\sigma_{K_1}(\bm{\Sigma}_C),\sigma_{K_2}(\bm{\Sigma}_R))$, a favorable property in high-dimensional or ill-conditioned regimes.

Building upon Theorem \ref{thm:1}, we analyze the convergence rates of the individual parameter components after applying an alignment transformation that maps the estimated parameters to their closest transformation in the equivalence class of the true parameters. Formally, let
\begin{equation}
    \begin{split}
        (\widehat{\bm{Q}}_1,\widehat{\bm{Q}}_2,\widehat{\bm{P}}_1,&\dots,\widehat{\bm{P}}_n) = \underset{{\substack{\bm{Q}_1\in\text{GL}(K_1),\\\bm{Q}_2\in\text{GL}(K_2),\\\bm{P}_1,\dots,\bm{P}_n\in\text{GL}(r)}}}{\arg\inf}\Bigg\{ \|(\widehat{\bm{C}}\bm{Q}_1-\bm{C}^*)\bm{\Sigma}_C\|_\text{F}^2 + \|(\widehat{\bm{R}}\bm{Q}_2-\bm{R}^*)\bm{\Sigma}_R\|_\text{F}^2 \\
        & + \sum_{i=1}^n\|(\bm{Q}_1^{-1}\widehat{\bm{L}}_{1i}\bm{P}_i - \bm{L}_{1i}^*)\bm{\Sigma}_i^{1/2}\|_\text{F}^2 + \sum_{i=1}^n\|(\bm{Q}_2^{-1}\widehat{\bm{L}}_{2i}\bm{P}_i^{-\top} - \bm{L}_{2i}^*)\bm{\Sigma}_i^{1/2}\|_\text{F}^2 \Bigg\}.
    \end{split}
\end{equation}

\begin{corollary}[Convergence Rates of Estimated Parameters]
    \label{cor:1}
    Under the conditions of Theorem \ref{thm:1}, after a sufficient number of iterations
    \begin{equation}
        T \gtrsim \frac{\log(C\ubar{\alpha}^{-2}m^{-2}\xi/\textup{dist}(\bm{\Theta}^{(0)},\bm{\Theta}^*)^2)}{\log(1-C\ubar{\alpha}\bar{\beta}^{-1})},
    \end{equation}
    the following convergence rates hold
    \begin{equation}
        \begin{split}
            \textup{dist}(\widehat{\bm{\Theta}},\bm{\Theta}^*)^2 & \lesssim \frac{\xi}{\ubar{\alpha}^2m^2},\\
            \frac{1}{n}\sum_{i=1}^n\|\widehat{\bm{B}}_i-\bm{B}_i^*\|_\textup{F}^2 & \lesssim \frac{\xi}{\ubar{\alpha}^2m^2n},\\
            \|\widehat{\bm{C}}\widehat{\bm{Q}}_1 - \bm{C}^*\|^2_\textup{F} & \lesssim \frac{\xi}{\sigma_{K_1}^2(\bm{\Sigma}_{C})\ubar{\alpha}^{2}m^{2}},\\
            \|\widehat{\bm{R}}\widehat{\bm{Q}}_2 - \bm{R}^*\|_\textup{F}^2 & \lesssim \frac{\xi}{\sigma_{K_2}^2(\bm{\Sigma}_{R})\ubar{\alpha}^{2}m^{2}},
        \end{split}
    \end{equation}
    and for the individual loading factors,
    \begin{equation}
      \frac{1}{n}\sum_{i=1}^n\|\widehat{\bm{Q}}_1^{-1}\widehat{\bm{L}}_{1i}\widehat{\bm{P}}_i - \bm{L}_{1i}^* \|_\textup{F}^2 + \frac{1}{n}\sum_{i=1}^n \|\widehat{\bm{Q}}_2^{-1} \widehat{\bm{L}}_{1i}\widehat{\bm{P}}_i^{-\top} - \bm{L}_{2i}^* \|_\textup{F}^2 \lesssim \frac{\xi}{\min_{1\leq i\leq n}\sigma_{r}(\bm{\Sigma}_i)\ubar{\alpha}^{2}m^{2}n}.
    \end{equation}
\end{corollary}

Corollary \ref{cor:1} presents the convergence rates of the shared and individualized parameters with the optimal alignment transformations. These bounds demonstrate that the average error $n^{-1}\sum_{i=1}^n\|\widehat{\bm{B}}_i-\bm{B}_i^*\|_\text{F}^2$ achieves a faster rate due to the information pooling across individuals. Meanwhile, the shared components $\bm{C}$ and $\bm{R}$ converge at rates governed by the smallest singular values of the population-level signal matrices, reflecting the distribution of heterogeneity across the shared subspaces. The individual loading matrices $\bm{L}_{1i}$ and $\bm{L}_{2i}$ converge at rates that depend critically on the minimal individual signal level $\min_i\sigma_r(\bm{\Sigma}_i)$. When this signal is bounded away from zero and the shared structure is sufficiently strong (e.g., $\sigma_{K_1}(\bm{\Sigma}_C),\sigma_{K_2}(\bm{\Sigma}_R)\gtrsim\sqrt{n}$), all components are estimated at statistically efficient rates.

\subsection{Loss of Ignoring Homogeneity}

To assess the value of the proposed homogeneity pursuit framework, we compare it with a fully heterogeneous baseline model if the latent shared structure is ignored. In this setting, each individual's coefficient matrix is modeled as $\bm{B}_i=\bm{C}_i\bm{R}_i^\top$. The corresponding loss function is
\begin{equation}
    \overline{\mathcal{L}}_i(\bm{C}_{i},\bm{R}_i) = \sum_{j=1}^{m_i}\mathcal{L}(\bm{C}_i\bm{R}_i^\top;\bm{X}_{ij},Y_{ij}),
\end{equation}
and the scaled gradient descent algorithm is modified accordingly, as summarized in Algorithm \ref{alg:2}.

\begin{algorithm}[!htp]
\caption{Scaled Gradient Descent Algorithm without Homogeneity Pursuit}
\label{alg:2}
\begin{flushleft}
\linespread{1.75}\selectfont
\textbf{Input}: initial values $\bm{C}_i^{(0)}$, $\bm{R}_i^{(0)}$ for $i=1,\dots,n$, step sizes $\eta_i$, no. of iterations $T$\\
\textbf{For} $i=1,\dots,n$ (in parallel)\\
\hspace*{1cm}\textbf{For} $t=0,\dots,T-1$\\
\hspace*{2cm}$\bm{C}_{i}^{(t+1)}\leftarrow\bm{C}_{i}^{(t)}-\eta_i\cdot\nabla_{\bm{C}_{i}}\overline{\mathcal{L}}_i(\bm{C}_{i}^{(t)},\bm{R}_i^{(t)})\left(\bm{R}_i^{(t)\top}\bm{R}_i^{(t)}\right)^{-1}$\\
\hspace*{2cm}$\bm{R}_{i}^{(t+1)}\leftarrow\bm{R}_{i}^{(t)}-\eta_i\cdot\nabla_{\bm{R}_{i}}\overline{\mathcal{L}}_i(\bm{C}_{i}^{(t)},\bm{R}_i^{(t)})\left(\bm{C}_i^{(t)\top}\bm{C}_i^{(t)}\right)^{-1}$\\
\hspace*{1cm}\textbf{End for}\\
\textbf{End for}\\
\textbf{Output}: $(\widecheck{\bm{C}}_1,\widecheck{\bm{R}}_1,\dots,\widecheck{\bm{C}}_n,\widecheck{\bm{R}}_n) = (\bm{C}_1^{(T)},\bm{R}_1^{(T)},\dots,\bm{C}_n^{(T)},\bm{R}_n^{(T)})$
\end{flushleft}\vspace{-0.1cm}
\end{algorithm}

Similarly to Assumption \ref{def:stability}, we consider the analogous gradient stability conditions, which leads to the convergence rates without homogeneity pursuit.

\begin{assumption}[Gradient Stability for Individualized Components]\label{def:stability2}
  For each $i$, there exist constants $\psi>0$, $\nu_{1i},\nu_{2i}>0$ such that
  \begin{equation}\label{eq:stability3}
    \begin{split}
        & \left\|\{\nabla_{\bm{C}_i}\overline{\mathcal{L}}_i(\bm{C}_i,\bm{R}_i) - \mathbb{E}[\nabla_{\bm{C}_i}\overline{\mathcal{L}}_i(\bm{C}_i,\bm{R}_i)]\}(\bm{R}_i^\top\bm{R}_i)^{-1/2}\right\|_\text{F}^2 \leq \psi\|\bm{B}_i-\bm{B}_i^*\|_\text{F}^2 + \nu_{1i},\\
        & \left\|\{\nabla_{\bm{R}_i}\overline{\mathcal{L}}_i(\bm{C}_i,\bm{R}_i) - \mathbb{E}[\nabla_{\bm{R}_i}\overline{\mathcal{L}}_i(\bm{C}_i,\bm{R}_i)]\}(\bm{C}_i^\top\bm{C}_i)^{-1/2}\right\|_\text{F}^2 \leq \psi\|\bm{B}_i-\bm{B}_i^*\|_\text{F}^2 + \nu_{2i}.
    \end{split}
  \end{equation}
\end{assumption}

\begin{corollary}[Convergence Rates without Homogeneity Pursuit]
    \label{cor:2}
    Under Assumptions \ref{def:RCG} and \ref{def:stability2}, if $\eta_i\asymp\alpha_i\beta_i^{-1}$, $\psi\lesssim\min_{1\leq i\leq n}\alpha_i^2m_i^2$, and the initial estimate error satisfies
    \begin{equation}
        \inf_{\bm{P}_1,\dots,\bm{P}_n\in\textup{GL}(r)}\left\{\left\|(\bm{C}_i^{(0)}\bm{P}_i-\bm{C}_i^*)\bm{\Sigma}_i^{1/2}\right\|_\textup{F}^2+\left\|(\bm{R}_i^{(0)}\bm{P}_i^{-\top}-\bm{R}_i^*)\bm{\Sigma}_i^{1/2}\right\|_\textup{F}^2\right\}\lesssim \alpha_i\beta_i^{-1}\sigma_r^2(\bm{\Sigma}_i),
    \end{equation}
    then given a sufficiently large number of iterations, for all $1\leq i\leq n$,
    \begin{equation}
        \begin{split}
            \|\widecheck{\bm{B}}_i - \bm{B}_i^*\|_\textup{F}^2  & \lesssim \frac{\nu_{1i}+\nu_{2i}}{\alpha_i^2 m_i^2},\\
            \|\widecheck{\bm{C}}_i\widecheck{\bm{P}}_i-\bm{C}_i^*\|_\textup{F}^2 + \|\widecheck{\bm{R}}_i\widecheck{\bm{P}}_i^{-\top}-\bm{R}_i^*\|_\textup{F}^2 & \lesssim \frac{\nu_{1i}+\nu_{2i}}{\sigma_r(\bm{\Sigma}_i)\alpha_i^2m_i^2},
        \end{split}
    \end{equation}
    where $\widecheck{\bm{B}}_i=\widecheck{\bm{C}}_i\widecheck{\bm{R}}_i^\top$ and $\widecheck{\bm{P}}_i\in\textup{GL}(r)$ is the optimal alignment matrix.
\end{corollary}

As shown in Corollary \ref{cor:2}, the estimation error of each $\widecheck{\bm{B}}_i$ scales as $(\nu_{1i}+\nu_{2i})/(\alpha_i^2m_i^2)$, where ignoring shared structure prevents the $O(n^{-1})$ error reduction achievable under homogeneity pursuit. Therefore, the homogeneity pursuit framework leads to significant statistical gains, both in terms of individual parameter estimation and recovery of shared structures, by leveraging the common low-rank subspaces across individuals. These advantages are formally reflected in the improved convergence rates and are further supported by the algorithm computational and scalability benefits.

\subsection{Initialization, Rank Selection, and Implementation}\label{sec:3.4}

The performance of the proposed Scaled Gradient Descent algorithms depends critically on the initialization quality, as well as the correct specification of the low-rank parameters $r$, $K_1$, and $K_2$. In this subsection, we describe practical strategies for initializing the shared subspaces $\bm{C}$ and $\bm{R}$, as well as methods for selecting the ranks $r$, $K_1$, and $K_2$.

When the ranks $r$, $K_1$, and $K_2$ are known or have been pre-selected, we describe a data-driven approach to initialize the shared column and row subspaces, represented by $\bm{C}\in\mathbb{R}^{p_1\times K_1}$ and $\bm{R}\in\mathbb{R}^{p_2\times K_2}$, respectively. Given preliminary estimates $\{\widecheck{\bm{C}}_i,\widecheck{\bm{R}}_i\}_{i=1}^n$ (e.g., from the heterogeneous baseline Algorithm \ref{alg:2}), we construct the aggregate matrices
\begin{equation}\label{eq:M_C_R}
    \bm{M}_C = \sum_{i=1}^n\widecheck{\bm{C}}_i\widecheck{\bm{C}}_i^\top\quad\bm{M}_R = \sum_{i=1}^n\widecheck{\bm{R}}_i\widecheck{\bm{R}}_i^\top,
\end{equation}
and take their top $K_1$ and $K_2$ eigenvectors (corresponding to the largest eigenvalues) to form the initial estimates $\bm{C}^{(0)}$ and $\bm{R}^{(0)}$. The individual loading matrices are initialized as $\bm{L}_{1i}^{(0)}=\bm{C}^{(0)\top}\widecheck{\bm{C}}_i$ and $\bm{L}_{2i}^{(0)}=\bm{R}^{(0)\top}\widecheck{\bm{R}}_i$.

To select $r$, one may pre-determine a rank upper bound $\bar{r}$ satisfying $r<\bar{r}\ll \bar{p}$, and obtain the initial estimators $\widecheck{\bm{B}}_i(\bar{r})$ based on $\bar{r}$. Then, we can apply the cumulative ridge-type ratio estimator to select $r$
\begin{equation}\label{eq:ridge_1}
    \widehat{r}={\arg\max}_{1\leq r\leq \bar{r}-1}\frac{\sum_{i=1}^n\sigma_{r}(\widecheck{\bm{B}}_i(\bar{r})) + \delta_1(n,m,\bar{p})}{\sum_{i=1}^n\sigma_{r+1}(\widecheck{\bm{B}}_i(\bar{r})) + \delta_1(n,m,\bar{p})},
\end{equation}
where $\delta_1(n,m,\bar{p})$ is the ridge parameter depending on $n$, $m$ and $\bar{p}$.

In real practice, given $r$, it is normal to assume both $K_1$ and $K_2$ are small, say not greater than $Cr$ for some constant $C$, for better interpretability.
To select $K_1$ and $K_2$, we can also apply another ridge-type ratio estimator to the eigenvalues of $\bm{M}_C$ and $\bm{M}_R$
\begin{equation}\label{eq:ridge_21}
  \widehat{K}_1 = {\arg\max}_{1\leq k\leq Cr}\frac{\lambda_{k}(\bm{M}_C) + \delta_2(n,m,p_1)}{\lambda_{r+1}(\bm{M}_C) + \delta_2(n,m,p_1)},
\end{equation}
and
\begin{equation}\label{eq:ridge_22}
  \widehat{K}_2 = {\arg\max}_{1\leq k\leq Cr}\frac{\lambda_{k}(\bm{M}_R) + \delta_2(n,m,p_2)}{\lambda_{r+1}(\bm{M}_R) + \delta_2(n,m,p_2)},
\end{equation}
where $\delta_2(n,m,p)$ is another ridge parameter. The theoretical justification of the above initialization and rank determination method will be given for specific models in Section \ref{sec:statistical_models}.

\section{Sparsity and Scaled Hard Thresholding}\label{sec:sparsity}

\subsection{Sparsity in Homogeneous Subspaces}

The homogeneity pursuit framework decomposes the individual regression coefficient matrices into a product of shared components ($\bm{C}$, $\bm{R}$) and individualized components ($\bm{L}_{1,i}$, $\bm{L}_{2,i}$). While the individualized components lie in low-dimensional spaces and are typically estimated with high accuracy, the shared components $\bm{C}\in\mathbb{R}^{p_1\times K_1}$ and $\bm{R}\in\mathbb{R}^{p_2\times K_2}$ remain high-dimensional when $p_1$ or $p_2$ is large.

To enhance interpretability, estimation efficiency, and computational scalability, we impose row-wise sparsity on the shared components $\bm{C}$ and $\bm{R}$. Specifically, we assume that only a subset of rows in $\bm{C}$ and $\bm{R}$ are nonzero, thereby identifying a small number of informative covariate dimensions that contribute to the response across all individuals.

Formally, let $S_1\subseteq\{1,\dots,p_1\}$ and $S_2\subseteq\{1,\dots,p_2\}$ denote the index sets of nonzero rows in $\bm{C}$ and $\bm{R}$, respectively. The sparsity structure satisfies two key properties. First, the sparsity pattern is invariant under transformations of the form $(\bm{C},\bm{R})\to(\bm{C}\bm{Q}_1,\bm{R}\bm{Q}_2)$ where $\bm{Q}_1\in\text{GL}(K_1)$ and $\bm{Q}_2\in\text{GL}(K_2)$. This ensures that the sparsity is well-defined on the underlying latent subspaces. Second, the active rows in $\bm{C}$ and $\bm{R}$ correspond to a subset of informative covariate dimensions that are predictive of the response $y_{ij}$. Specifically, the inner product
    \begin{equation}
         \langle\bm{X}_{ij},\bm{C}\bm{L}_{1i}\bm{L}_{2i}^\top\bm{R}^\top\rangle = \langle\bm{C}^\top_{S_1}(\bm{X}_{ij})_{[S_1,S_2]}\bm{R}_{S_2},\bm{L}_{1i}\bm{L}_{2i}^\top\rangle,
     \end{equation} 
where $(\bm{X}_{ij})_{[S_1,S_2]}$ denotes the submatrix of $\bm{X}_{ij}$ restricted to rows in $S_1$ and columns in $S_2$, and $\bm{C}_{S_1}$ and $\bm{R}_{S_2}$ are the corresponding submatrices of $\bm{C}$ and $\bm{R}$ restricted to rows in $S_1$ and $S_2$, respectively. This induces a reduced effective dimensionality in the shared components, improving both statistical and computational performance.

\subsection{Scaled Hard Thresholding for Structured Sparsity}

To induce sparsity in $\bm{C}$ and $\bm{R}$, a natural approach is to apply hard thresholding to the rows of these matrices during the optimization process. Specifically, for a matrix $\bm{M}\in\mathbb{R}^{d\times K}$, we define the hard thresholding operator $\text{HT}(\bm{M},s):\mathbb{R}^{d\times K}\times\mathbb{R}^+\to\mathbb{R}^{d\times K}$ as the operator that retains the $s$ rows of $\bm{M}$ with the largest Euclidean norms and sets the remaining rows to zero.

However, directly applying hard thresholding to $\bm{C}$ and $\bm{R}$ is problematic due to the non-uniqueness of the tensor decomposition. In particular, for equivalent decompositions $\bm{\Theta}=(\bm{C},\bm{R},\{\bm{L}_{1i},\bm{L}_{2i}\}_{i=1}^n)$ and $\bm{\Theta}'=(\bm{C}\bm{Q}_1,\bm{R}\bm{Q}_2,\{\bm{L}_{1i},\bm{L}_{2i}\}_{i=1}^n)$, the set of rows with largest norms may differ, leading to inconsistent sparsity patterns across equivalent solutions.

To address this, we apply hard thresholding to scaled versions of $\bm{C}$ and $\bm{R}$ that are invariant under transformations. Specifically, we consider the transformed terms $\bm{C}(\widetilde{\bm{C}}^\top\widetilde{\bm{C}})^{1/2}$ and $\bm{R}(\widetilde{\bm{R}}^\top\widetilde{\bm{R}})^{1/2}$, where $\widetilde{\bm{C}}$ and $\widetilde{\bm{R}}$ are the auxiliary variables defined in Section \ref{sec:estimation_methodology}. The row-wise Euclidean norms of these transformed matrices remain invariant under equivalent decompositions, ensuring that the sparsity structure is consistent across the equivalent class of parameters.

The scaled hard thresholding procedure is then followed by a renormalization step to restore the original scale:
\begin{equation}
    \bm{C} \to \text{HT}(\bm{C}(\widetilde{\bm{C}}^\top\widetilde{\bm{C}})^{1/2},s_1)(\widetilde{\bm{C}}^\top\widetilde{\bm{C}})^{-1/2},~~
    \bm{R} \to \text{HT}(\bm{R}(\widetilde{\bm{R}}^\top\widetilde{\bm{R}})^{1/2},s_2)(\widetilde{\bm{R}}^\top\widetilde{\bm{R}})^{-1/2}.
\end{equation}
This approach is incorporated into the scaled gradient descent algorithm as shown in Algorithm \ref{alg:3}, where $s_1$ and $s_2$ denote the sparsity levels (i.e., the number of nonzero rows) for $\bm{C}$ and $\bm{R}$, respectively.

In the following analysis, denote $s_1^*$ and $s_2^*$ as the true row-wise sparsity levels of $\bm{C}^*$ and $\bm{R}^*$, respectively. To establish theoretical guarantees, we introduce a gradient stability condition tailored to the sparsity structure included by hard thresholding.

\begin{assumption}[Gradient Stability on Sparse Sets]
\label{def:sparse_stability}
Given sparsity levels $s_1$ and $s_2$, the partial gradients are stable on sparse index sets in the sense that there exist some parameters $\phi>0$ and $\xi_{C,s_1},\xi_{R,s_2},\xi_{11},\xi_{21},\dots,\xi_{1n},\xi_{2n}>0$ such that
\begin{equation}
    \begin{split}
        & \left\|\{(\nabla_{\bm{C}}\mathcal{L}_n(\bm{\Theta}) - \mathbb{E}[\nabla_{\bm{C}}\mathcal{L}_n(\bm{\Theta}))(\widetilde{\bm{C}}^\top\widetilde{\bm{C}})^{-1/2}]\}_{S_1}\right\|_\textup{F}^2 \leq \phi \|\cm{B}-\cm{B}^*\|_\textup{F}^2 + \xi_{C,s_1},\\
        & \left\|\{(\nabla_{\bm{R}}\mathcal{L}_n(\bm{\Theta}) - \mathbb{E}[\nabla_{\bm{R}}\mathcal{L}_n(\bm{\Theta}))(\widetilde{\bm{R}}^\top\widetilde{\bm{R}})^{-1/2}]\}_{S_2}\right\|_\textup{F}^2 \leq \phi \|\cm{B}-\cm{B}^*\|_\textup{F}^2 + \xi_{R,s_2},\\
        & \left\|(\bm{C}^\top\bm{C})^{-1/2}(\nabla_{\bm{L}_{1i}}\mathcal{L}_n(\bm{\Theta})-\mathbb{E}[\nabla_{\bm{L}_{1i}}\mathcal{L}_n(\bm{\Theta})])(\bm{L}_{2i}^\top\bm{R}^\top\bm{R}\bm{L}_{2i})^{-1/2} \right\|_\textup{F}^2 \leq \phi\|\bm{B}_i-\bm{B}_i^*\|_\textup{F}^2 + \xi_{1i},\\
        & \left\|(\bm{R}^\top\bm{R})^{-1/2}(\nabla_{\bm{L}_{2i}}\mathcal{L}_n(\bm{\Theta})-\mathbb{E}[\nabla_{\bm{L}_{2i}}\mathcal{L}_n(\bm{\Theta})])(\bm{L}_{1i}^\top\bm{C}^\top\bm{C}\bm{L}_{1i})^{-1/2} \right\|_\textup{F}^2 \leq \phi\|\bm{B}_i-\bm{B}_i^*\|_\textup{F}^2 + \xi_{2i},
    \end{split}
\end{equation}
uniformly over all index sets $S_1$ and $S_2$ with $\textup{card}(S_1)\leq s_1+s_1^*$ and $\textup{card}(S_2)\leq s_2+s_2^*$.
\end{assumption}

\begin{theorem}\label{thm:2}
    Under the Assumptions \ref{def:RCG} and \ref{def:sparse_stability}, and provided that the step size, sparsity level, and initial error satisfy 
    \begin{equation}
      \begin{split}
        \eta\asymp\ubar{\alpha}/\bar{\beta},\quad\phi\lesssim\ubar{\alpha}^2m^2,\quad\sqrt{s_1^*/(s_1-s_1^*)}+\sqrt{s_2^*/(s_2-s_2^*)}\lesssim\ubar{\alpha}\bar{\beta}^{-1},\quad\text{and}\\
        \textup{dist}(\bm{\Theta}^{(0)},\bm{\Theta}^*)\lesssim\ubar{\alpha}^{1/2}\bar{\beta}^{-1/2}\ubar{\sigma},
      \end{split}
    \end{equation}
    then the iterates of scaled gradient descent algorithm with scaled hard thresholding satisfy
    \begin{equation}
        \begin{split}
            \textup{dist}(\bm{\Theta}^{(t)},\bm{\Theta}^*)^2 & \leq (1-C\ubar{\alpha}\bar{\beta}^{-1})^t\cdot\textup{dist}(\bm{\Theta}^{(0)},\bm{\Theta}^*)^2 + C\ubar{\alpha}^{-2}m^{-2}\xi_s,\\
            \|\cm{B}^{(t)}-\cm{B}^*\|_\textup{F}^2 & \leq C(1-C\ubar{\alpha}\bar{\beta}^{-1})^t\cdot\textup{dist}(\bm{\Theta}^{(0)},\bm{\Theta}^*)^2 + C\ubar{\alpha}^{-2}m^{-2}\xi_s,
        \end{split}
    \end{equation}
    where $\xi_s=\xi_{C,s_1}+\xi_{R,s_2}+\sum_{i=1}^n(\xi_{1i}+\xi_{2i})$.
\end{theorem}

Theorem \ref{thm:2} presents the convergence guarantees of the scaled gradient descent algorithm with scaled hard thresholding given some sufficient conditions. Compared to Theorem \ref{thm:1}, the statisical convergence rates now depend on the estimation accuracy on the sparse index sets, captured by $\xi_s$, which is substantially smaller than $\xi$. The convergence rates of the shared and individualized components can be derived analogously.

\begin{algorithm}[!htp]
\caption{Scaled Gradient Descent Algorithm with Scaled Hard Thresholding}
\label{alg:3}
\begin{flushleft}
\linespread{1.75}\selectfont
\textbf{Input}: $\bm{C}^{(0)}$, $\bm{R}^{(0)}$, $\bm{L}_{1i}^{(0)}$, $\bm{L}_{2i}^{(0)}$ for $i=1,\dots,n$, step size $\eta$, no. of iterations $T$, sparsity levels $s_1$ and $s_2$\\
\textbf{for} $t=0,\dots,T-1$\\
\hspace*{1cm}$\bm{C}^{(t+0.5)}\leftarrow \bm{C}^{(t)} - \eta\nabla_{\bm{C}}\mathcal{L}_n(\bm{\Theta}^{(t)})\left(\widetilde{\bm{C}}^{(t)\top}\widetilde{\bm{C}}^{(t)}\right)^{-1}$\\
\hspace*{1cm}$\bm{R}^{(t+0.5)}\leftarrow \bm{R}^{(t)} - \eta\nabla_{\bm{R}}\mathcal{L}_n(\bm{\Theta}^{(t)})\left(\widetilde{\bm{R}}^{(t)\top}\widetilde{\bm{R}}^{(t)}\right)^{-1}$\\
\hspace*{1cm}\textbf{for} $i=1,\dots,n$ in parallel\\
\hspace*{2cm}$\bm{L}_{1i}^{(t+1)}\leftarrow\bm{L}_{1i}^{(t)}-\eta\left(\bm{C}^{(t)\top}\bm{C}^{(t)}\right)^{-1}\nabla_{\bm{L}_{1i}}\mathcal{L}_n(\bm{\Theta}^{(t)})\left(\bm{L}_{2i}^{(t)\top}\bm{R}^{(t)\top}\bm{R}^{(t)}\bm{L}_{2i}^{(t)}\right)^{-1}$\\
\hspace*{2cm}$\bm{L}_{2i}^{(t+1)}\leftarrow\bm{L}_{2i}^{(t)}-\eta\left(\bm{R}^{(t)\top}\bm{R}^{(t)}\right)^{-1}\nabla_{\bm{L}_{2i}}\mathcal{L}_n(\bm{\Theta}^{(t)})\left(\bm{L}_{1i}^{(t)\top}\bm{C}^{(t)\top}\bm{C}^{(t)}\bm{L}_{2i}^{(t)}\right)^{-1}$\\
\hspace*{1cm}\textbf{end for}\\
\hspace*{1cm}$\bm{C}^{(t+1)}\leftarrow \text{HT}\left(\bm{C}^{(t+0.5)}\left(\widetilde{\bm{C}}^{(t+0.5)\top}\widetilde{\bm{C}}^{(t+0.5)}\right)^{1/2},s_1\right)\left(\widetilde{\bm{C}}^{(t+0.5)\top}\widetilde{\bm{C}}^{(t+0.5)}\right)^{-1/2}$\\
\hspace*{1cm}$\bm{R}^{(t+1)}\leftarrow \text{HT}\left(\bm{R}^{(t+0.5)}\left(\widetilde{\bm{R}}^{(t+0.5)\top}\widetilde{\bm{R}}^{(t+0.5)}\right)^{1/2},s_2\right)\left(\widetilde{\bm{R}}^{(t+0.5)\top}\widetilde{\bm{R}}^{(t+0.5)}\right)^{-1/2}$\\
\textbf{end for}
\end{flushleft}\vspace{-0.1cm}
\end{algorithm}

\subsection{Initialization and Rank Selection with Sparsity}

If homogeneity is ignored, the scaled gradient descent algorithm with scaled hard thresholding can be extended in a straightforward manner. For each individual $i$, after taking the scaled gradient descent update for $\bm{C}_i$ and $\bm{R}_i$, we apply the hard thresholding to the scaled components $\bm{C}_i(\bm{R}_i^\top\bm{R}_i)^{1/2}$ and $\bm{R}_i(\bm{C}_i^\top\bm{C}_i)^{1/2}$, followed by the de-scaled factors $(\bm{R}_i^\top\bm{R}_i)^{-1/2}$ and $(\bm{C}_i^\top\bm{C}_i)^{-1/2}$, respectively, to restore the magnitudes. The resulting updates are summarized in Algorithm \ref{alg:4}.

\begin{algorithm}[!htp]
\caption{Scaled Gradient Descent Algorithm with Scaled Hard Thresholding for Heterogeneous Model}
\label{alg:4}
\begin{flushleft}
\linespread{1.75}\selectfont
\textbf{input}: initial values $\bm{C}_i^{(0)}$, $\bm{R}_i^{(0)}$ for $i=1,\dots,n$, step sizes $\eta_i$, no. of iterations $T$, sparsity levels $(s_1,s_2)$\\
\textbf{for} $i=1,\dots,n$ in parallel\\
\hspace*{1cm}\textbf{for} $t=0,\dots,T-1$\\
\hspace*{2cm}$\bm{C}_{i}^{(t+0.5)}\leftarrow\bm{C}_{i}^{(t)}-\eta_i\cdot\nabla_{\bm{C}_{i}}\overline{\mathcal{L}}_i(\bm{C}_{i}^{(t)},\bm{R}_i^{(t)})\left(\bm{R}_i^{(t)\top}\bm{R}_i^{(t)}\right)^{-1}$\\
\hspace*{2cm}$\bm{R}_{i}^{(t+0.5)}\leftarrow\bm{R}_{i}^{(t)}-\eta_i\cdot\nabla_{\bm{R}_{i}}\overline{\mathcal{L}}_i(\bm{C}_{i}^{(t)},\bm{R}_i^{(t)})\left(\bm{C}_i^{(t)\top}\bm{C}_i^{(t)}\right)^{-1}$\\
\hspace*{2cm}$\bm{C}_i^{(t+1)}\leftarrow \text{HT}\left(\bm{C}_i^{(t+0.5)}\left(\bm{R}_i^{(t+0.5)\top}\bm{R}_i^{(t+0.5)}\right)^{1/2},s_1\right)\left(\bm{R}_i^{(t+0.5)\top}\bm{R}_i^{(t+0.5)}\right)^{-1/2}$\\
\hspace*{2cm}$\bm{R}_i^{(t+1)}\leftarrow \text{HT}\left(\bm{R}^{(t+0.5)}\left(\bm{C}_i^{(t+0.5)\top}\bm{C}_i^{(t+0.5)}\right)^{1/2},s_2\right)\left(\bm{C}_i^{(t+0.5)\top}\bm{C}_i^{(t+0.5)}\right)^{-1/2}$\\
\hspace*{1cm}\textbf{end for}\\
\textbf{end for}\\
\textbf{return} $(\widecheck{\bm{C}}_1,\widecheck{\bm{R}}_1,\dots,\widecheck{\bm{C}}_n,\widecheck{\bm{R}}_n) = (\bm{C}_1^{(T)},\bm{R}_1^{(T)},\dots,\bm{C}_n^{(T)},\bm{R}_n^{(T)})$
\end{flushleft}\vspace{-0.1cm}
\end{algorithm}

This algorithm can be used for initialization with sparsity structures on $\bm{C}_i$ and $\bm{R}_i$. Then, after obtaining the sparse $\widecheck{\bm{C}}_i$ and $\widecheck{\bm{R}}_i$, for $i=1,\dots,n$, we can construct the cumulative subspace estimator $\bm{M}_C$ and $\bm{M}_R$ as in \eqref{eq:M_C_R}. Taking the top $K_1$ and $K_2$ eigenvectors of $\bm{M}_C$ and $\bm{M}_R$, we can obtain the initial values of $\bm{C}^{(0)}$ and $\bm{R}^{(0)}$ and the corresponding loading factors $\bm{L}_{1i}^{(0)}=\bm{C}^{(0)\top}\widecheck{\bm{C}}_i$ and $\bm{L}_{2i}^{(0)}=\bm{R}^{(0)\top}\widecheck{\bm{R}}_i$.

For rank selection, we can also consider a rank upper bound $\bar{r}$ and use the ridge-type ratio estimator
\begin{equation}\label{eq:ridge1}
    \widehat{r}={\arg\max}_{1\leq r\leq \bar{r}-1}\frac{\sum_{i=1}^n\sigma_{r}(\widecheck{\bm{B}}_i(\bar{r})) + \delta_1(n,m,\bar{s})}{\sum_{i=1}^n\sigma_{r+1}(\widecheck{\bm{B}}_i(\bar{r})) + \delta_1(n,m,\bar{s})},
\end{equation}
where the ridge parameter is replaced by $\delta_1(n,m,\bar{s})$ depending on $\bar{s}=\max(s_1,s_2)$. Similarly, after selecting $r$, we can also apply the ridge-type ratio estimators 
\begin{equation}\label{eq:ridge21}
  \widehat{K}_1 = {\arg\max}_{1\leq k\leq Cr}\frac{\lambda_{k}(\bm{M}_C) + \delta_2(n,m,\bar{s})}{\lambda_{r+1}(\bm{M}_C) + \delta_2(n,m,\bar{s})},
\end{equation}
and
\begin{equation}\label{eq:ridge22}
  \widehat{K}_2 = {\arg\max}_{1\leq k\leq Cr}\frac{\lambda_{k}(\bm{M}_R) + \delta_2(n,m,\bar{s})}{\lambda_{r+1}(\bm{M}_R) + \delta_2(n,m,\bar{s})},
\end{equation}
where $\delta_2(n,m,\bar{s})$ is the ridge parameter related to the sparsity level $\bar{s}$.

\section{A Couple of Specific Examples}\label{sec:statistical_models}

\subsection{Linear Trace Regression}\label{sec:5.1}

We first consider the linear trace regression setting, in which repeated observations are available for each of $n$ individuals. For invidual $i$ and observation $j$, the response is given by
\begin{equation}
    Y_{ij} = \langle\bm{X}_{ij},\bm{B}_i^*\rangle + \varepsilon_{ij},~\text{for}~j=1,\dots,m_i,~\text{and}~i=1,\dots,n,
\end{equation}
where $\varepsilon_{ij}$ is a mean-zero noise term. The coefficient matrices $\bm{B}_i^*$ are assumed to follow a low-rank structure with shared and individualized components, as introduced previously. Specifically, the coefficien matrix for individual $i$ admits the decomposition
\begin{equation}
  \bm{B}_i^*=\bm{C}\bm{L}_{1i}^*\bm{L}_{2i}^{*\top}\bm{R}^\top,
\end{equation}
where $\bm{C}\in\mathbb{R}^{p_1\times K_1}$ and $\bm{R}\in\mathbb{R}^{p_2\times K_2}$ are the shared components across individuals, and $\bm{L}_{1i}\in\mathbb{R}^{K_1\times r}$ and $\bm{L}_2\in\mathbb{R}^{K_2\times r}$ are the components specific to individual $i$.

To establish theoretical guarantees, we impose the following assumptions.
\begin{assumption}[Sub-Gaussian Covariates and Noise]\label{asmp:LinearModel1}
    The matrix-valued covariate $\bm{X}_{ij}$ is $\kappa^2$-sub-Gaussian, i.e., for any $\lambda>0$,
    \begin{equation}
        \max_{1\leq i\leq n}\sup_{\|\bm{v}_1\|_2=\|\bm{v}_2\|_2=1}\mathbb{E}[\exp(\lambda\bm{v}_1^\top\bm{X}_{ij}\bm{v}_2)]\leq\lambda^2\kappa^2/2.
    \end{equation}
    For $i=1,\dots,n$, the variance of $\bm{X}_{ij}$ satisfies
    \begin{equation}
      \alpha_{x,i}\leq \lambda_{\min}(\mathbb{E}[\textup{vec}(\bm{X}_{ij})\textup{vec}(\bm{X}_{ij})^\top]) \leq \lambda_{\max}(\mathbb{E}[\textup{vec}(\bm{X}_{ij})\textup{vec}(\bm{X}_{ij})^\top])\leq \beta_{x,i}.
    \end{equation}
    Moreover, the noise terms $\varepsilon_{ij}$ are mean-zero and $\sigma^2$-sub-Gaussian, 
    \begin{equation}
      \mathbb{E}[\exp(\lambda\varepsilon_{ij})]\leq \lambda^2\sigma^2/2,\quad\text{for any }\lambda>0.
    \end{equation}
\end{assumption}

\begin{assumption}[Balanced Sample Sizes]\label{asmp:LinearModel2}
    For all individuals, the number of repeated samples is asymptotically identical, i.e., $m_1\asymp \cdots\asymp m_n\asymp m$.
\end{assumption}

\begin{assumption}[Signal Strength]\label{asmp:LinearModel3}
    The population-level signal strength associated with the shared components satisfies
    \begin{equation}
      \sigma_1(\bm{\Sigma}_C)\asymp\sigma_{K_1}(\bm{\Sigma}_C)\asymp \sigma_1(\bm{\Sigma}_R)\asymp\sigma_{K_2}(\bm{\Sigma}_R)\asymp \sqrt{n},
    \end{equation}
    and the individual-specific signal strength associated with the individualized components satisfies that
    \begin{equation}
      \sigma_{r}(\bm{\Sigma}_1)\asymp\cdots\asymp\sigma_r(\bm{\Sigma}_n)\asymp C,
    \end{equation}
    for some constant $C$.
\end{assumption}

The sub-Gaussian conditions on $\bm{X}_{ij}$ and $\varepsilon_{ij}$ in Assumption \ref{asmp:LinearModel1} are widely adopted in matrix regression literature. Additionally, their distributions are allowed to vary across individuals. Assumption \ref{asmp:LinearModel2} can simplify the theoretical results, and it can be relaxed to accommodate imbalanced sample sizes. Assumption \ref{asmp:LinearModel3} implies that the shared low-rank subspaces are sufficiently informative, and that each individualized component contributes to a non-negligible signal.

As Algorithm \ref{alg:1} can be viewed as a special case of Algorithm \ref{alg:2} with $s_1=p_1$ and $s_2=p_2$, we only present the statistical results for the sparse setting. We apply Algorithm \ref{alg:2} to obtain the initial values $\widecheck{\bm{C}}_i$ and $\widecheck{\bm{R}}_i$, select $(\widehat{r},\widehat{K}_1,\widehat{K}_2)$ using the ridge-type ratio estimators in \eqref{eq:ridge1}, \eqref{eq:ridge21} and \eqref{eq:ridge22} with appropriate tuning parameters $\delta_1(n,m,\bar{s})\asymp n\bar{s}^{1/4}m^{-1/4}$ and $\delta_2(n,m,\bar{s})=n\bar{s}^{1/2}m^{-1/2}$. We then execute Algorithm \ref{alg:1} to obtain the final estimates $\widehat{\bm{\Theta}}=(\widehat{\bm{C}},\widehat{\bm{R}},\{\widehat{\bm{L}}_{1i},\widehat{\bm{L}}_{2i}\}_{i=1}^n)$ and corresponding $\widehat{\bm{B}}_i$. 

To evaluate the estimation accuracy of the shared components $\bm{C}$ and $\bm{R}$, we consider the $\sin\theta$ distance $\|\sin\theta(\widehat{\bm{C}},\bm{C}^*)\|^2_\textup{F}$ and $\|\sin\theta(\widehat{\bm{R}},\bm{R}^*)\|^2_\textup{F}$.

\begin{theorem}[Rates of Trace Regression under Homogeneity Pursuit]\label{thm:LinearModel}
    Under Assumptions \ref{asmp:LinearModel1} to \ref{asmp:LinearModel3}, if the sample size satisfies
    \begin{equation}
      m\gtrsim(\bar{\beta}_x^3\ubar{\alpha}_x^{-3}+\kappa^2\ubar{\alpha}_x^{-2})\bar{s}\log(n\bar{p}),
    \end{equation}
    as $nm\to\infty$, the ridge-type ratio estimators consistently select the true ranks:
    \begin{equation}
      \mathbb{P}(\widehat{r}=r,\widehat{K}_1=K_1,\widehat{K}_2=K_2)\to 1.
    \end{equation}
    Moreover, using the tuning parameters as in Theorem \ref{thm:2}, after
    \begin{equation}
      T\gtrsim \log((1-C\ubar{\alpha}_x\bar{\beta}_x^{-1})^{-1})\log n
    \end{equation}
    iterations, the estimates satisfy
    \begin{equation}
        \begin{split}
            \textup{dist}(\widehat{\bm{\Theta}},\bm{\Theta}^*)^2 & \lesssim \frac{\sigma^2\bar{\beta}_x^2}{\ubar{\alpha}_x^{2}}\cdot\frac{\bar{s}\log\bar{p}+\log n}{m},\\
            \frac{1}{n}\sum_{i=1}^n\|\widehat{\bm{B}}_i-\bm{B}_i^*\|_\textup{F}^2 & \lesssim \frac{\sigma^2\bar{\beta}_x^2}{\ubar{\alpha}_x^{2}}\cdot \frac{\bar{s}\log\bar{p}+\log n}{nm},\\
            \|\sin\theta(\widehat{\bm{C}},\bm{C}^*)\|^2_\textup{F} & \lesssim \frac{\sigma^2\bar{\beta}_x^2}{\ubar{\alpha}_x^{2}}\cdot \frac{\bar{s}\log\bar{p}+\log n}{nm},\\
            \|\sin\theta(\widehat{\bm{R}},\bm{R}^*)\|^2_\textup{F} & \lesssim \frac{\sigma^2\bar{\beta}_x^2}{\ubar{\alpha}_x^{2}}\cdot \frac{\bar{s}\log\bar{p}+\log n}{nm}.
        \end{split}
    \end{equation}
\end{theorem}

Theorem \ref{thm:LinearModel} presents the selection consistency of $(\widehat{r},\widehat{K}_1,\widehat{K}_2)$ and the statistical rates for linear trace regression, as well as sample size requirement and number of iteration requirement for convergence. Due to the homogeneity pursuit and sparsity structures, the requirement of $m$ is very mild. In addition, the sufficient number of iterations is roughly $\log(n)$, showing its computational efficiency and broad applicability for large-scale data with large $m$ and $n$. The statistical rates demonstrate that homogeneity pursuit yields an additional $n^{-1}$ factor in the convergence of heterogeneous components, improving estimation accuracy across individuals. The results can be directly extended to those of Algorithm \ref{alg:1} with $\bar{s}=\bar{p}$. Due to homogeneity pursuit, the average convergence rates of the individualized parameters $\bm{B}_i$, $\bm{L}_{1i}$, and $\bm{L}_{2i}$, as well as those of the shared components $\bm{C}$ and $\bm{R}$, have a rate improvement with a factor of $n^{-1}$, compared with the following corollary if homogeneity is ignored.

\begin{corollary}[Rates with Homogeneity Ignored]
  If $m\gtrsim \ubar{\alpha}_x^{-2}\kappa^2\bar{s}\log(n\bar{s})$, then under the same conditions as in Theorem \ref{thm:LinearModel}, for all $i=1,\dots,n$,
  \begin{equation}
    \|\widecheck{\bm{B}}_i - \bm{B}_i^*\|_\textup{F}^2 \lesssim \frac{\sigma^2\ubar{\alpha}_x^{-2}\bar{\beta}_x^2\bar{s}\log(n\bar{s})}{m}.
  \end{equation}
\end{corollary}

Define the class
\begin{equation}
  \begin{split}
    \mathcal{B}(\bar{p},\bar{s}^*;r,K_1,K_2)
    = & \{\cm{B}\in\mathbb{R}^{p_1\times p_2\times n}:p_1,p_2\leq\bar{p},\text{rank}(\cm{B}_i)\leq r\text{ for }1\leq i\leq n,\\
    & \|\cm{B}_{(\ell)}\|_{2,0}\leq\bar{s},~\text{rank}(\cm{B}_{(\ell)})\leq K_\ell,\text{ for }\ell=1,2 \}.
  \end{split}
\end{equation}

\begin{theorem}[Minimax Lower Bound for Linear Trace Regression]\label{thm:LinearModelLower}
  Suppose $\cm{B}^*$ in linear trace regressionbelongs to $\mathcal{B}(\bar{p},\bar{s}^*;r,K_1,K_2)$ for any $\bar{p}\geq\bar{s}^*$ and finite $(r,K_1,K_2)$. Then, for any estimator $\widetilde{\cm{B}}$ based on $n$ individuals and $m$ observations per individual,
  \begin{equation}
    \begin{split}
      \inf_{{\widetilde{\scalebox{0.7}{\cm{B}}}}}\sup_{\scalebox{0.7}{\cm{B}}^*\in\mathcal{B}(\bar{p},\bar{s}^*;r,K_1,K_2)}\mathbb{E}\left[\|\widetilde{\cm{B}}-\cm{B}^*\|_\textup{F}^2\right] \gtrsim \frac{\sigma^2}{\bar{\beta}_x^2}\cdot\frac{\bar{s}^*\log\bar{p}+\log n}{m}.
    \end{split}
  \end{equation}
\end{theorem}

The minimax lower bound in Theorem \ref{thm:LinearModelLower} matches the upper bounds in Theorem \ref{thm:LinearModel} up to constants depending on $\bar{\beta}_x$ and $\ubar{\alpha}_x$. This establishes the minimax optimality of the proposed estimator. By the approximate equivalence in Proposition \ref{prop:1}, the same conclusion holds for $\widehat{\bm{\Theta}}$.

\subsection{Generalized Linear Model}

We now extend the framework to the generalized linear model (GLM) with loss function
\begin{equation}\label{eq:GLM_loss}
    \mathcal{L}(\bm{B}_i;\bm{X}_{ij},y_{ij}) = g(\langle\bm{X}_{ij},\bm{B}_i\rangle) - y_{ij}\langle\bm{X}_{ij},\bm{B}_i\rangle,
\end{equation}
where $g(\cdot)$ is a smooth function. For instance, if $g(t)=t^2/2$, it reduces to the linear model in Section \ref{sec:5.1}; if $g(t)=\log(1+\exp(t))$, it becomes the matrix logistic regression. In this subsection, we focus on \eqref{eq:GLM_loss} with $g$ being a Lipschitz function. 

\begin{assumption}[Lipschitz Smooth Link Function]
  \label{asmp:Lip}
  The function $g$ is $L$-Lipschitz smooth, i.e., $|\nabla^2 g|\leq L$, for some constant $L>0$.
\end{assumption}

The Lipschitz condition naturally holds for logistic and probit regression. For this class of models, we consider Assumptions \ref{asmp:LinearModel2}-\ref{asmp:LinearModel3}, and the following assumption.

\begin{assumption}[Sub-Gaussian Covariates and Bounded Responses]\label{asmp:GLM1}
    The covariate $\bm{X}_{ij}$ is $\kappa^2$-sub-Gaussian and has bounded covariance eigenvalues (as in Assumption \ref{asmp:LinearModel1}). The response satisfies $|Y_{ij}|\leq B$ for some constant $B>0$.
\end{assumption}

For the binary classification problem with $Y_{ij}\in\{0,1\}$, Assumption \ref{asmp:GLM1} holds with $B=1$. In addition, we also consider the balanced sample size condition in Assumption \ref{alg:2} and signal strength condition in Assumption \ref{alg:3}. 

For the GLM in \eqref{eq:GLM_loss}, we consider the model with sparsity, as it naturally reduces to the conventional case by setting $s_1=p_1$ and $s_2=p_2$. For initialization and rank selection, we apply Algorithm \ref{alg:4} and ridge-type ratio estimators in \eqref{eq:ridge1}, \eqref{eq:ridge21} and \eqref{eq:ridge22}. Then, we have the following rank consistency and  convergence rates of the estimate obtained by Algorithm \ref{alg:3}.

\begin{theorem}\label{thm:GLM}
    Under Assumptions \ref{asmp:LinearModel2}, \ref{asmp:LinearModel3}, \ref{asmp:Lip} and \ref{asmp:GLM1}, if 
    \begin{equation}
      m\gtrsim(\bar{\beta}_x^3\ubar{\alpha}_x^{-3}+\kappa^2\ubar{\alpha}_x^{-2})\bar{s}\log(n\bar{p}),
    \end{equation}
    as $nm\to\infty$, the ridge-type ratio estimators consistently select the true ranks,
    \begin{equation}
      \mathbb{P}(\widehat{r}=r,\widehat{K}_1=K_1,\widehat{K}_2=K_2)\to 1.
    \end{equation}
    Moreover, using the tuning parameters as in Theorem \ref{thm:2}, after
    \begin{equation}
      T\gtrsim \log((1-C\ubar{\alpha}_x\bar{\beta}_x^{-1})^{-1})\log n
    \end{equation}
    iterations, the estimates satisfy
    \begin{equation}
        \begin{split}
            \textup{dist}(\widehat{\bm{\Theta}},\bm{\Theta}^*)^2 & \lesssim \frac{(L+B)^2\bar{\beta}_x^2}{\ubar{\alpha}_x^{2}}\cdot \frac{\bar{s}\log\bar{p}+\log n}{m},\\
            \frac{1}{n}\sum_{i=1}^n\|\widehat{\bm{B}}_i-\bm{B}_i^*\|_\textup{F}^2 & \lesssim \frac{(L+B)^2\bar{\beta}_x^2}{\ubar{\alpha}_x^{2}}\cdot \frac{\bar{s}\log\bar{p}+\log n}{nm},\\
            \|\sin\theta(\widehat{\bm{C}},\bm{C}^*)\|^2_\textup{F} & \lesssim \frac{(L+B)^2\bar{\beta}_x^2}{\ubar{\alpha}_x^{2}}\cdot \frac{\bar{s}\log\bar{p}+\log n}{m},\\
            \|\sin\theta(\widehat{\bm{R}},\bm{R}^*)\|^2_\textup{F} & \lesssim \frac{(L+B)^2\bar{\beta}_x^2}{\ubar{\alpha}_x^{2}}\cdot \frac{\bar{s}\log\bar{p}+\log n}{nm}.
        \end{split}
    \end{equation}
\end{theorem}
Theorem \ref{thm:GLM} presents the rank selection consistency of the ridge-type ratio estimators and statistical rates of the estimate obtained by Algorithm \ref{alg:3}. The required sample size and number of iterations are the same as those for linear trace model, as in Theorem \ref{thm:LinearModel}. The convergence rates confirm the efficiency improvement by homogeneity pursuit, compared with those of the fully heterogeneous models.

Similarly to trace regression, we have the following minimax lower bound for the matrix logistic regression.

\begin{theorem}[Minimax Lower Bound of Matrix Logistic Regression]\label{thm:GLMLower}
  Suppose $\cm{B}^*$ in matrix logistic regression belongs to $\mathcal{B}(\bar{p},\bar{s}^*;r,K_1,K_2)$ for any $\bar{p}\geq\bar{s}^*$ and finite $(r,K_1,K_2)$. Then, for any estimator $\widetilde{\cm{B}}$ based on $n$ individuals and $m$ observations per individual,
  \begin{equation}
    \begin{split}
      \inf_{{\widetilde{\scalebox{0.7}{\cm{B}}}}}\sup_{\scalebox{0.7}{\cm{B}}^*\in\mathcal{B}(\bar{p},\bar{s}^*;r,K_1,K_2)}\mathbb{E}\left[\|\widetilde{\cm{B}}-\cm{B}^*\|_\textup{F}^2\right] \gtrsim \frac{\bar{s}^*\log\bar{p}+\log n}{\bar{\beta}_x^2m}.
    \end{split}
  \end{equation}
\end{theorem}

The minimax lower bound in Theorem \ref{thm:GLMLower} matches the upper bounds in Theorem \ref{thm:GLM}, which shows the minimax optimality of the proposed method. By the approximate equivalence in Proposition \ref{prop:1}, the minimax optimality of $\widehat{\bm{\Theta}}$ can be established accordingly.

\section{Simulation Experiments and Real Data Analysis}\label{sec:numerical}

In the first two subsections, we conduct a series of simulation experiments to empirically verify the theoreitical guarantees of the proposed estimation procedure. In the third subsection, we apply the methodology to New York Taxi Dataset to evaluate its practical utility. Following the framework outlined in Section \ref{sec:statistical_models}, we consider two models:
\begin{itemize}
  \item Model I (Linear Trace Regression):
  \begin{equation}
    Y_{ij} = \langle\bm{X}_{ij},\bm{B}_i^*\rangle + \varepsilon_{ij},\quad i=1,\dots,n,\quad j=1,\dots,m,
  \end{equation}
  where $\varepsilon_{ij}$ is a mean-zero noise term.
  \item Model II (Matrix Logistic Regression):
  \begin{equation}
    \mathbb{P}(Y_{ij}=1|\bm{X}_{ij}) = \frac{\exp(\langle\bm{X}_{ij},\bm{B}_i^*\rangle)}{1+\exp(\langle\bm{X}_{ij},\bm{B}_i^*\rangle)},\quad\mathbb{P}(Y_{ij}=0|\bm{X}_{ij}) = \frac{1}{1+\exp(\langle\bm{X}_{ij},\bm{B}_i^*\rangle)},
  \end{equation}
  for $i=1,\dots,n$ and $j=1,\dots,m$.
\end{itemize}
For both models, the coefficient matrix $\bm{B}_i^*$ is assumed to follow a low-rank structure with shared and individualized components. Specifically,
\begin{equation}
  \bm{B}_i^*=\bm{C}^*\bm{L}_{1i}^*\bm{L}_{2i}^{*\top}\bm{R}^{*\top},\quad i=1,\dots,n.
\end{equation}

\subsection{Experiment I: Rank Selection Consistency}

Experiment I aims to verify the rank selection consistency of proposed ridge-type ratio estimators. For both models, we consider the dimensions $p_1=p_2=20$, ranks $K_1=K_2=4$, $r=2$, and two sets of sample size settings $(n=16,m\in\{8,16,32,64,128\})$ and $(m=16,n\in\{8,16,32,64,128\})$. For $\bm{B}_i^*$, we consider two parameter settings (a) no sparsity structure in $\bm{C}$ and $\bm{R}$; (b) only the first five rows in $\bm{C}$ and $\bm{R}$ are nonzero. For each setting and replication, we fix $\bm{\Sigma}=\text{diag}(5,5)$ and generate the random matrices $\bm{U}_i$, $\bm{V}_i$, $\bm{C}^*$, and $\bm{R}^*$ with orthonormal columns, so that $\bm{B}_i^*=\bm{C}^*\bm{U}_i\bm{\Sigma}\bm{V}_i^\top\bm{R}^{*\top}$. Moreover, for both models we consider $\text{vec}(\bm{X}_{ij})\sim N(0,\bm{I}_{p_1p_2})$, and generate $Y_{ij}\sim N(0,1)$ for Model I.

For the parameter setting (a) without sparsity structure, we apply the Algorithm \ref{alg:2} with rank $\bar{r}=5$ to obtain $\widecheck{\bm{B}}_i(5)$ for $i=1,\dots,n$, and use the ridge-type ratio estimator in \eqref{eq:ridge_1} with $\delta_1(n,m,\bar{p})=0.1n\bar{p}m^{-1/4}$ to select $\widehat{r}$. Then, we apply the Algorithm \ref{alg:2} again with rank $\widehat{r}$ to obtain $\widecheck{\bm{B}}_i(\widehat{r})$ and the corresponding aggregate matrices $\bm{M}_C$ and $\bm{M}_R$ in \eqref{eq:M_C_R}, and apply another ridge-type ratio estimators in \eqref{eq:ridge_21} and \eqref{eq:ridge_22} with $\delta_2(n,m,\bar{p})=0.1n\bar{p}m^{-1/2}$ to select $\widehat{K}_1$ and $\widehat{K}_2$. For the parameter setting (b), we replace Algorithm \ref{alg:2} with Algorithm \ref{alg:4}, and use the ridge-type ratio estimators in \eqref{eq:ridge1}, \eqref{eq:ridge21} and \eqref{eq:ridge22} with $\delta_1(n,m,\bar{s})=0.1n\bar{s}m^{-1/4}$ and $\delta_2(n,m,\bar{s})=0.1n\bar{s}m^{-1/2}$. For each configuration, we replicate the rank selection procedure 500 times, and summarize the proportion of replications that achieve correct selection of $(r,K_1,K_2)$ in Figure \ref{fig:Exp1}. 

\begin{figure}[!htp]
  \includegraphics[width=\textwidth]{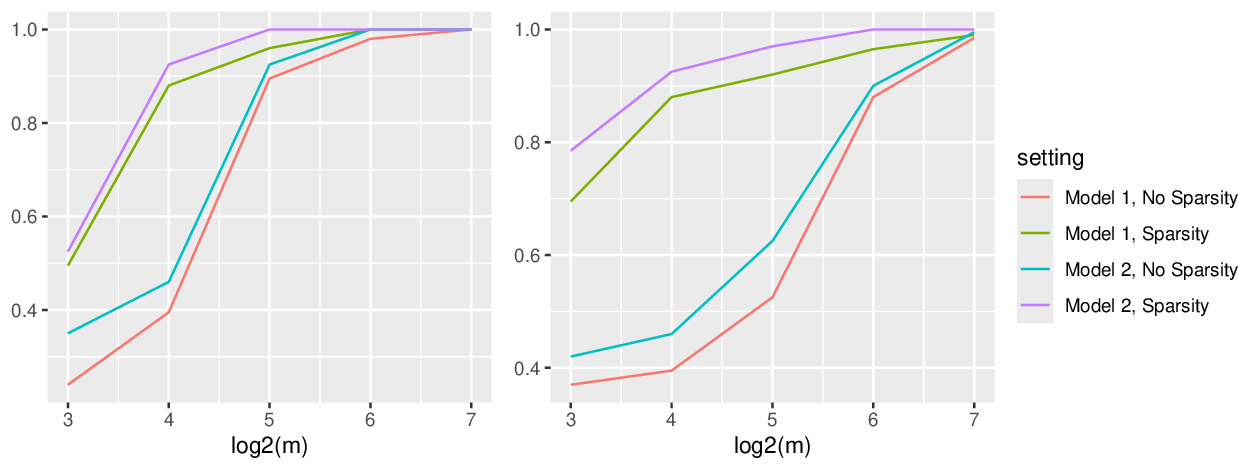}
  \vspace{-0.6cm}
  \caption{Proportions of correct rank selection for $n=16$ and $\log_2(m)\in\{3,4,5,6,7\}$ (left panel) and for $m=16$ and $\log_2(n)\in\{3,4,5,6,7\}$ (right panel).}
  \label{fig:Exp1}
\end{figure}
\vspace{-0.5cm}

The results in Figure \ref{fig:Exp1} confirms that when either $m$ or $n$ is sufficiently large, the proposed ridge-type ratio estimators consistently recover the true ranks. In addition, for model with additional sparsity structures on $\bm{C}$ and $\bm{R}$, the correct selection rate is higher than those without sparsity, reflecting the improved initialized efficiency from enforcing structured sparsity on the coefficient matrices.

\subsection{Experiment II: Convergence Rates}

Experiment II is designed to empirically verify the statistical convergence rates of the proposed estimators, as established in Theorems \ref{thm:LinearModel} and \ref{thm:GLM} for linear and logistic models, respectively. Similarly to Experiment I, we consider both models with $p_1=p_2=20$, $K_1=K_2=4$ and $r=2$, as well as the parameter settings (a) and (b) (without and with sparsity structures). We first apply the data-driven rank selection methods as in Experiment I, and then apply Algorithm \ref{alg:1} for setting (a) and \ref{alg:3} for setting (b), respectively.

\begin{figure}[!htp]
  \includegraphics[width=\textwidth]{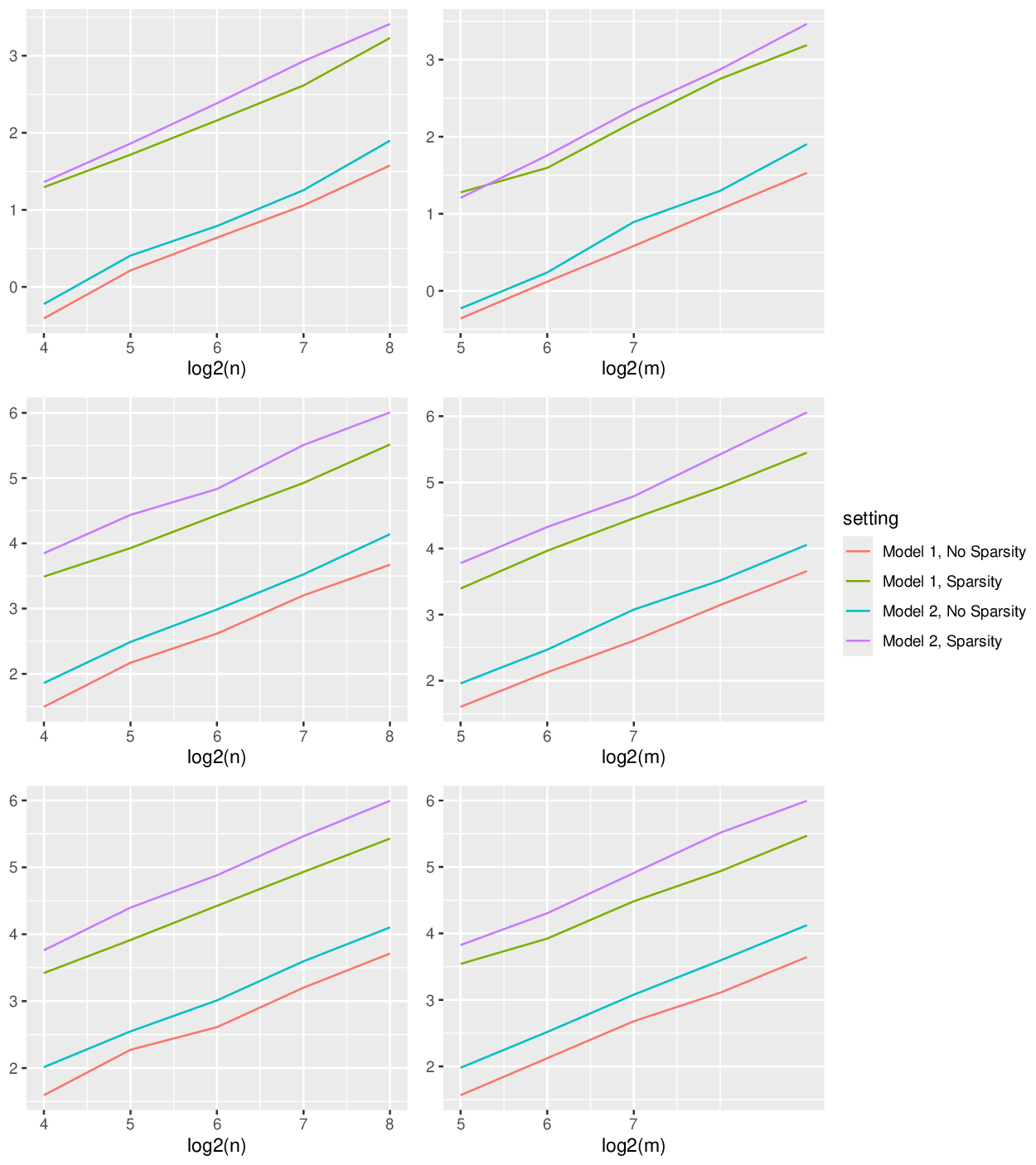}
  \vspace{-0.6cm}
  \caption{Average log errors $-\log(\|\widehat{\cm{B}}-\cm{B}^*\|_\textup{F}^2/n)$ (first row), $-\log\|\widehat{\bm{C}}\widehat{\bm{C}}^\top-\bm{C}^*\bm{C}^{*\top}\|_\textup{F}^2$ (second row), and $-\log\|\widehat{\bm{R}}\widehat{\bm{R}}^\top-\bm{R}^*\bm{R}^{*\top}\|_\textup{F}^2$ (third row), with $m=40$ and $\log_2(n)\in\{4,5,6,7,8\}$ (left panel) and with $n=20$ and $\log_2(m)\in\{5,6,7,8,9\}$ (right panel).}
  \label{fig:Exp2}
\end{figure}

First, for both models, we consider the fixed $m=40$ and varying $n\in\{16,32,64,128,256\}$. According to Theorems \ref{thm:LinearModel} and \ref{thm:GLM}, $-\log(\|\widehat{\cm{B}}-\cm{B}^*\|_\text{F}^2/n) = C+\log n + \log m-\log(\bar{s}\log(\bar{p})+\log n)$, which implies that the negative log error has an approximate linear relationship with $\log(n)$ given other parameters fixed. Similarly, we also consider the fixed $n=20$ and varying $m\in\{32,64,128,256,512\}$, where the negative log error grows linearly with $\log m$. To evaluate the estimation accuracy of the shared components $\bm{C}$ and $\bm{R}$, we normalize $\widehat{\bm{C}}$ and $\widehat{\bm{R}}$ to have orthonormal columns, and calculate $-\log\|\widehat{\bm{C}}\widehat{\bm{C}}^\top - \bm{C}^*\bm{C}^{*\top}\|_\text{F}^2$ and $-\log\|\widehat{\bm{R}}\widehat{\bm{R}}^\top - \bm{R}^*\bm{R}^{*\top}\|_\text{F}^2$. By the approximate equivalence between the projection matrix error and $\sin\Theta$ distance, the negative log errors for $\bm{C}$ and $\bm{R}$ are also supposed to grow linearly with $\log n$ and $\log m$. To numerically verify these results, we replicate the estimation procedure 500 times for each pair of $m$ and $n$, and summarize the average negative log errors in Figure \ref{fig:Exp2}.

The observed linear trends in negative log errors confirm the predicted convergence rates in Theorems \ref{thm:LinearModel} and \ref{thm:GLM}. The shared components $\bm{C}$ and $\bm{R}$ exhibit improved estimation accuracy under homogeneity pursuit, highlighting the benefits of shared low-rank structure in multi-individual regression.

\subsection{Real Data Analysis: New York Taxi Trip Dataset}

This subsection presents an application of the proposed multi-individual matrix regression framework with homogeneity pursuit to the New York City (NYC) Taxi Trip Dataset, aiming to predict taxi trip fees based on pickup-dropoff area interactions across different times of the day.

The data are obtained from the NYC Taxi Trip Record Dataset (\url{www.nyc.gov/site/tlc/about/tlc-trip-record-data.page}), which records detailed trip information for yellow taxis operating in NYC. For our analysis, we focus on predicting the total taxi fare (in thousands of U.S. dollars), denoted as $Y_{ij}$, during the time slot $i$ and week $j$. The trip origin-destination interaction is encoded via a $69\times 69$ matrix $\bm{X}_{ij}$, where each row corresponds to one of the 69 predefined taxi pickup zones in downtown Manhattan, and each column corresponds to one of the 69 dropoff zones. The entry in $X_{ij}$ reflects the aggregate trip distance during time slot $i$ and week $j$. 

To capture temporal heterogeneity in traffic patterns, the daytime is divided into four distinct time slots: morning (8 AM - 11 AM) (typically associated with trraffic congestion), noon (11 AM - 2 PM), afternoon (2 PM - 5 PM), and later afternoon (5 PM - 8 PM) (also prone to congestion). Each of these four time slots per day is treated as an individual, resulting in 28 individuals per week (i.e., $n=28$), aggregated over 52 weeks in the year 2018 (i.e., $m=52$). Thus, the full dataset consists of $1456$ observations, each contributing multiple trip records.

We adopt the linear trace regression model to relate the response $Y_{ij}$ (total fare in \$1000s) to the predictor matrix $\bm{X}_{ij}$ (trip origin-destination matrix):
\begin{equation}
  Y_{ij} = \langle\bm{X}_{ij},\bm{B}_i^*\rangle+\varepsilon_{ij},
\end{equation}
where $\bm{B}_i^*\in\mathbb{R}^{69\times 69}$ is the coefficient matrix for individual $i$, and $\varepsilon_{ij}$ represents noise. The data are split into a training set (first 40 weeks) and a testing set (last 12 weeks). All modeling and rank selection are performed on the training data, and prediction performance is evaluated on the test set.

We compare the following three modeling approaches.
\begin{itemize}
  \item[1.] Fully Homogeneous Model: Assumes a single shared matrix across all individuals $\bm{B}_i^*=\bm{B}^*$. Esimated via ordinary least squares (Homo-OLS) and via low-rank estimation as in Algorithm \ref{alg:2} with $n=1$, $m=28\times 40=1120$, and $r=2$ (Homo-LR).
  \item[2.] Fully Heterogeneous Model: Allows each individual $i$ to have its own uncontrained coefficient matrix $\bm{B}_i^*$. Estimated via OLS (Hetero-OLS) and Algorithm \ref{alg:2} with $n=28$, $m=40$, and $r=2$ (Hetero-LR).
  \item[3.] Proposed Homogeneity-Pursuit Model: Assumes a structured low-rank representation with shared components $\bm{C}$ and $\bm{R}$, and individualized components $\bm{L}_{1i}$ and $\bm{L}_{2i}$: $\bm{B}_i^* = \bm{C}\bm{L}_{1i}\bm{L}_{2i}^\top\bm{R}^{*\top}$. Estimated via Algorithm \ref{alg:1} with $r=2$ and $K_1=K_2=4$ (Homo-Pursuit) and Algorithm \ref{alg:3} with $r=2$, $K_1=K_2=4$, and $s_1=s_2=10$ (Homo-Pursuit-Sparsity).
\end{itemize}

We evaluate the out-of-sample predictive performance of the six methods using Root Mean Squared Error (RMSE) on the test set (weeks 41-52). The average RMSE and its standard deviation are summarized in Table \ref{tbl:1}.

\begin{table}
  \caption{Prediction performance (RMSE and standard deviation) on test set.}
  \label{tbl:1}
  \begin{tabular}{ccc}
    \toprule
    \textbf{Method} & \textbf{Average RMSE} (\$1000) & \textbf{Std. Dev.} \\
    \midrule
    Homo-OLS & 298.7 & 44.1 \\
    Homo-LR  & 280.1 & 43.2 \\
    Hetero-OLS & 274.1 & 43.8 \\
    Hetero-LR & 263.3  & 39.4 \\
    Homo-Pursuit & 233.1 & \textbf{35.4} \\
    Homo-Pursuit-Sparsity & \textbf{231.5} & 36.0 \\
    \bottomrule
  \end{tabular}
\end{table}

The results demonstrate several key insights. Across all modeling strategies, incorporating low-rank structure leads to better predictive accuracy compared to their OLS counterparts, which highlights the benefit of dimension reduction in handling high-dimensional yet intrinsically low-rank structure of $\bm{B}_i^*$. In addition, the heterogeneous models (Hetero-OLS and Hetero-LR) generally outperform the fully homogeneous ones, reflecting the temporal and spatial variability in taxi demand and traffic conditions. However, the proposed models (Homo-Pursuit and Homo-Pursuit-Sparsity), which balance shared information and individual flexibility, achieve the best predictive performance. By learning a small number of shared components $(r=2,K_1=K_2=4)$, the proposed methods effectively distills common traffic and fare patterns across different time slots, while still allowing for individualized adjustments. The reduced vriance in prediction error further indicates greater robustness and stability of the homogeneity puruit approach.

The real data analysis on the NYC Taxi Trip Dataset provides empirical support for the effectiveness of the proposed multi-individual matrix regression framework with homogeneity pursuit. By leveraging structured low-rank representations, the method not only improves predictive accuracy but also enhances model interpretability and stability, making it a promising tool for urban mobility analytics and beyond.

\section{Conclusion and Discussion}\label{sec:conclusion}

We have introduced a multi-individual matrix regression framework with homogeneity pursuit, which models matrix-valued responses across individuals by decomposing each coefficient matrix into shared low-rank components and individualized components. This structure captures common patterns across individuals while allowing for personalized deviations, improving both interpretability and estimation efficiency.

We formulated the problem under a linear trace regression model and extended it to generalized linear models (GLMs). A scaled gradient descent algorithm with scaled hard thresholding was developed to estimate the shared and individualized components, with theoretical guarantees on estimation error, rank selection consistency, and convergence rates. Our method incorporates sparsity on shared components and uses ridge-type ratio estimators for rank selection, achieving strong theoretical performance under mild conditions.

Theoretically, we demonstrated that homogeneity pursuit leads to improved convergence rates, with an additional $n^{-1}$ factor reflecting the benefits of shared structure. Minimax lower bounds established the optimality of our approach in high-dimensional settings. Simulation experiments verified rank selection consistency, convergence behavior, and the advantages of low-rank and sparse modeling. In particular, the proposed estimators achieved rates aligned with our theoretical predictions.

In a real data application to the New York City Taxi Trip Dataset, the proposed method outperformed fully homogeneous, fully heterogeneous, and their low-rank or OLS counterparts, achieving the lowest prediction error and smallest variability. These results highlight the practical value of homogeneity pursuit in extracting structured information from multi-individual matrix data.

In summary, our framework provides a flexible, interpretable, and statistically efficient approach to modeling matrix-valued responses with shared and individual structures. It advances the methodology for multi-individual regression and is applicable to problems in urban analytics, personalized modeling, and beyond. Future work may explore nonlinear extensions and dynamic adaptations.

%%  MR numbers will be added by VTeX.                      %%
%%                                                         %%
%%  Use \cite{...} to cite references in text.             %%
%%                                                         %%
%%%%%%%%%%%%%%%%%%%%%%%%%%%%%%%%%%%%%%%%%%%%%%%%%%%%%%%%%%%%%

%% if your bibliography is in bibtex format, uncomment commands:
%\bibliographystyle{imsart-number} % Style BST file (imsart-number.bst or imsart-nameyear.bst)
%\bibliography{bibliography}       % Bibliography file (usually '*.bib')

\begin{appendix}

\section{Tensor Algebra and Notations}\label{append:A}

Tensors are multi-dimensional arrays that generalize matrices to higher-order data. A $d$-th order tensor is represented as $\cm{X} \in \mathbb{R}^{p_1 \times p_2 \times \cdots \times p_d}$, where $p_k$ is the dimension along the $k$-th mode. In this article, we adopt the following notation:

\begin{itemize}
  \item \textbf{Vectors}: denoted by boldface lowercase letters, e.g., $\bm{x} \in \mathbb{R}^p$,
  \item \textbf{Matrices}: denoted by boldface uppercase letters, e.g., $\bm{X} \in \mathbb{R}^{p \times q}$,
  \item \textbf{Tensors}: denoted by boldface Euler letters, e.g., $\cm{X} \in \mathbb{R}^{p_1 \times \cdots \times p_d}$.
\end{itemize}

We refer readers to \citet{kolda2009tensor} for a comprehensive review of tensor operations and decompositions.

\paragraph*{Mode-$k$ matricization}

The \textbf{mode-$k$ matricization} (or \textbf{unfolding}) of a tensor $\cm{X} \in \mathbb{R}^{p_1 \times \cdots \times p_d}$ is a matrix obtained by rearranging the fibers of $\cm{X}$ along the $k$-th mode into columns. The result is denoted by $\cm{X}_{(k)} \in \mathbb{R}^{p_k \times p_{-k}}$, where $p_{-k} = \prod_{\ell=1, \ell \neq k}^d p_\ell$.

Each column of $\cm{X}_{(k)}$ corresponds to a \textit{fiber} of $\cm{X}$ along mode $k$, stacked in lexicographic order. The element $(i_1, i_2, \ldots, i_d)$ of $\cm{X}$ is mapped to the $(i_k, j)$-th entry of $\cm{X}_{(k)}$, where the index $j$ is given by:
\begin{equation}
  j = 1 + \sum_{\substack{s=1 \\ s 
\neq k}}^d (i_s - 1) \cdot J_s^{(k)},
  \quad \text{with} \quad
  J_s^{(k)} = \prod_{\substack{\ell=1 \\ \ell < s \\ \ell 
\neq k}}^d p_\ell, \quad \text{and} \quad p_0 = 1.
\end{equation}
This operation is central to defining mode-$k$ products and understanding Tucker decompositions.

\paragraph*{Mode-$k$ Product}

For a tensor $\cm{X} \in \mathbb{R}^{p_1 \times \cdots \times p_d}$ and a matrix $\bm{Y} \in \mathbb{R}^{q_k \times p_k}$, the \textbf{mode-$k$ product}, denoted $\cm{X} \times_k \bm{Y}$, results in a new tensor of size $p_1 \times \cdots \times p_{k-1} \times q_k \times p_{k+1} \times \cdots \times p_d$. Its entries are given by:
\begin{equation}
  \left( \cm{X} \times_k \bm{Y} \right)_{i_1 \cdots i_{k-1} j i_{k+1} \cdots i_d}
  = \sum_{i_k=1}^{p_k} \cm{X}_{i_1 \cdots i_k \cdots i_d} \cdot \bm{Y}_{j i_k},
  \quad \text{for all } j = 1, \ldots, q_k.
\end{equation}
This operation recombines the tensor along mode $k$ with the matrix $\bm{Y}$.

\paragraph*{Inner Product}

For two tensors $\cm{X} \in \mathbb{R}^{p_1 \times \cdots \times p_d}$ and $\cm{Y} \in \mathbb{R}^{p_1 \times \cdots \times p_{d}}$, their \textbf{inner product} is defined as:
\begin{equation}
  \langle \cm{X}, \cm{Y} \rangle
  = \sum_{i_1=1}^{p_1} \cdots \sum_{i_{d}=1}^{p_{d}} \cm{X}_{i_1 \cdots i_d} \cdot \cm{Y}_{i_1 \cdots i_d},
\end{equation}
and it results in a $(d - d_0)$-th order tensor with entries indexed by $(i_{d_0+1}, \ldots, i_d)$.
In the special case where $d = d_0$, the generalized inner product reduces to the standard \textbf{Frobenius inner product}, and we define the \textbf{Frobenius norm} of $\cm{X}$ as:
\begin{equation}
  \| \cm{X} \|_{\text{F}} = \sqrt{ \langle \cm{X}, \cm{X} \rangle }.
\end{equation}

\paragraph*{Tucker Decomposition and Tucker Ranks}

The \textbf{Tucker rank} of a tensor $\cm{X} \in \mathbb{R}^{p_1 \times \cdots \times p_d}$ is a vector $(r_1, \ldots, r_d)$, where each $r_k$ is the rank of the mode-$k$ matricization $\cm{X}_{(k)}$, i.e.,
\begin{equation}
  r_k = \text{rank}(\cm{X}_{(k)}) \in \mathbb{N}, \quad \text{for } k = 1, \ldots, d.
\end{equation}

While Tucker ranks are defined via matricization ranks, they correspond to the number of components retained in the \textbf{Tucker decomposition} of $\cm{X}$. If $\cm{X}$ has Tucker ranks $(r_1, \ldots, r_d)$, it can be written as:
\begin{equation}
  \cm{X} = \cm{Y} \times_{j=1}^d \bm{Y}_j
  = \cm{Y} \times_1 \bm{Y}_1 \times_2 \bm{Y}_2 \cdots \times_d \bm{Y}_d,
\end{equation}
where $\bm{Y}_j \in \mathbb{R}^{p_j \times r_j}$ is the factor matrix for mode $j$, and $\cm{Y} \in \mathbb{R}^{r_1 \times \cdots \times r_d}$ is the core tensor.

The mode-$k$ matricization of $\cm{X}$ under the Tucker decomposition can be expressed as:
\begin{equation}
  \cm{X}_{(k)} = \cm{Y}_{(k)} \left( \otimes_{j=1,j\neq k}^d \bm{Y}_j \right)^\top,
\end{equation}
where $\otimes$ is the Kronecker product, and the product is taken over all modes except $k$. This structure is central to our algorithm and theoretical analysis, as it allows us to work with low-rank representations in high-dimensional spaces.

\section{Computational Convergence Analysis}\label{append:B}

This appendix presents the proof of the computational convergence results for the proposed algorithms. Appendix \ref{sec:A.1} presents the proof of Theorem \ref{thm:1}, Appendix \ref{sec:A.2} presents the proof of Corollaries \ref{cor:1} and \ref{cor:2}, and Appendix \ref{sec:A.3} presents the proof of Theorem \ref{thm:2}. Auxiliary lemmas and their proofs are relegated to Appendix \ref{sec:A.4}.

\subsection{Proof of Theorem \ref{thm:1}}\label{sec:A.1}

\begin{proof}
    
    The proof of Theorem \ref{thm:1} consists of three steps and is constructed in an inductive approach. In the first step, we present the notations and state some key conditions. In the second step, we develop the core convergence analysis. Specifically, given some conditions hold at the step $t$, we present an upper bound for the step $t+1$, i.e.,
    \begin{equation}
        \text{dist}(\bm{\Theta}^{(t+1)},\bm{\Theta}^*)^2 \leq (1-C\ubar{\alpha}/\bar{\beta})\cdot\text{dist}(\bm{\Theta}^{(t)},\bm{\Theta}^*)^2 + C\ubar{\alpha}^{-1}\bar{\beta}^{-1}\xi.
    \end{equation}
    In the last step, we show that the required conditions also hold for the step $t+1$, and complete the proof by induction.\\

    \noindent\textit{Step 1.} (Notations and conditions) 

    \noindent We first present the notations in the proof. Denote the de-scaled partial gradients as
    \begin{equation}
        \begin{split}
            \bm{G}_C^{(t)} & = \nabla_{\bm{C}}\mathcal{L}_n(\bm{\Theta}^{(t)})(\widetilde{\bm{C}}^{(t)\top}\widetilde{\bm{C}}^{(t)})^{-1/2},~~\bm{G}_R^{(t)} = \nabla_{\bm{R}}\mathcal{L}_n(\bm{\Theta}^{(t)})(\widetilde{\bm{R}}^{(t)\top}\widetilde{\bm{R}}^{(t)})^{-1/2}, \\
            \bm{G}_{1i}^{(t)} & = (\bm{C}^{(t)\top}\bm{C}^{(t)})^{-1/2}\nabla_{\bm{L}_{1i}}\mathcal{L}_n(\bm{\Theta}^{(t)})(\bm{L}_{2i}^{(t)\top}\bm{R}^{(t)\top}\bm{R}^{(t)}\bm{L}_{2i}^{(t)})^{-1/2}, \\
            \bm{G}_{2i}^{(t)} & = (\bm{R}^{(t)\top}\bm{R}^{(t)})^{-1/2}\nabla_{\bm{L}_{2i}}\mathcal{L}_n(\bm{\Theta}^{(t)})(\bm{L}_{1i}^{(t)\top}\bm{C}^{(t)\top}\bm{C}^{(t)}\bm{L}_{1i}^{(t)})^{-1/2},
        \end{split}
    \end{equation}
    and their corresponding estimation errors as
    \begin{equation}
        \begin{split}
            \bm{\Delta}_C^{(t)} & =\bm{G}_C^{(t)} - \mathbb{E}[\nabla_{\bm{C}}\mathcal{L}_n(\bm{\Theta}^{(t)})](\widetilde{\bm{C}}^{(t)\top}\widetilde{\bm{C}}^{(t)})^{-1/2},\\
            \bm{\Delta}_R^{(t)} & =\bm{G}_R^{(t)} - \mathbb{E}[\nabla_{\bm{R}}\mathcal{L}_n(\bm{\Theta}^{(t)})](\widetilde{\bm{R}}^{(t)\top}\widetilde{\bm{R}}^{(t)})^{-1/2},\\
            \bm{\Delta}_{1i}^{(t)} & = \bm{G}_{1i}^{(t)} - (\bm{C}^{(t)\top}\bm{C}^{(t)})^{-1/2}\mathbb{E}[\nabla_{\bm{L}_{1i}}\mathcal{L}_n(\bm{\Theta}^{(t)})](\bm{L}_{2i}^{(t)\top}\bm{R}^{(t)\top}\bm{R}^{(t)}\bm{L}_{2i}^{(t)})^{-1/2},\\
            \bm{\Delta}_{2i}^{(t)} & = \bm{G}_{2i}^{(t)} - (\bm{R}^{(t)\top}\bm{R}^{(t)})^{-1/2}\mathbb{E}[\nabla_{\bm{L}_{2i}}\mathcal{L}_n(\bm{\Theta}^{(t)})](\bm{L}_{1i}^{(t)\top}\bm{C}^{(t)\top}\bm{C}^{(t)}\bm{L}_{1i}^{(t)})^{-1/2}.
        \end{split}
    \end{equation}
    Then, the proposed scaled gradient descent update at step $t$ can be formulated as
    \begin{equation}
        \begin{split}
            \bm{C}^{(t+1)} = & \bm{C}^{(t)} - \eta\cdot\bm{G}_C^{(t)}(\widetilde{\bm{C}}^{(t)\top}\widetilde{\bm{C}}^{(t)})^{-1/2}\\
            = & \bm{C}^{(t)} - \eta\cdot\mathbb{E}[\nabla_{\bm{C}}\mathcal{L}_n(\bm{\Theta}^{(t)})](\widetilde{\bm{C}}^{(t)\top}\widetilde{\bm{C}}^{(t)})^{-1} - \eta\cdot\bm{\Delta}_C^{(t)}(\widetilde{\bm{C}}^{(t)\top}\widetilde{\bm{C}}^{(t)})^{-1/2},\\
            \bm{R}^{(t+1)} = & \bm{R}^{(t)} - \eta\cdot\bm{G}_R^{(t)}(\widetilde{\bm{R}}^{(t)\top}\widetilde{\bm{R}}^{(t)})^{-1/2}\\
            = & \bm{R}^{(t)} - \eta\cdot\mathbb{E}[\nabla_{\bm{R}}\mathcal{L}_n(\bm{\Theta}^{(t)})](\widetilde{\bm{R}}^{(t)\top}\widetilde{\bm{R}}^{(t)})^{-1} - \eta\cdot\bm{\Delta}_R^{(t)}(\widetilde{\bm{R}}^{(t)\top}\widetilde{\bm{R}}^{(t)})^{-1/2},\\
            \bm{L}_{1i}^{(t+1)} = & \bm{L}_{1i}^{(t)} - \eta\cdot(\bm{C}^{(t)\top}\bm{C}^{(t)})^{-1/2}\bm{G}_{1i}^{(t)}(\bm{L}_{2i}^{(t)\top}\bm{R}^{(t)\top}\bm{R}^{(t)}\bm{L}_{2i}^{(t)})^{-1/2}\\
            = & \bm{L}_{1i}^{(t)} - \eta\cdot(\bm{C}^{(t)\top}\bm{C}^{(t)})^{-1}\mathbb{E}[\nabla_{\bm{L}_{1i}}\mathcal{L}_n(\bm{\Theta}^{(t)})](\bm{L}_{2i}^{(t)\top}\bm{R}^{(t)\top}\bm{R}^{(t)}\bm{L}_{2i}^{(t)})^{-1} \\
            & - \eta\cdot(\bm{C}^{(t)\top}\bm{C}^{(t)})^{-1/2}\bm{\Delta}_{1i}^{(t)}(\bm{L}_{2i}^{(t)\top}\bm{R}^{(t)\top}\bm{R}^{(t)}\bm{L}_{2i}^{(t)})^{-1/2},\\
            \bm{L}_{2i}^{(t+1)} = & \bm{L}_{2i}^{(t)} - \eta\cdot(\bm{R}^{(t)\top}\bm{R}^{(t)})^{-1/2}\bm{G}_{2i}^{(t)}(\bm{L}_{1i}^{(t)\top}\bm{C}^{(t)\top}\bm{C}^{(t)}\bm{L}_{1i}^{(t)})^{-1/2}\\
            = & \bm{L}_{2i}^{(t)} - \eta\cdot(\bm{R}^{(t)\top}\bm{R}^{(t)})^{-1}\mathbb{E}[\nabla_{\bm{L}_{2i}}\mathcal{L}_n(\bm{\Theta}^{(t)})](\bm{L}_{1i}^{(t)\top}\bm{C}^{(t)\top}\bm{C}^{(t)}\bm{L}_{1i}^{(t)})^{-1} \\
            & - \eta\cdot(\bm{R}^{(t)\top}\bm{R}^{(t)})^{-1/2}\bm{\Delta}_{2i}^{(t)}(\bm{L}_{1i}^{(t)\top}\bm{C}^{(t)\top}\bm{C}^{(t)}\bm{L}_{1i}^{(t)})^{-1/2}.
        \end{split}
    \end{equation}

    Define the optimal transformations
    \begin{equation}
        \begin{split}
            & (\bm{Q}_{1t},\bm{Q}_{2t},\bm{P}_{1t},\dots,\bm{P}_{nt})\\
            = & \inf_{\bm{Q}_1,\bm{Q}_2,\bm{P}_1,\dots,\bm{P}_n}\Bigg\{ \|(\bm{C}^{(t)}\bm{Q}_1-\bm{C}^*)\bm{\Sigma}_{C}\|_\text{F}^2 + \|(\bm{R}^{(t)}\bm{Q}_2-\bm{R}^*)\bm{\Sigma}_{R}\|_\text{F}^2 \\
            & + \sum_{i=1}^n\|(\bm{Q}_1^{-1}\bm{L}_{1i}^{(t)}\bm{P}_i - \bm{L}_{1i}^*)\bm{\Sigma}_i^{1/2}\|_\text{F}^2 + \sum_{i=1}^n\|(\bm{Q}_2^{-1}\bm{L}_{2i}^{(t)}\bm{P}_i^{-\top} - \bm{L}_{2i}^*)\bm{\Sigma}_i^{1/2}\|_\text{F}^2 \Bigg\}
        \end{split}
    \end{equation}
    to align $(\bm{C}^{(t)},\bm{R}^{(t)},\bm{L}_{11}^{(t)},\bm{L}_{21}^{(t)},\dots,\bm{L}_{1n}^{(t)},\bm{L}_{2n}^{(t)})$ and $(\bm{C}^{*},\bm{R}^{*},\bm{L}_{11}^{*},\bm{L}_{21}^{*},\dots,\bm{L}_{1n}^{*},\bm{L}_{2n}^{*})$. 

    For simplicity, as we focus on the generic step $t$, denote the decomposition components with these optimal alignments as 
    $\bm{C}=\bm{C}^{(t)}\bm{Q}_{1t}$, 
    $\widetilde{\bm{C}}=\widetilde{\bm{C}}^{(t)}\bm{Q}_{1t}^{-\top}$, 
    $\bm{R}=\bm{R}^{(t)}\bm{Q}_{2t}$, 
    $\widetilde{\bm{R}}=\widetilde{\bm{R}}^{(t)}\bm{Q}_{2t}^{-\top}$, 
    $\cm{G}=\cm{G}^{(t)}\times_1\bm{Q}_{1t}^{\top}\times_2\bm{Q}_{2t}^{\top}$, 
    $\bm{L}_{1i}=\bm{Q}_{1t}^{-1}\bm{L}_{1i}^{(t)}\bm{P}_{it}$, 
    $\bm{L}_{2i}=\bm{Q}_{2t}^{-1}\bm{L}_{2i}^{(t)}\bm{P}_{it}^{-\top}$, 
    $\bm{R}_i=\bm{R}\bm{L}_{2i}$, and 
    $\bm{C}_i=\bm{C}\bm{L}_{2i}$. 
    Thus, we have $(\bm{C}^{(t)\top}\bm{C}^{(t)})^{-1}=\bm{Q}_{1t}(\bm{C}^\top\bm{C})^{-1}\bm{Q}_{1t}^{\top}$, $(\widetilde{\bm{C}}^{(t)\top}\widetilde{\bm{C}}^{(t)})^{-1}=\bm{Q}_{1t}^{-\top}(\widetilde{\bm{C}}^\top\widetilde{\bm{C}})^{-1}\bm{Q}_{1t}^{-1}$, $(\bm{R}^{(t)\top}\bm{R}^{(t)})^{-1}=\bm{Q}_{2t}(\bm{R}^\top\bm{R})^{-1}\bm{Q}_{2t}^{\top}$, $(\widetilde{\bm{R}}^{(t)\top}\widetilde{\bm{R}}^{(t)})^{-1}=\bm{Q}_{2t}^{-\top}(\widetilde{\bm{R}}^\top\widetilde{\bm{R}})^{-1}\bm{Q}_{2t}^{-1}$, $(\bm{L}_{2i}^{(t)\top}\bm{R}^{(t)\top}\bm{R}^{(t)}\bm{L}_{2i}^{(t)})^{-1}=\bm{P}_{it}^{-\top}(\bm{R}_i^\top\bm{R}_i)^{-1}\bm{P}_{it}^{-1}$, $(\bm{L}_{1i}^{(t)\top}\bm{C}^{(t)\top}\bm{C}^{(t)}\bm{L}_{1i}^{(t)})^{-1}=\bm{P}_{it}(\bm{C}_i^\top\bm{C}_i)^{-1}\bm{P}_{it}^{\top}$.

    Then, the proposed scaled gradient descent can be rewritten as
    \begin{equation}
        \begin{split}
            \bm{C}^{(t+1)} = & \bm{C}^{(t)} - \eta\cdot\bm{H}_C^{(t)}\widetilde{\bm{C}}^{(t)}(\widetilde{\bm{C}}^{(t)\top}\widetilde{\bm{C}}^{(t)})^{-1} - \eta\cdot\bm{\Delta}_C^{(t)}(\widetilde{\bm{C}}^{(t)\top}\widetilde{\bm{C}}^{(t)})^{-1/2}\\
            = & \bm{C}^{(t)} - \eta\cdot\bm{H}_C^{(t)}\widetilde{\bm{C}}(\widetilde{\bm{C}}^{\top}\widetilde{\bm{C}})^{-1}\bm{Q}_{1t}^{-1} - \eta\cdot\bm{\Delta}_C^{(t)}(\widetilde{\bm{C}}^{(t)\top}\widetilde{\bm{C}}^{(t)})^{-1/2},
        \end{split}
    \end{equation}
    \begin{equation}
        \begin{split}
            \bm{R}^{(t+1)} = & \bm{R}^{(t)} - \eta\cdot\bm{H}_R^{(t)}\widetilde{\bm{R}}^{(t)}(\widetilde{\bm{R}}^{(t)\top}\widetilde{\bm{R}}^{(t)})^{-1} - \eta\cdot\bm{\Delta}_R^{(t)}(\widetilde{\bm{R}}^{(t)\top}\widetilde{\bm{R}}^{(t)})^{-1/2}\\
            = & \bm{R}^{(t)} - \eta\cdot\bm{H}_R^{(t)}\widetilde{\bm{R}}(\widetilde{\bm{R}}^{\top}\widetilde{\bm{R}})^{-1}\bm{Q}_{2t}^{-1} - \eta\cdot\bm{\Delta}_R^{(t)}(\widetilde{\bm{R}}^{(t)\top}\widetilde{\bm{R}}^{(t)})^{-1/2},
        \end{split}
    \end{equation}
    \begin{equation}
        \begin{split}
            \bm{L}_{1i}^{(t+1)} = & \bm{L}_{1i}^{(t)} - \eta\cdot(\bm{C}^{(t)\top}\bm{C}^{(t)})^{-1}\bm{C}^{(t)\top}\bm{H}_i^{(t)}\bm{R}^{(t)}\bm{L}_{2i}^{(t)}(\bm{L}_{2i}^{(t)\top}\bm{R}^{(t)\top}\bm{R}^{(t)}\bm{L}_{2i}^{(t)})^{-1} \\
            & - \eta\cdot(\bm{C}^{(t)\top}\bm{C}^{(t)})^{-1/2}\bm{\Delta}_{1i}^{(t)}(\bm{L}_{2i}^{(t)\top}\bm{R}^{(t)\top}\bm{R}^{(t)}\bm{L}_{2i}^{(t)})^{-1/2}\\
            = & \bm{L}_{1i}^{(t)} - \eta\cdot\bm{Q}_{1t}(\bm{C}^{\top}\bm{C})^{-1}\bm{C}^{\top}\bm{H}_i^{(t)}\bm{R}_i(\bm{R}_i^{\top}\bm{R}_i)^{-1}\bm{P}_{it}^{-1}\\
            & - \eta\cdot(\bm{C}^{(t)\top}\bm{C}^{(t)})^{-1/2}\bm{\Delta}_{1i}^{(t)}(\bm{L}_{2i}^{(t)\top}\bm{R}^{(t)\top}\bm{R}^{(t)}\bm{L}_{2i}^{(t)})^{-1/2},\\
        \end{split}
    \end{equation}
    and
    \begin{equation}
        \begin{split}
            \bm{L}_{2i}^{(t+1)} = & \bm{L}_{2i}^{(t)} - \eta\cdot(\bm{R}^{(t)\top}\bm{R}^{(t)})^{-1}\bm{R}^{(t)\top}\bm{H}_i^{(t)\top}\bm{C}^{(t)}\bm{L}_{1i}^{(t)}(\bm{L}_{1i}^{(t)\top}\bm{C}^{(t)\top}\bm{C}^{(t)}\bm{L}_{1i}^{(t)})^{-1} \\
            & - \eta\cdot(\bm{R}^{(t)\top}\bm{R}^{(t)})^{-1/2}\bm{\Delta}_{2i}^{(t)}(\bm{L}_{1i}^{(t)\top}\bm{C}^{(t)\top}\bm{C}^{(t)}\bm{L}_{1i}^{(t)})^{-1/2}\\
            = & \bm{L}_{2i}^{(t)} - \eta\cdot\bm{Q}_{2t}(\bm{R}^{\top}\bm{R})^{-1}\bm{R}^{\top}\bm{H}_i^{(t)\top}\bm{C}_i(\bm{C}_i^{\top}\bm{C}_i)^{-1}\bm{P}_{it}^{\top}\\
            & - \eta\cdot(\bm{R}^{(t)\top}\bm{R}^{(t)})^{-1/2}\bm{\Delta}_{1i}^{(t)}(\bm{L}_{1i}^{(t)\top}\bm{C}^{(t)\top}\bm{C}^{(t)}\bm{L}_{1i}^{(t)})^{-1/2},
        \end{split}
    \end{equation}
    where
    \begin{equation}
        \begin{split}
            \bm{H}_C^{(t)} & = \sum_{i=1}^nm_i\cdot\mathbb{E}\left[\bm{e}_n(n)^\top\otimes\nabla\mathcal{L}(\bm{B}_i^{(t)};\bm{X}_{ij},y_{ij})\right],\\
            \bm{H}_R^{(t)} & = \sum_{i=1}^nm_i\cdot\mathbb{E}\left[\bm{e}_n(n)^\top\otimes\nabla\mathcal{L}(\bm{B}_i^{(t)};\bm{X}_{ij},y_{ij})^\top\right],\\
            \text{and  }\bm{H}_i^{(t)} & = m_i\cdot\mathbb{E}[\nabla\mathcal{L}(\bm{B}_i^{(t)};\bm{X}_{ij},y_{ij})],
        \end{split}
    \end{equation}
    such that $\|\bm{H}_C^{(t)}\|_\text{F}^2=\|\bm{H}_R^{(t)}\|_\text{F}^2=\sum_{i=1}^n\|\bm{H}_i^{(t)}\|_\text{F}^2$.

    Next, we present some necessary conditions. First, the de-scaled partial gradients satisfy the stability conditions
    \begin{equation}
        \begin{split}
            \|\bm{\Delta}_C^{(t)}\|_\text{F}^2 & \leq \phi\|\cm{B}^{(t)} - \cm{B}^*\|_\text{F}^2 + \xi_C,~~\|\bm{\Delta}_R^{(t)}\|_\text{F}^2 \leq \phi\|\cm{B}^{(t)} - \cm{B}^*\|_\text{F}^2 + \xi_R, \\
            \|\bm{\Delta}_{1i}^{(t)}\|_\text{F}^2 & \leq \phi\|\bm{B}_i^{(t)} - \bm{B}_i^*\|_\text{F}^2 + \xi_{1i}, ~~\|\bm{\Delta}_{2i}^{(t)}\|_\text{F}^2 \leq \phi\|\bm{B}_i^{(t)} - \bm{B}_i^*\|_\text{F}^2 + \xi_{2i} .
        \end{split}
    \end{equation}
    Second, we assume that for all $t=0,1,2,\dots,T$,
    \begin{equation}\label{eq:cond1}
        \text{dist}(\bm{\Theta}^{(t)},\bm{\Theta}^*)\leq \epsilon\underline{\sigma}
    \end{equation}
    for some sufficiently small $\epsilon$. Then, by Lemma \ref{lemma:matrix_perturb}, we have
    \begin{equation}
        \|\cm{G}^{(t)}\times_1\bm{Q}_1^{-1}\times_2\bm{Q}_2^{-1}-\cm{G}^*\|_\text{F}^2 = \|\cm{G}-\cm{G}^*\|_\text{F}^2 \leq (1+\sqrt{2})\epsilon^2\underline{\sigma}^2
    \end{equation}
    and
    \begin{equation}
        \|(\bm{C}-\bm{C}^*)\bm{\Sigma}_C\|_\text{F}^2 + \|(\bm{R}-\bm{R}^*)\bm{\Sigma}_R\|_\text{F}^2 + \|\cm{G}-\cm{G}^*\|_\text{F}^2 \leq (2+\sqrt{2})\epsilon^2\underline{\sigma}^2.
    \end{equation}
    Further, by Lemma \ref{lemma:tensor_perturb}, the following upper bounds hold for some small $C$,
    \begin{equation}\label{eq:perturb_upper_bound1}
        \begin{split}
            \max(\|\bm{C}-\bm{C}^*\|,\|\bm{R}-\bm{R}^*\|,\|(\cm{G}-\cm{G}^*)_{(1)}\bm{\Sigma}_C^{-1}\|,\|(\cm{G}-\cm{G}^*)_{(2)}\bm{\Sigma}_R^{-1}\|) & \leq C\epsilon,\\
            \|\bm{C}(\bm{C}^\top\bm{C})^{-1}\| \leq 1+C\epsilon,~~~
            \|\bm{C}(\bm{C}^\top\bm{C})^{-1}-\bm{C}^*\| & \leq C\epsilon,\\
            \|(\widetilde{\bm{C}}-\widetilde{\bm{C}}^*)\bm{\Sigma}_C^{-1}\| \leq C\epsilon,~~~
            \|\widetilde{\bm{C}}(\widetilde{\bm{C}}^\top\widetilde{\bm{C}})^{-1}\bm{\Sigma}_C\| & \leq 1+C\epsilon,\\
            \|\widetilde{\bm{C}}(\widetilde{\bm{C}}^\top\widetilde{\bm{C}})^{-1}\bm{\Sigma}_C - \widetilde{\bm{C}}\bm{\Sigma}_C^{-1}\| \leq C\epsilon,~~~
            \|\bm{\Sigma}_C(\widetilde{\bm{C}}^\top\widetilde{\bm{C}})^{-1}\bm{\Sigma}_C\| & \leq 1+C\epsilon.
        \end{split}
    \end{equation}
    By symmetry, the corresponding upper bounds for $\bm{R}$ and $\widetilde{\bm{R}}$ hold as well. Similarly, by Lemma \ref{lemma:matrix_perturb}, we also have
    \begin{equation}\label{eq:perturb_upper_bound2}
        \begin{split}
            \max_{1\leq i\leq n}(\|(\bm{C}_i-\bm{C}_i^*)\bm{\Sigma}_i^{1/2}\|,\|(\bm{R}_i-\bm{R}_i^*)\bm{\Sigma}_i^{1/2}\|) & \leq C\epsilon,\\
            \|(\bm{C}_i^*(\bm{C}_i^{*\top}\bm{C}_i^*)^{-1}-\bm{C}_i(\bm{C}_i^\top\bm{C}_i)^{-1})\bm{\Sigma}_i^{1/2}\| & \leq C\epsilon,\\
            \|(\bm{R}_i^*(\bm{R}_i^{*\top}\bm{R}_i^*)^{-1}-\bm{R}_i(\bm{R}_i^\top\bm{R}_i)^{-1})\bm{\Sigma}_i^{1/2}\| & \leq C\epsilon.
        \end{split}
    \end{equation}~

    \noindent\textit{Step 2.} (Upper bound for $\text{dist}(\bm{\Theta}^{(t+1)},\bm{\Theta}^*)$)

    \noindent By definition of $\bm{Q}_{1t}$, $\bm{Q}_{2t}$, and $\bm{P}_{it}$'s,
    \begin{equation}
        \begin{split}
            & \text{dist}(\bm{\Theta}^{(t+1)},\bm{\Theta}^*)\\
            = & \|(\bm{C}^{(t+1)}\bm{Q}_{1,t+1}-\bm{C}^*)\bm{\Sigma}_C\|_\text{F}^2 + \|(\bm{R}^{(t+1)}\bm{Q}_{2,t+1}-\bm{R}^*)\bm{\Sigma}_R\|_\text{F}^2 \\
            & + \sum_{i=1}^n\|(\bm{Q}_{1,t+1}^{-1}\bm{L}_{1i}^{(t+1)}\bm{P}_{i,t+1} - \bm{L}_{1i}^*)\bm{\Sigma}_i^{1/2}\|_\text{F}^2 + \sum_{i=1}^n\|(\bm{Q}_{2,t+1}^{-1}\bm{L}_{2i}^{(t+1)}\bm{P}_{i,t+1}^{-\top} - \bm{L}_{2i}^*)\bm{\Sigma}_i^{1/2}\|_\text{F}^2\\
            \leq & \|(\bm{C}^{(t+1)}\bm{Q}_{1t}-\bm{C}^*)\bm{\Sigma}_C\|_\text{F}^2 + \|(\bm{R}^{(t+1)}\bm{Q}_{2t}-\bm{R}^*)\bm{\Sigma}_R\|_\text{F}^2 \\
            & + \sum_{i=1}^n\|(\bm{Q}_{1t}^{-1}\bm{L}_{1i}^{(t+1)}\bm{P}_{it} - \bm{L}_{1i}^*)\bm{\Sigma}_i^{1/2}\|_\text{F}^2 + \sum_{i=1}^n\|(\bm{Q}_{2t}^{-1}\bm{L}_{2i}^{(t+1)}\bm{P}_{it}^{-\top} - \bm{L}_{2i}^*)\bm{\Sigma}_i^{1/2}\|_\text{F}^2.
        \end{split}
    \end{equation}
    Then, we will derive the upper bounds for the terms of $\bm{C}$, $\bm{R}$, $\bm{L}_{1i}$'s and $\bm{L}_{2i}$'s, respectively.\\

    \noindent{\textit{Step 2.1}} ($\bm{C}$ and $\bm{R}$ upper bounds)

    \noindent We first develop an upper bound for $\|(\bm{C}^{(t+1)}\bm{Q}_{1t}-\bm{C}^*)\bm{\Sigma}_C\|_\text{F}$. By mean inequality and the upper bounds in \eqref{eq:perturb_upper_bound1}, for any $\zeta>0$,
    \begin{equation}
        \begin{split}
            & \|(\bm{C}^{(t+1)}\bm{Q}_{1t}-\bm{C}^*)\bm{\Sigma}_C\|_\text{F}^2\\
            = & \|(\bm{C} - \eta\cdot\bm{H}_C^{(t)}\widetilde{\bm{C}}(\widetilde{\bm{C}}^{\top}\widetilde{\bm{C}})^{-1} - \eta\cdot\bm{\Delta}_C^{(t)}(\widetilde{\bm{C}}^{(t)\top}\widetilde{\bm{C}}^{(t)})^{-1/2}\bm{Q}_{1t}-\bm{C}^*)\bm{\Sigma}_C\|_\text{F}^2\\
            \leq & (1+\zeta)\|(\bm{C}-\bm{C}^*-\eta\cdot\bm{H}_C^{(t)}\widetilde{\bm{C}}(\widetilde{\bm{C}}^{\top}\widetilde{\bm{C}})^{-1})\bm{\Sigma}_C\|_\text{F}^2\\
            & + (1+\zeta^{-1})\eta^2\|\bm{\Delta}_C^{(t)}(\widetilde{\bm{C}}^{(t)\top}\widetilde{\bm{C}}^{(t)})^{-1/2}\bm{Q}_{1t}\bm{\Sigma}_C\|_\text{F}^2\\
            \leq & (1+\zeta)\|(\bm{C}-\bm{C}^*-\eta\cdot\bm{H}_C^{(t)}\widetilde{\bm{C}}(\widetilde{\bm{C}}^{\top}\widetilde{\bm{C}})^{-1})\bm{\Sigma}_C\|_\text{F}^2\\
            & + (1+\zeta^{-1})\eta^2\|\bm{\Delta}_C^{(t)}\|_\text{F}^2\cdot\|\bm{\Sigma}_C\bm{Q}_{1t}^\top(\widetilde{\bm{C}}^{(t)\top}\widetilde{\bm{C}}^{(t)})^{-1}\bm{Q}_{1t}\bm{\Sigma}_C\|\\
            \leq & (1+\zeta)\|(\bm{C}-\bm{C}^*-\eta\cdot\bm{H}_C^{(t)}\widetilde{\bm{C}}(\widetilde{\bm{C}}^{\top}\widetilde{\bm{C}})^{-1})\bm{\Sigma}_C\|_\text{F}^2+ C(1+\zeta^{-1})\eta^2\|\bm{\Delta}_C^{(t)}\|_\text{F}^2.
        \end{split}
    \end{equation}
    For any $\gamma>0$, the first term can be rewritten as
    \begin{equation}
        \begin{split}
            & \|(\bm{C}-\bm{C}^*-\eta\cdot\bm{H}_C^{(t)}\widetilde{\bm{C}}(\widetilde{\bm{C}}^{\top}\widetilde{\bm{C}})^{-1})\bm{\Sigma}_C\|_\text{F}^2\\
            = & \|(\bm{C}-\bm{C}^*)\bm{\Sigma}_C\|_\text{F}^2 + \eta^2\cdot\|\bm{H}_C^{(t)}\widetilde{\bm{C}}(\widetilde{\bm{C}}^{\top}\widetilde{\bm{C}})^{-1}\bm{\Sigma}_C\|_\text{F}^2 \\
            & - 2\eta\langle (\bm{C}-\bm{C}^*)\bm{\Sigma}_C, \bm{H}_C^{(t)}\widetilde{\bm{C}}(\widetilde{\bm{C}}^{\top}\widetilde{\bm{C}})^{-1}\bm{\Sigma}_C\rangle\\
            = & \|(\bm{C}-\bm{C}^*)\bm{\Sigma}_C\|_\text{F}^2 + \eta^2\cdot\|\bm{H}_C^{(t)}\widetilde{\bm{C}}(\widetilde{\bm{C}}^{\top}\widetilde{\bm{C}})^{-1}\bm{\Sigma}_C\|_\text{F}^2 \\
            & - 2\eta\langle (\bm{C}-\bm{C}^*)\widetilde{\bm{C}}^{*\top}, \bm{H}_C^{(t)}\rangle\\
            & + 2\eta\langle (\bm{C}-\bm{C}^*)\bm{\Sigma}_C, \bm{H}_C^{(t)}(\widetilde{\bm{C}}^*(\widetilde{\bm{C}}^{*\top}\widetilde{\bm{C}}^*)^{-1/2} - \widetilde{\bm{C}}(\widetilde{\bm{C}}^{\top}\widetilde{\bm{C}})^{-1}\bm{\Sigma}_C)\rangle\\
            \leq & \|(\bm{C}-\bm{C}^*)\bm{\Sigma}_C\|_\text{F}^2 + \eta^2\cdot\|\bm{H}_C^{(t)}\|_\text{F}^2\cdot\|\widetilde{\bm{C}}(\widetilde{\bm{C}}^{\top}\widetilde{\bm{C}})^{-1}\bm{\Sigma}_C\|^2-2\eta\cdot M_C \\
            & + 2\eta\cdot\|(\bm{C}-\bm{C}^*)\bm{\Sigma}_C\|_\text{F}\cdot\|\bm{H}_C^{(t)}\|_\text{F}\cdot\|\widetilde{\bm{C}}^*(\widetilde{\bm{C}}^{*\top}\widetilde{\bm{C}}^*)^{-1/2} - \widetilde{\bm{C}}(\widetilde{\bm{C}}^{\top}\widetilde{\bm{C}})^{-1}\bm{\Sigma}_C\|\\
            \leq & \|(\bm{C}-\bm{C}^*)\bm{\Sigma}_C\|_\text{F}^2 + C\eta^2\cdot\|\bm{H}_C^{(t)}\|_\text{F}^2 - 2\eta\cdot M_C\\
            & + C\epsilon\eta\gamma^{-1}\cdot\|\bm{H}_C^{(t)}\|_\text{F}^2 + C\epsilon\eta \gamma\cdot\|(\bm{C}-\bm{C}^*)\bm{\Sigma}_C\|_\text{F}^2,
        \end{split}
    \end{equation}
    where $M_C=\langle(\bm{C}-\bm{C}^*)\widetilde{\bm{C}}^{*\top},\bm{H}_C^{(t)}\rangle$.

    Similarly, for $\bm{R}$, we have
    \begin{equation}
        \begin{split}
            & \|(\bm{R}^{(t+1)}\bm{Q}_{2t}-\bm{R}^*)\bm{\Sigma}_R\|_\text{F}^2\\
            \leq & (1+\zeta)\|(\bm{R}-\bm{R}^*-\eta\cdot\mathbb{E}[\nabla\overline{\mathcal{L}}_n(\cm{B}^{(t)})]_{(2)}\widetilde{\bm{R}}(\widetilde{\bm{R}}^\top\widetilde{\bm{R}})^{-1})\bm{\Sigma}_R\|_\text{F}^2 \\
            & + (1+\zeta^{-1})\eta^2\|\bm{\Delta}_R^{(t)}(\widetilde{\bm{R}}^{(t)\top}\widetilde{\bm{R}}^{(t)})^{-1/2}\bm{Q}_{2t}\bm{\Sigma}_R\|_\text{F}^2\\
            \leq & (1+\zeta)\|(\bm{R}-\bm{R}^*-\eta\cdot\mathbb{E}[\nabla\overline{\mathcal{L}}_n(\cm{B}^{(t)})]_{(2)}\widetilde{\bm{R}}(\widetilde{\bm{R}}^\top\widetilde{\bm{R}})^{-1})\bm{\Sigma}_R\|_\text{F}^2  + C(1+\zeta^{-1})\eta^2\|\bm{\Delta}_R^{(t)}\|_\text{F}^2.
        \end{split}
    \end{equation}
    For any $\gamma>0$, the first term can be further bounded by
    \begin{equation}
        \begin{split}
            & \|(\bm{R}-\bm{R}^*-\eta\cdot\bm{H}_R^{(t)}\widetilde{\bm{R}}(\widetilde{\bm{R}}^\top\widetilde{\bm{R}})^{-1})\bm{\Sigma}_R\|_\text{F}^2\\
            \leq & \|(\bm{R}-\bm{R}^*)\bm{\Sigma}_2^*\|_\text{F}^2 + C_1\eta^2\|\bm{H}_R^{(t)}\|_\text{F}^2 - 2\eta\cdot M_R\\
            & + C\epsilon\eta\gamma^{-1}\cdot\|\bm{H}_R^{(t)}\|_\text{F}^2 + C\epsilon\eta\gamma\cdot\|(\bm{R}-\bm{R}^*)\bm{\Sigma}_R\|_\text{F}^2,
        \end{split}
    \end{equation}
    where $M_R=\langle(\bm{R}-\bm{R}^*)\widetilde{\bm{R}}^{*\top},\bm{H}_R^{(t)}\rangle$.\\

    \noindent\textit{Step 2.2} (Upper bounds for $\bm{L}_{1i}$ and $\bm{L}_{2i}$)

    \noindent Next, we focus on the terms related to $\bm{L}_{1i}$ and $\bm{L}_{2i}$. For any $\zeta>0$,
    \begin{equation}
        \begin{split}
            & \|(\bm{Q}_{1t}^{-1}\bm{L}_{1i}^{(t+1)}\bm{P}_{it} - \bm{L}_{1i}^*)\bm{\Sigma}_i^{1/2}\|_\text{F}^2 \\
            = & \|(\bm{Q}_{1t}^{-1}[\bm{L}_{1i}^{(t)} - \eta\cdot\bm{Q}_{1t}(\bm{C}^{\top}\bm{C})^{-1}\bm{C}^{\top}\mathbb{E}[\nabla\mathcal{L}(\bm{B}_i^{(t)})]\bm{R}_i(\bm{R}_i^{\top}\bm{R}_i)^{-1}\bm{P}_{it}^{-1}]\bm{P}_{it} \\ 
            & -\eta\bm{Q}_{1t}^{-1}(\bm{C}^{(t)\top}\bm{C}^{(t)})^{-1/2}\bm{\Delta}_{1i}^{(t)}(\bm{L}_{2i}^{(t)\top}\bm{R}^{(t)\top}\bm{R}^{(t)}\bm{L}_{2i}^{(t)})^{-1/2}\bm{P}_{it} - \bm{L}_{1i}^*)\bm{\Sigma}_i^{1/2}\|_\text{F}^2 \\
            \leq & (1+\zeta)\|(\bm{L}_{1i}-\bm{L}_{1i}^*-\eta(\bm{C}^{\top}\bm{C})^{-1}\bm{C}^{\top}\bm{H}_i^{(t)}\bm{R}_i(\bm{R}_i^{\top}\bm{R}_i)^{-1})\bm{\Sigma}_i^{1/2}\|_\text{F}^2\\
            & + (1+\zeta^{-1})\eta^2\|\bm{Q}_{1t}^{-1}(\bm{C}^{(t)\top}\bm{C}^{(t)})^{-1/2}\bm{\Delta}_{1i}^{(t)}(\bm{L}_{2i}^{(t)\top}\bm{R}^{(t)\top}\bm{R}^{(t)}\bm{L}_{2i}^{(t)})^{-1/2}\bm{P}_{it}\bm{\Sigma}_i^{1/2}\|_\text{F}^2\\
            \leq & (1+\zeta)\|(\bm{L}_{1i}-\bm{L}_{1i}^*-\eta(\bm{C}^{\top}\bm{C})^{-1}\bm{C}^{\top}\bm{H}_i^{(t)}\bm{R}_i(\bm{R}_i^{\top}\bm{R}_i)^{-1})\bm{\Sigma}_i^{1/2}\|_\text{F}^2 + C(1+\zeta^{-1})\eta^2\|\bm{\Delta}_{1i}^{(t)}\|_\text{F}^2.
        \end{split}
    \end{equation}
    For any $\gamma>0$, by \eqref{eq:perturb_upper_bound1} and \eqref{eq:perturb_upper_bound2}, the first term can be bounded by
    \begin{equation}
        \begin{split}
            & \|(\bm{L}_{1i}-\bm{L}_{1i}^*-\eta(\bm{C}^{\top}\bm{C})^{-1}\bm{C}^{\top}\bm{H}_i^{(t)}\bm{R}_i(\bm{R}_i^{\top}\bm{R}_i)^{-1})\bm{\Sigma}_i^{1/2}\|_\text{F}^2\\
            = & \|(\bm{L}_{1i}-\bm{L}_{1i}^*)\bm{\Sigma}_i^{1/2}\|_\text{F}^2 + \eta^2\|(\bm{C}^{\top}\bm{C})^{-1}\bm{C}^{\top}\bm{H}_i^{(t)}\bm{R}_i(\bm{R}_i^{\top}\bm{R}_i)^{-1}\bm{\Sigma}_i^{1/2}\|_\text{F}^2\\
            & -2\eta\langle(\bm{L}_{1i}-\bm{L}_{1i}^*)\bm{\Sigma}_i^{1/2}, (\bm{C}^{\top}\bm{C})^{-1}\bm{C}^{\top}\bm{H}_i^{(t)}\bm{R}_i(\bm{R}_i^{\top}\bm{R}_i)^{-1}\bm{\Sigma}_i^{1/2}\rangle\\
            \leq & \|(\bm{L}_{1i}-\bm{L}_{1i}^*)\bm{\Sigma}_i^{1/2}\|_\text{F}^2 + \eta^2\|(\bm{C}^{\top}\bm{C})^{-1}\bm{C}^{\top}\bm{H}_i^{(t)}\bm{R}_i(\bm{R}_i^{\top}\bm{R}_i)^{-1}\bm{\Sigma}_i^{1/2}\|_\text{F}^2\\
            & -2\eta\langle(\bm{L}_{1i}-\bm{L}_{1i}^*)\bm{\Sigma}_i^{1/2}, \bm{C}^{*\top}\bm{H}_i^{(t)}\bm{R}_i^*(\bm{R}_i^{*\top}\bm{R}_i^*)^{-1}\bm{\Sigma}_i^{1/2}\rangle\\
            & +2\eta\langle(\bm{L}_{1i}-\bm{L}_{1i}^*)\bm{\Sigma}_i^{1/2}, (\bm{C}^{*\top}-(\bm{C}^\top\bm{C})^{-1}\bm{C}^\top)\bm{H}_i^{(t)}\bm{R}_i^*(\bm{R}_i^{*\top}\bm{R}_i^*)^{-1}\bm{\Sigma}_i^{1/2}\rangle\\
            & +2\eta\langle(\bm{L}_{1i}-\bm{L}_{1i}^*)\bm{\Sigma}_i^{1/2}, \bm{C}^{*\top}\bm{H}_i^{(t)}(\bm{R}_i^*(\bm{R}_i^{*\top}\bm{R}_i^*)^{-1}-\bm{R}_i(\bm{R}_i^\top\bm{R}_i)^{-1})\bm{\Sigma}_i^{1/2}\rangle\\
            &-2\eta\langle(\bm{L}_{1i}-\bm{L}_{1i}^*)\bm{\Sigma}_i^{1/2}, (\bm{C}^{*\top}-(\bm{C}^\top\bm{C})^{-1}\bm{C}^\top)\bm{H}_i^{(t)}(\bm{R}_i^*(\bm{R}_i^{*\top}\bm{R}_i^*)^{-1}-\bm{R}_i(\bm{R}_i^\top\bm{R}_i)^{-1})\bm{\Sigma}_i^{1/2}\rangle\\
            \leq & \|(\bm{L}_{1i}-\bm{L}_{1i}^*)\bm{\Sigma}_i^{1/2}\|_\text{F}^2 +C \eta^2\|\bm{H}_i^{(t)}\|_\text{F}^2 - 2\eta M_{1i}\\
            & +C\epsilon\eta\gamma^{-1}\cdot\|\bm{H}_i^{(t)}\|_\text{F}^2 + C\epsilon\eta\gamma\cdot\|(\bm{L}_{1i}-\bm{L}_{1i}^*)\bm{\Sigma}_i^{1/2}\|_\text{F}^2,
        \end{split}
    \end{equation}
    where $M_{1i}=\langle(\bm{L}_{1i}-\bm{L}_{1i}^*)\bm{\Sigma}_i^{1/2}, \bm{C}^{*\top}\bm{H}_i^{(t)}\bm{R}_i^*(\bm{R}_i^{*\top}\bm{R}_i^*)^{-1}\bm{\Sigma}_i^{1/2}\rangle$.

    Similarly, for any $\zeta>0$, we can have the bound for $\bm{L}_{2i}$
    \begin{equation}
        \begin{split}
            & \|(\bm{Q}_{2t}^{-1}\bm{L}_{2i}^{(t+1)}\bm{P}_{it}^{-\top} - \bm{L}_{2i}^*)\bm{\Sigma}_i^{1/2}\|_\text{F}^2 \\
            = & \|(\bm{Q}_{2t}^{-1}[\bm{L}_{2i}^{(t)} - \eta\cdot\bm{Q}_{2t}(\bm{R}^{\top}\bm{R})^{-1}\bm{R}^{\top}\bm{H}_i^{(t)\top}\bm{C}_i(\bm{C}_i^{\top}\bm{C}_i)^{-1}\bm{P}_{it}^{\top}]\bm{P}_{it}^{-\top} \\ 
            & -\eta\bm{Q}_{2t}^{-1}(\bm{R}^{(t)\top}\bm{R}^{(t)})^{-1/2}\bm{\Delta}_{2i}^{(t)}(\bm{L}_{1i}^{(t)\top}\bm{C}^{(t)\top}\bm{C}^{(t)}\bm{L}_{1i}^{(t)})^{-1/2}\bm{P}_{it}^{-\top} - \bm{L}_{2i}^*)\bm{\Sigma}_i^{1/2}\|_\text{F}^2 \\
            \leq & (1+\zeta)\|(\bm{L}_{2i}-\bm{L}_{2i}^*-\eta(\bm{R}^{\top}\bm{R})^{-1}\bm{C}^{\top}\bm{H}_i^{(t)\top}\bm{C}_i(\bm{C}_i^{\top}\bm{C}_i)^{-1})\bm{\Sigma}_i^{1/2}\|_\text{F}^2\\
            & + (1+\zeta^{-1})C\eta^2\|\bm{\Delta}_{2i}^{(t)}\|_\text{F}^2.
        \end{split}
    \end{equation}
    For any $\gamma>0$, the first term can be bounded by
    \begin{equation}
        \begin{split}
            & \|(\bm{L}_{2i}-\bm{L}_{2i}^*-\eta(\bm{R}^{\top}\bm{R})^{-1}\bm{C}^{\top}\bm{H}_i^{(t)\top}\bm{C}_i(\bm{C}_i^{\top}\bm{C}_i)^{-1})\bm{\Sigma}_i^{1/2}\|_\text{F}^2\\
            = & \|(\bm{L}_{2i}-\bm{L}_{2i}^*)\bm{\Sigma}_i^{1/2}\|_\text{F}^2 + \eta^2\|(\bm{R}^{\top}\bm{R})^{-1}\bm{R}^{\top}\bm{H}_i^{(t)\top}\bm{C}_i(\bm{C}_i^{\top}\bm{C}_i)^{-1}\bm{\Sigma}_i^{1/2}\|_\text{F}^2\\
            & -2\eta\langle(\bm{L}_{2i}-\bm{L}_{2i}^*)\bm{\Sigma}_i^{1/2}, (\bm{R}^{\top}\bm{R})^{-1}\bm{R}^{\top}\bm{H}_i^{(t)\top}\bm{C}_i(\bm{C}_i^{\top}\bm{C}_i)^{-1}\bm{\Sigma}_i^{1/2}\rangle\\
            \leq & \|(\bm{L}_{2i}-\bm{L}_{2i}^*)\bm{\Sigma}_i^{1/2}\|_\text{F}^2 + C\eta^2\|\bm{H}_i^{(t)}\|_\text{F}^2 - 2\eta M_{2i}\\
            & +C\epsilon\eta\gamma^{-1}\cdot\|\bm{H}_i^{(t)}\|_\text{F}^2 + C\epsilon\eta\gamma\cdot\|(\bm{L}_{2i}-\bm{L}_{2i}^*)\bm{\Sigma}_i^{1/2}\|_\text{F}^2,
        \end{split}
    \end{equation}
    where $M_{2i}=\langle(\bm{L}_{2i}-\bm{L}_{2i}^*)\bm{\Sigma}_i^{1/2}, \bm{R}^{*\top}\bm{H}_i^{(t)\top}\bm{C}_i^*(\bm{C}_i^{*\top}\bm{C}_i^*)^{-1}\bm{\Sigma}_i^{1/2}\rangle$.\\

    \noindent\textit{Step 2.3} (Combined results)

    \noindent Combining these results, we have
    \begin{equation}
        \begin{split}
            & \|(\bm{C}-\bm{C}^*-\eta\cdot\bm{H}_C^{(t)}\widetilde{\bm{C}}(\widetilde{\bm{C}}^{\top}\widetilde{\bm{C}})^{-1})\bm{\Sigma}_C\|_\text{F}^2\\
            & + \|(\bm{R}-\bm{R}^*-\eta\cdot\bm{H}_R^{(t)}\widetilde{\bm{R}}(\widetilde{\bm{R}}^{\top}\widetilde{\bm{R}})^{-1})\bm{\Sigma}_R\|_\text{F}^2\\
            & + \sum_{i=1}^n\|(\bm{L}_{1i}-\bm{L}_{1i}^*-\eta(\bm{C}^{\top}\bm{C})^{-1}\bm{C}^{\top}\bm{H}_i^{(t)}\bm{R}_i(\bm{R}_i^{\top}\bm{R}_i)^{-1}\bm{P}_{it})\bm{\Sigma}_i^{1/2}\|_\text{F}^2\\
            & + \sum_{i=1}^n\|(\bm{L}_{2i}-\bm{L}_{2i}^*-\eta(\bm{R}^{\top}\bm{R})^{-1}\bm{R}^{\top}\bm{H}_i^{(t)\top}[\bm{R}_i(\bm{R}_i^{\top}\bm{R}_i)^{-1}\bm{P}_{it}^{-\top})\bm{\Sigma}_i^{1/2}\|_\text{F}^2\\
            \leq & (1+C\epsilon\eta\gamma)\text{dist}(\bm{\Theta}^{(t)},\bm{\Theta}^*)^2 + (C\eta^2+C\epsilon\eta\gamma^{-1})\sum_{i=1}^n\|\mathbb{E}[\nabla\mathcal{L}(\bm{B}_i^{(t)};\bm{X}_{ij},y_{ij})]\|_\text{F}^2\\
            &-2\eta\left(M_C+M_R+\sum_{i=1}^n(M_{1i}+M_{2i})\right).
        \end{split}
    \end{equation}
    Hence, it suffices to construct reasonable lower bounds for $M_C+M_R+\sum_{i=1}^n(M_{1i}+M_{2i})$. 

    Denote the tensor $\cm{T}_n\in\mathbb{R}^{p_1\times p_2\times n}$ as
    \begin{equation}
        \begin{split}
            (\cm{T}_n^{(t)})_{(1)} & = [m_1\cdot\mathbb{E}[\nabla\mathcal{L}(\bm{B}_1^{(t)};\bm{X}_{1j},y_{1j})],\dots,m_n\cdot\mathbb{E}[\nabla\mathcal{L}(\bm{B}_n^{(t)};\bm{X}_{nj},y_{nj})]]\\
            (\cm{T}_n^{(t)})_{(2)} & = [m_1\cdot\mathbb{E}[\nabla\mathcal{L}(\bm{B}_1^{(t)};\bm{X}_{1j},y_{1j})]^\top,\dots,m_n\cdot\mathbb{E}[\nabla\mathcal{L}(\bm{B}_n^{(t)};\bm{X}_{nj},y_{nj})]^\top],
        \end{split}
    \end{equation}
    such that $\|\cm{T}_n^{(t)}\|_\text{F}^2 = \sum_{i=1}^nm_i^2\|\mathbb{E}[\nabla\mathcal{L}(\bm{B}_i^{(t)};\bm{X}_{ij},y_{ij})]\|_\text{F}^2 =  \|\bm{H}_C^{(t)}\|_\text{F}^2 = \|\bm{H}_R^{(t)}\|_\text{F}^2 = \sum_{i=1}^n \|\bm{H}_i^{(t)}\|_\text{F}^2$.

    Then, we have
    \begin{equation}
        \begin{split}
            & M_C+M_R+\sum_{i=1}^n(M_{1i}+M_{2i})\\
            = & \langle\cm{T}_n^{(t)},\cm{G}^*\times_1(\bm{C}-\bm{C}^*)\times_2\bm{R}^*\rangle + \langle\cm{T}_n,\cm{G}^*\times_1\bm{C}^*\times_2(\bm{R}-\bm{R}^*)\rangle\\
            & + \frac{1}{n}\sum_{i=1}^n\left[\langle\mathbb{E}[\nabla\mathcal{L}(\bm{B}_i^{(t)})], \bm{C}^*(\bm{L}_{1i}-\bm{L}_{1i}^*)\bm{L}_{2i}^{*\top}\bm{R}^{*\top}+ \bm{C}^*\bm{L}_{1i}^{*}(\bm{L}_{2i}-\bm{L}_{2i}^*)^\top\bm{R}^{*\top}\rangle\right]\\
            = & \langle\cm{T}_n^{(t)},\cm{G}^*\times_1(\bm{C}-\bm{C}^*)\times_2\bm{R}^*+\cm{G}^*\times_1\bm{C}^*\times_2(\bm{R}-\bm{R}^*)+\bm{\Delta}_G\times_1\bm{C}^*\times_2\bm{R}^*\rangle
        \end{split}
    \end{equation}
    where $(\bm{\Delta}_G)_{(1)}=[(\bm{L}_{11}-\bm{L}_{11}^*)\bm{L}_{21}^{*\top}+\bm{L}_{11}^{*}(\bm{L}_{21}-\bm{L}_{21}^*)^\top,\dots,(\bm{L}_{1n}-\bm{L}_{1n}^*)\bm{L}_{2n}^{*\top}+\bm{L}_{1n}^{*}(\bm{L}_{2n}-\bm{L}_{2n}^*)^\top]$. Additionally,
    \begin{equation}
        \begin{split}
            & \cm{G}^*\times_1(\bm{C}-\bm{C}^*)\times_2\bm{R}^*+\cm{G}^*\times_1\bm{C}^*\times_2(\bm{R}-\bm{R}^*)+\bm{\Delta}_G\times_1\bm{C}^*\times_2\bm{R}^*\\
            = & \cm{B} - \cm{B}^* - (\cm{G}-\cm{G}^*-\bm{\Delta}_G)\times_1\bm{C}^*\times_2\bm{R}^* - \cm{G}^*\times_1(\bm{C}-\bm{C}^*)\times_2(\bm{R}-\bm{R}^*)\\
            & - (\cm{G}-\cm{G}^*)\times_1\bm{C}^*\times_2(\bm{R}-\bm{R}^*) - (\cm{G}-\cm{G}^*)\times_1(\bm{C}-\bm{C}^*)\times_2\bm{R}^*\\
            & - (\cm{G}-\cm{G}^*)\times_1(\bm{C}-\bm{C}^*)\times_2(\bm{R}-\bm{R}^*).
        \end{split}
    \end{equation}
    Therefore, by the RCG condition, we have that for all $i=1,2\dots,n$,
    \begin{equation}
        \langle\mathbb{E}[\nabla\mathcal{L}(\bm{B}_i;\bm{X}_{ij},y_{ij})],\bm{B}_i - \bm{B}_i^*\rangle \geq \frac{\alpha_i}{2}\|\bm{B}_i^{(t)} - \bm{B}_i^*\|_\text{F}^2 + \frac{1}{2\beta_i}\|\mathbb{E}[\nabla\mathcal{L}(\bm{B}_i;\bm{X}_{ij},y_{ij})]\|_\text{F}^2
    \end{equation}
    and hence,
    \begin{equation}
        \begin{split}
            \langle\cm{T}_n^{(t)},\cm{B}^{(t)} - \cm{B}^*\rangle & = \sum_{i=1}^n m_i\langle \mathbb{E}[\nabla\mathcal{L}(\bm{B}_i^{(t)};\bm{X}_{ij},y_{ij})], \bm{B}_i^{(t)}-\bm{B}_i^*\rangle \\
            & \geq \sum_{i=1}^n\frac{\alpha_im_i}{2}\|\bm{B}_i^{(t)} - \bm{B}_i^*\|_\text{F}^2 + \sum_{i=1}^n\frac{m_i}{2\beta_i}\|\mathbb{E}[\nabla\mathcal{L}(\bm{B}_i^{(t)};\bm{X}_{ij},y_{ij})]\|_\text{F}^2\\
            & \geq \frac{\ubar{\alpha} m}{2}\|\cm{B}^{(t)} - \cm{B}^*\|_\text{F}^2 + \frac{1}{2\bar{\beta} m}\|\cm{T}_n^{(t)}\|_\text{F}^2,
        \end{split}
    \end{equation}
    where $\ubar{\alpha}=\min_{1\leq i\leq n}(\alpha_i m_i/m)$, $\bar{\beta}=\max_{1\leq i\leq n}(\beta_i m_i/m)$, and $m=n^{-1}\sum_{i=1}^n m_i$.

    Then, we have
    \begin{equation}
        \begin{split}
            & M_C+M_R+\sum_{i=1}^n(M_{1i}+M_{2i})\\
            = & \langle\cm{T}_n^{(t)},\cm{B}^{(t)}-\cm{B}^*\rangle - \langle\cm{T}_n,(\cm{G}-\cm{G}^*-\bm{\Delta}_G)\times_1\bm{C}^*\times_2\bm{R}^*\rangle\\
            & - \langle\cm{T}_n^{(t)},\cm{G}^*\times_1(\bm{C}-\bm{C}^*)\times_2(\bm{R}-\bm{R}^*)\rangle\\
            & - \langle\cm{T}_n^{(t)},(\cm{G}-\cm{G}^*)\times_1\bm{C}^*\times_2(\bm{R}-\bm{R}^*)\rangle\\
            & - \langle\cm{T}_n^{(t)},(\cm{G}-\cm{G}^*)\times_1(\bm{C}-\bm{C}^*)\times_2\bm{R}^*\rangle\\
            & - \langle\cm{T}_n^{(t)},(\cm{G}-\cm{G}^*)\times_1(\bm{C}-\bm{C}^*)\times_2(\bm{R}-\bm{R}^*)\rangle\\
            \geq & \frac{\ubar{\alpha} m}{2}\|\cm{B}^{(t)}-\cm{B}^*\|_\text{F}^2 + \frac{1}{2\bar{\beta} m}\|\cm{T}_n^{(t)}\|_\text{F}^2 - \frac{1}{4\bar{\beta} m}\|\cm{T}_n^{(t)}\|_\text{F}^2\\
            & - 5\bar{\beta} m \|\cm{G}-\cm{G}^*-\bm{\Delta}_G\|_\text{F}^2 - 5\bar{\beta} m \|\cm{G}^*\times_1(\bm{C}-\bm{C}^*)\times_2(\bm{R}-\bm{R}^*)\|_\text{F}^2 \\
            & - 5\bar{\beta} m \|(\cm{G}-\cm{G}^*)\times_2(\bm{R}-\bm{R}^*)\|_\text{F}^2 - 5\bar{\beta} m \|(\cm{G}-\cm{G}^*)\times_1(\bm{C}-\bm{C}^*)\|_\text{F}^2\\
            & - 5\bar{\beta} m \|(\cm{G}-\cm{G}^*)\times_1(\bm{C}-\bm{C}^*)\times_2(\bm{R}-\bm{R}^*)\|_\text{F}^2\\
            \geq & \frac{\ubar{\alpha} m}{2}\|\cm{B}^{(t)}-\cm{B}^*\|_\text{F}^2 + \frac{1}{4\bar{\beta} m}\|\cm{T}_n^{(t)}\|_\text{F}^2 - C\epsilon^2\bar{\beta} m \cdot\text{dist}(\bm{\Theta}^{(t)},\bm{\Theta}^*)^2.
        \end{split}
    \end{equation}

    Combining the above result, we have
    \begin{equation}
        \begin{split}
            & \text{dist}(\bm{\Theta}^{(t+1)},\bm{\Theta}^*)^2\\
            \leq & (1+\zeta)(1+C\epsilon\eta\gamma+C\epsilon^2\eta\bar{\beta} m-C\eta\ubar{\alpha} m)\text{dist}(\bm{\Theta}^{(t)},\bm{\Theta}^*)^2\\
            & + (1+\zeta)(C\eta^2+C\epsilon\eta\gamma^{-1}-C\eta\bar{\beta}^{-1}m^{-1})\|\cm{T}_n^{(t)}\|_\text{F}^2\\
            & + (1+\zeta^{-1})C\eta^2\left(\|\bm{\Delta}_C^{(t)}\|_\text{F}^2 + \|\bm{\Delta}_R^{(t)}\|_\text{F}^2 + \sum_{i=1}^n(\|\bm{\Delta}_{1i}^{(t)}\|_\text{F}^2+\|\bm{\Delta}_{2i}^{(t)}\|_\text{F}^2) \right).
        \end{split}
    \end{equation}
    Taking $\gamma=Cm\ubar{\alpha}^{1/2}\bar{\beta}^{1/2}$, $\epsilon\lesssim\ubar{\alpha}^{1/2}\bar{\beta}^{-1/2}$, $\eta=C\bar{\beta}^{-1}m^{-1}$, and $\zeta\lesssim\ubar{\alpha}\bar{\beta}^{-1}$, we have that
    \begin{equation}
        \begin{split}
            & \text{dist}(\bm{\Theta}^{(t+1)},\bm{\Theta}^*)^2\\
            \leq & (1-C\ubar{\alpha}\bar{\beta}^{-1})\text{dist}(\bm{\Theta}^{(t)},\bm{\Theta}^*)^2 \\
            & + C\ubar{\alpha}^{-1}\bar{\beta}^{-1}m^{-2}\left(\|\bm{\Delta}_C^{(t)}\|_\text{F}^2 + \|\bm{\Delta}_R^{(t)}\|_\text{F}^2 + \sum_{i=1}^n(\|\bm{\Delta}_{1i}^{(t)}\|_\text{F}^2+\|\bm{\Delta}_{2i}^{(t)}\|_\text{F}^2) \right).
        \end{split}
    \end{equation}~
    
    \noindent\textit{Step 3.} (Induction and completion of proof)

    \noindent According to the stability of $\bm{\Delta}_R^{(t)},\bm{\Delta}_C^{(t)},\bm{\Delta}_{11}^{(t)},\bm{\Delta}_{21}^{(t)},\dots,\bm{\Delta}_{1n}^{(t)},\bm{\Delta}_{2n}^{(t)}$, with $\phi\lesssim\ubar{\alpha}^2m^2$, we have
    \begin{equation}
        \begin{split}
            \text{dist}(\bm{\Theta}^{(t+1)},\bm{\Theta}^*)^2 & \leq (1-C\ubar{\alpha}\bar{\beta}^{-1})\text{dist}(\bm{\Theta}^{(t)},\bm{\Theta}^*)^2 + C\ubar{\alpha}^{-1}\bar{\beta}^{-1}m^{-2}\xi.
        \end{split}
    \end{equation}
    As $\xi\lesssim\ubar{\alpha}^2\underline{\sigma}^2m^2$,
    \begin{equation}
        \text{dist}(\bm{\Theta}^{(t+1)},\bm{\Theta}^*)^2 \lesssim \ubar{\alpha}\bar{\beta}^{-1}\underline{\sigma},
    \end{equation}
    which guarantees that the condition \eqref{eq:cond1} holds for the step $t+1$. Then, by induction,
    \begin{equation}
        \text{dist}(\bm{\Theta}^{(t)},\bm{\Theta}^*)^2
        \leq (1-C\ubar{\alpha}\bar{\beta}^{-1})^t\cdot\text{dist}(\bm{\Theta}^{(0)},\bm{\Theta}^*)^2 + C\ubar{\alpha}^{-2}m^{-2}\xi,
    \end{equation}
    for all $t=1,2,\dots,T$.

    By Lemmas \ref{lemma:tensor_perturb} and \ref{lemma:matrix_perturb},
    \begin{equation}
        \|\cm{B}^{(t)}-\cm{B}^*\|_\text{F}^2 \leq C(1-C\ubar{\alpha}\bar{\beta}^{-1})^t\cdot\text{dist}(\bm{\Theta}^{(0)},\bm{\Theta}^*)^2 + C\ubar{\alpha}^{-2}m^{-2}\xi,
    \end{equation}
    and for sufficiently large $T$,
    \begin{equation}
        \max(\|\bm{C}^{(T)}\bm{Q}_{1T}-\bm{C}^*\|,\|\bm{R}^{(T)}\bm{Q}_{2T}-\bm{R}^*\|) \lesssim \underline{\sigma}^{-1}\ubar{\alpha}^{-1}m^{-1}\xi.
    \end{equation}

\end{proof}

\subsection{Proof of Corollaries \ref{cor:1} and \ref{cor:2}}\label{sec:A.2}

\begin{proof}[Proof of Corollary \ref{cor:1}]
    
    Based on the results in Theorem \ref{thm:1}, when
    \begin{equation}
        T \gtrsim \frac{\log(C\alpha^{-2}m\xi/\text{dist}(\bm{\Theta}^{(0)},\bm{\Theta}^*)^2)}{\log(1-C\alpha\beta^{-1})},
    \end{equation}
    the computational error $(1-C\alpha\beta^{-1})^T\text{dist}(\bm{\Theta}^{(0)},\bm{\Theta}^*)^2$ is dominated by the statistical error $C\alpha^{-2}m^{-2}\xi$. Hence,
    \begin{equation}
        \text{dist}(\widehat{\bm{\Theta}},\bm{\Theta}^*) \lesssim \frac{\xi}{\alpha^{2}m^{2}}
    \end{equation}
    and
    \begin{equation}
        \|\widehat{\cm{B}} - \cm{B}^*\|_\text{F}^2 \lesssim \frac{\xi}{\alpha^{2}m^{2}}.
    \end{equation}
    The latter further implies that
    \begin{equation}
        \frac{1}{n}\sum_{i=1}^n\|\widehat{\bm{B}}_i-\bm{B}_i^*\|_\text{F}^2 \lesssim \frac{\xi}{\alpha^2m^2n}.
    \end{equation}

\end{proof}

\newpage

\subsection{Proof of Theorem \ref{thm:2}}\label{sec:A.3}

The proof of Theorem \ref{thm:2} follows a similar approach as that of Theorem \ref{thm:1}. It also consists of three major steps and is constructed in an inductive approach. In the first step, we present the notations used in the proof, and state some key conditions. In the second step, we develop the core convergence analysis. Specifically, given some conditions hold at the step $t$, we present an upper bound for the step $t+1$, i.e.,
\begin{equation}
    \text{dist}(\bm{\Theta}^{(t+1)},\bm{\Theta}^*)^2 \leq (1-C\alpha\beta^{-1})\cdot\text{dist}(\bm{\Theta}^{(t)},\bm{\Theta}^*)^2 + C\alpha^{-1}\beta^{-1}\xi.
\end{equation}
In the last step, we show that the required conditions also hold for the step $t+1$, and complete the proof by induction.

The proposed scaled project gradient descent algorithm is formulated as
\begin{equation}
    \begin{split}
        \bm{C}^{(t+0.5)} = & \bm{C}^{(t)} - \eta\cdot\bm{H}_C^{(t)}\widetilde{\bm{C}}(\widetilde{\bm{C}}^\top\widetilde{\bm{C}})^{-1}\bm{Q}_{1t}^{-1} - \eta\cdot\bm{\Delta}_C^{(t)}(\widetilde{\bm{C}}^{(t)\top}\widetilde{\bm{C}}^{(t)})^{-1/2}\\
        \bm{R}^{(t+0.5)} = & \bm{R}^{(t)} - \eta\cdot\bm{H}_R^{(t)}\widetilde{\bm{R}}(\widetilde{\bm{R}}^\top\widetilde{\bm{R}})^{-1}\bm{Q}_{2t}^{-1} - \eta\cdot\bm{\Delta}_R^{(t)}(\widetilde{\bm{R}}^{(t)\top}\widetilde{\bm{R}}^{(t)})^{-1/2}\\
        \bm{L}_{1i}^{(t)} = & \bm{L}_{1i}^{(t)} - \eta\cdot\bm{Q}_{1t}(\bm{C}^{\top}\bm{C})^{-1}\bm{C}^{\top}\bm{H}_i^{(t)}\bm{R}_i(\bm{R}_i^{\top}\bm{R}_i)^{-1}\bm{P}_{it}^{-1}\\
        & - \eta\cdot(\bm{C}^{(t)\top}\bm{C}^{(t)})^{-1/2}\bm{\Delta}_{1i}^{(t)}(\bm{L}_{2i}^{(t)\top}\bm{R}^{(t)\top}\bm{R}^{(t)}\bm{L}_{2i}^{(t)})^{-1/2},\\
        \bm{L}_{2i}^{(t)} = & \bm{L}_{2i}^{(t)} - \eta\cdot\bm{Q}_{2t}(\bm{R}^{\top}\bm{R})^{-1}\bm{R}^{\top}\bm{H}_i^{(t)\top}\bm{C}_i(\bm{C}_i^{\top}\bm{C}_i)^{-1}\bm{P}_{it}^{\top}\\
        & - \eta\cdot(\bm{R}^{(t)\top}\bm{R}^{(t)})^{-1/2}\bm{\Delta}_{2i}^{(t)}(\bm{L}_{1i}^{(t)\top}\bm{C}^{(t)\top}\bm{C}^{(t)}\bm{L}_{1i}^{(t)})^{-1/2},\\
        \bm{C}^{(t+1)} = & \text{HT}(\bm{C}^{(t+0.5)}(\widetilde{\bm{C}}^{(t+0.5)\top}\widetilde{\bm{C}}^{(t+0.5)})^{1/2},s_1)(\widetilde{\bm{C}}^{(t+0.5)\top}\widetilde{\bm{C}}^{(t+0.5)})^{-1/2}\\
        \bm{R}^{(t+1)} = & \text{HT}(\bm{R}^{(t+0.5)}(\widetilde{\bm{R}}^{(t+0.5)\top}\widetilde{\bm{R}}^{(t+0.5)})^{1/2},s_2)(\widetilde{\bm{R}}^{(t+0.5)\top}\widetilde{\bm{R}}^{(t+0.5)})^{-1/2}
    \end{split}
\end{equation}

First, we analyze the scaled hard thresholding step. For any $t=1,2,\dots,T$,
\begin{equation}
    \begin{split}
        & \|(\bm{C}^{(t+1)}\bm{Q}_{1,t+1}-\bm{C}^*)\bm{\Sigma}_{C}\|_\text{F}^2 + \|(\bm{R}^{(t+1)}\bm{Q}_{2,t+1}-\bm{R}^*)\bm{\Sigma}_{R}\|_\text{F}^2 \\
        + & \sum_{i=1}^n\left\{ \|(\bm{Q}_{1,t+1}^{-1}\bm{L}_{1i}^{(t+1)}\bm{P}_{it}-\bm{L}_{1i}^*)\bm{\Sigma}_i^{1/2}\|_\text{F}^2 + \|(\bm{Q}_{2,t+1}^{-1}\bm{L}_{2i}^{(t+1)}\bm{P}_{i,t+1}^{-\top}-\bm{L}_{2i}^*)\bm{\Sigma}_i^{1/2}\|_\text{F}^2 \right\} \\
        = & \|(\bm{C}^{(t+1)}\bm{Q}_{1,t+1}-\bm{C}^*)_{\bar{S}_{1,t+1}}\bm{\Sigma}_{C}\|_\text{F}^2 + \|(\bm{R}^{(t+1)}\bm{Q}_{2,t+1}-\bm{R}^*)_{\bar{S}_{2,t+1}}\bm{\Sigma}_{R}\|_\text{F}^2 \\
        + & \sum_{i=1}^n\left\{ \|(\bm{Q}_{1,t+1}^{-1}\bm{L}_{1i}^{(t+1)}\bm{P}_{i,t+1}-\bm{L}_{1i}^*)\bm{\Sigma}_i^{1/2}\|_\text{F}^2 + \|(\bm{Q}_{2,t+1}^{-1}\bm{L}_{2i}^{(t+1)}\bm{P}_{i,t+1}^{-\top}-\bm{L}_{2i}^*)\bm{\Sigma}_i^{1/2}\|_\text{F}^2 \right\} \\
        \leq & \|(\bm{C}^{(t+1)}\bm{Q}_{1,t+0.5}-\bm{C}^*)_{\bar{S}_{1,t+1}}\bm{\Sigma}_{C}\|_\text{F}^2 + \|(\bm{R}^{(t+1)}\bm{Q}_{2,t+0.5}-\bm{R}^*)_{\bar{S}_{2,t+1}}\bm{\Sigma}_{R}\|_\text{F}^2 \\
        + & \sum_{i=1}^n\left\{ \|(\bm{Q}_{1,t+0.5}^{-1}\bm{L}_{1i}^{(t+1)}\bm{P}_{i,t+0.5}-\bm{L}_{1i}^*)\bm{\Sigma}_i^{1/2}\|_\text{F}^2 + \|(\bm{Q}_{2,t+0.5}^{-1}\bm{L}_{2i}^{(t+1)}\bm{P}_{i,t+0.5}^{-\top}-\bm{L}_{2i}^*)\bm{\Sigma}_i^{1/2}\|_\text{F}^2 \right\},
    \end{split}
\end{equation}
where $\bm{Q}_{1,t+0.5}$, $\bm{Q}_{2,t+0.5}$, $\bm{L}_{1,t+0.5},\dots,\bm{L}_{n,t+0.5}$ are the optimal alignment matrices for $\bm{C}^{(t+0.5)},\bm{R}^{(t+0.5)}$, $\bm{L}_{1i}^{(t+0.5)}$, $\bm{L}_{2i}^{(t+0.5)}$.

Then, since $\bm{C}^{(t+1)} = \text{HT}(\bm{C}^{(t+0.5)}(\widetilde{\bm{C}}^{(t+0.5)\top}\widetilde{\bm{C}}^{(t+0.5)})^{1/2},s_1)(\widetilde{\bm{C}}^{(t+0.5)\top}\widetilde{\bm{C}}^{(t+0.5)})^{-1/2}$, for $q_1=2\sqrt{s_1^*/(s_1-s_1^*)}$, we have the upper bound
\begin{equation}
    \begin{split}
        & \|(\bm{C}^{(t+1)}\bm{Q}_{1,t+0.5}-\bm{C}^*)_{\bar{S}_{1,t+1}}\bm{\Sigma}_{C}\|_\text{F}^2 \\
        = & \|[\text{HT}(\bm{C}^{(t+0.5)}(\widetilde{\bm{C}}^{(t+0.5)\top}\widetilde{\bm{C}}^{(t+0.5)})^{1/2},s_1)(\widetilde{\bm{C}}^{(t+0.5)\top}\widetilde{\bm{C}}^{(t+0.5)})^{-1/2}\bm{Q}_{1,t+0.5}-\bm{C}^*]_{\bar{S}_{1,t+1}}\bm{\Sigma}_{C}\|_\text{F}^2 \\
        \leq & \|\text{HT}(\bm{C}^{(t+0.5)}_{\bar{S}_{1,t+1}}(\widetilde{\bm{C}}^{(t+0.5)\top}\widetilde{\bm{C}}^{(t+0.5)})^{1/2},s_1)(\widetilde{\bm{C}}^{(t+0.5)\top}\widetilde{\bm{C}}^{(t+0.5)})^{-1/2}\bm{Q}_{1,t+0.5}\bm{\Sigma}_{C}-\bm{C}^*_{\bar{S}_{1,t+1}}\bm{\Sigma}_{C}\|_\text{F}^2 \\
        \leq & \|\text{HT}(\bm{C}^{(t+0.5)}_{\bar{S}_{1,t+1}}(\widetilde{\bm{C}}^{(t+0.5)\top}\widetilde{\bm{C}}^{(t+0.5)})^{1/2},s_1) - \bm{C}^*_{\bar{S}_{1,t+1}}\bm{Q}_{1,t+0.5}^{-1}(\widetilde{\bm{C}}^{(t+0.5)\top}\widetilde{\bm{C}}^{(t+0.5)})^{1/2}\|_\text{F}^2\\
        &\cdot\|(\widetilde{\bm{C}}^{(t+0.5)\top}\widetilde{\bm{C}}^{(t+0.5)})^{-1/2}\bm{Q}_{1,t+0.5}\bm{\Sigma}_{C}\|^2\\
        \leq & (1+q_1)\|\bm{C}^{(t+0.5)}_{\bar{S}_{1,t+1}}(\widetilde{\bm{C}}^{(t+0.5)\top}\widetilde{\bm{C}}^{(t+0.5)})^{1/2} - \bm{C}^*_{\bar{S}_{1,t+1}}\bm{Q}_{1,t+0.5}^{-1}(\widetilde{\bm{C}}^{(t+0.5)\top}\widetilde{\bm{C}}^{(t+0.5)})^{1/2}\|_\text{F}^2\\
        & \cdot\|(\widetilde{\bm{C}}^{(t+0.5)\top}\widetilde{\bm{C}}^{(t+0.5)})^{-1/2}\bm{Q}_{1,t+0.5}\bm{\Sigma}_{C}\|^2\\
        \leq & (1+q_1)\|(\bm{C}^{(t+0.5)}\bm{Q}_{1,t+0.5}-\bm{C}^*)_{\bar{S}_{1,t+1}}\bm{\Sigma}_C\|_\text{F}^2 \cdot \|\bm{\Sigma}_C^{-1}\bm{Q}_{1,t+0.5}^{-1}(\widetilde{\bm{C}}^{(t+0.5)\top}\widetilde{\bm{C}}^{(t+0.5)})^{1/2}\|^2\\
        & \cdot\|(\widetilde{\bm{C}}^{(t+0.5)\top}\widetilde{\bm{C}}^{(t+0.5)})^{-1/2}\bm{Q}_{1,t+0.5}\bm{\Sigma}_{C}\|^2\\
        \leq & (1+q_1)(1+C\epsilon^2)\|(\bm{C}^{(t+0.5)}\bm{Q}_{1,t+0.5}-\bm{C}^*)_{\bar{S}_{1,t+1}}\bm{\Sigma}_C\|_\text{F}^2.
    \end{split}
\end{equation}
Similarly, we have
\begin{equation}
    \|(\bm{R}^{(t+1)}\bm{Q}_{2,t+0.5}-\bm{R}^*)_{\bar{S}_{2,t+1}}\bm{\Sigma}_{R}\|_\text{F}^2 \leq (1+q_2)(1+C\epsilon^2)\|(\bm{R}^{(t+0.5)}\bm{Q}_{2,t+0.5}-\bm{R}^*)_{\bar{S}_{2,t+1}}\bm{\Sigma}_R\|_\text{F}^2.
\end{equation}
Given that $\max(q_1,q_2)\lesssim\alpha\beta^{-1}$ and $\epsilon\lesssim\alpha^{1/2}\beta^{-1/2}$, we have
\begin{equation}
    \begin{split}
        & \|(\bm{C}^{(t+1)}\bm{Q}_{1,t+1}-\bm{C}^*)\bm{\Sigma}_{C}\|_\text{F}^2 + \|(\bm{R}^{(t+1)}\bm{Q}_{2,t+1}-\bm{R}^*)\bm{\Sigma}_{R}\|_\text{F}^2 \\
        + & \sum_{i=1}^n\left\{ \|(\bm{Q}_{1,t+1}^{-1}\bm{L}_{1i}^{(t+1)}\bm{P}_{it}-\bm{L}_{1i}^*)\bm{\Sigma}_i^{1/2}\|_\text{F}^2 + \|(\bm{Q}_{2,t+1}^{-1}\bm{L}_{2i}^{(t+1)}\bm{P}_{i,t+1}^{-\top}-\bm{L}_{2i}^*)\bm{\Sigma}_i^{1/2}\|_\text{F}^2 \right\} \\
        \leq & (1+C\alpha\beta^{-1})\Bigg[\|(\bm{C}^{(t+0.5)}\bm{Q}_{1,t+0.5}-\bm{C}^*)_{\bar{S}_{1,t+1}}\bm{\Sigma}_C\|_\text{F}^2 + \|(\bm{R}^{(t+0.5)}\bm{Q}_{2,t+0.5}-\bm{R}^*)_{\bar{S}_{2,t+1}}\bm{\Sigma}_R\|_\text{F}^2\\
        & + \sum_{i=1}^n\left\{ \|(\bm{Q}_{1,t+0.5}^{-1}\bm{L}_{1i}^{(t+1)}\bm{P}_{i,t+0.5}-\bm{L}_{1i}^*)\bm{\Sigma}_i^{1/2}\|_\text{F}^2 + \|(\bm{Q}_{2,t+0.5}^{-1}\bm{L}_{2i}^{(t+1)}\bm{P}_{i,t+0.5}^{-\top}-\bm{L}_{2i}^*)\bm{\Sigma}_i^{1/2}\|_\text{F}^2 \right\} \Bigg].
    \end{split}
\end{equation}

Second, we derive the upper bounds for the gradient descent step. Following the proof of Theorem \ref{thm:1}, we have
\begin{equation}
    \begin{split}
        & \Bigg[\|(\bm{C}^{(t+0.5)}\bm{Q}_{1,t+0.5}-\bm{C}^*)_{\bar{S}_{1,t+1}}\bm{\Sigma}_C\|_\text{F}^2 + \|(\bm{R}^{(t+0.5)}\bm{Q}_{2,t+0.5}-\bm{R}^*)_{\bar{S}_{2,t+1}}\bm{\Sigma}_R\|_\text{F}^2\\
        & + \sum_{i=1}^n\left\{ \|(\bm{Q}_{1,t+0.5}^{-1}\bm{L}_{1i}^{(t+1)}\bm{P}_{i,t+0.5}-\bm{L}_{1i}^*)\bm{\Sigma}_i^{1/2}\|_\text{F}^2 + \|(\bm{Q}_{2,t+0.5}^{-1}\bm{L}_{2i}^{(t+1)}\bm{P}_{i,t+0.5}^{-\top}-\bm{L}_{2i}^*)\bm{\Sigma}_i^{1/2}\|_\text{F}^2 \right\} \Bigg]\\
        \leq & (1-C\alpha\beta^{-1})\Bigg[\|(\bm{C}^{(t)}\bm{Q}_{1t}-\bm{C}^*)\bm{\Sigma}_C\|_\text{F}^2 + \|(\bm{R}^{(t)}\bm{Q}_{2t}-\bm{R}^*)\bm{\Sigma}_R\|_\text{F}^2\\
        & + \sum_{i=1}^n\left\{ \|(\bm{Q}_{1t}^{-1}\bm{L}_{1i}^{(t)}\bm{P}_{it}-\bm{L}_{1i}^*)\bm{\Sigma}_i^{1/2}\|_\text{F}^2 + \|(\bm{Q}_{2t}^{-1}\bm{L}_{2i}^{(t)}\bm{P}_{it}^{-\top}-\bm{L}_{2i}^*)\bm{\Sigma}_i^{1/2}\|_\text{F}^2 \right\} \Bigg]\\
        & + C\alpha^{-1}\beta^{-1} \left[\|(\bm{\Delta}_{C}^{(t)})_{\bar{S}_{1,t+1}}\|_\text{F}^2 + \|(\bm{\Delta}_{R}^{(t)})_{\bar{S}_{2,t+1}}\|_\text{F}^2 + \sum_{i=1}^n\Big\{\|\bm{\Delta}_{1i}^{(t)}\|_\text{F}^2 + \|\bm{\Delta}_{2i}^{(t)}\|_\text{F}^2\Big\} \right].
    \end{split}
\end{equation}

By the stability of the de-scaled gradients, when $\psi\lesssim\alpha^2$, we have
\begin{equation}
    \textup{dist}(\bm{\Theta}^{(t)},\bm{\Theta}^*)^2 \lesssim (1-C\alpha\beta^{-1})^t \text{dist}(\bm{\Theta}^{(0)},\bm{\Theta}^*)^2 + C\alpha^{-2}\left[\xi_{C,s_1}+\xi_{R,s_2}+\sum_{i=1}^n(\xi_{1i}+\xi_{2i})\right].
\end{equation}

\subsection{Auxiliary lemmas}\label{sec:A.4}

We first state the perturbation bounds for the tensor Tucker decomposition (Lemma 10 in \citet{tong2021accelerating}). Assume that $\bm{F}=(\cm{G},\bm{C},\bm{R})$ and $\bm{F}^*=(\cm{G}^*,\bm{C}^*,\bm{R}^*)$ are aligned, and introduce the following notations $\bm{\Delta}_C=\bm{C}-\bm{C}^*$, $\bm{\Delta}_R=\bm{R}-\bm{R}^*$, and $\bm{\Delta}_G=\cm{G}-\cm{G}^*$.

\begin{lemma}\label{lemma:tensor_perturb}
    Suppose that $\textup{dist}(\bm{F},\bm{F}^*)^2\leq\epsilon^2\underline{\sigma}^2$ for some $\epsilon<1$. Then, the following bounds on spectral norm hold
    \begin{equation}
        \begin{split}
            \max(\|\bm{C}-\bm{C}^*\|,\|\bm{R}-\bm{R}^*\|,\|(\cm{G}-\cm{G}^*)_{(1)}\bm{\Sigma}_C^{-1}\|,\|(\cm{G}-\cm{G}^*)_{(2)}&\bm{\Sigma}_R^{-1}\|)\leq\epsilon,\\
            \|\bm{C}(\bm{C}^\top\bm{C})^{-1}\| \leq (1-\epsilon)^{-1};~~
            \|\bm{C}(\bm{C}^\top\bm{C})^{-1}-\bm{C}^*\| & \leq \frac{\sqrt{2}\epsilon}{1-\epsilon},\\
            \|(\widetilde{\bm{C}}-\widetilde{\bm{C}}^*)\bm{\Sigma}_C^{-1}\| \leq 3\epsilon+3\epsilon^2+\epsilon^3,~~
            \|\widetilde{\bm{C}}(\widetilde{\bm{C}}^\top\widetilde{\bm{C}})^{-1}\bm{\Sigma}_C\| & \leq (1-\epsilon)^{-3},\\
            \|\widetilde{\bm{C}}(\widetilde{\bm{C}}^\top\widetilde{\bm{C}})^{-1}\bm{\Sigma}_C - \widetilde{\bm{C}}\bm{\Sigma}_C^{-1}\| \leq \frac{\sqrt{2}(3\epsilon+3\epsilon^2+\epsilon^3)}{(1-\epsilon)^3},~~
            \|\bm{\Sigma}_C(\widetilde{\bm{C}}^\top\widetilde{\bm{C}})^{-1}\bm{\Sigma}_C\| & \leq (1-\epsilon)^{-6}.
        \end{split}
    \end{equation}
    By symmetry, a corresponding set of bounds hold for $\bm{R}$ and $\widetilde{\bm{R}}$. In addition, we have
    \begin{equation}
        \begin{split}
            &\|\cm{G}\times_1\bm{C}\times_2\bm{R}-\cm{G}^*\times_1\bm{C}^*\times_2\bm{R}^*\|_\textup{F} \\
            & \leq \left(1+\frac{3\epsilon}{2}+\epsilon^2+\frac{\epsilon^3}{4}\right)(\|(\bm{C}-\bm{C}^*)\bm{\Sigma}_C\|_\textup{F} + \|(\bm{R}-\bm{R}^*)\bm{\Sigma}_R\|_\textup{F}+\|\cm{G}-\cm{G}^*\|_\textup{F}).
        \end{split}
    \end{equation}
\end{lemma}

For the matrix factorization $\bm{L}\bm{R}^\top$, we also have the perturbation bound as in Lemmas 9 and 13 of \citet{tong2021accelerating}.

\begin{lemma}\label{lemma:matrix_perturb}
    Denote $\bm{\Delta}_L=\bm{L}-\bm{L}^*$ and $\bm{\Delta}_L=\bm{R}-\bm{R}^*$. Suppose that 
    \begin{equation}
       \|\bm{\Delta}_L\bm{\Sigma}^{1/2}\|_\textup{F}^2+\|\bm{\Delta}_R\bm{\Sigma}^{1/2}\|_\textup{F}^2 \leq \epsilon^2\sigma_r^2(\bm{L}^*\bm{R}^{*\top}).
    \end{equation}
    Then, we have
    \begin{equation}
        \begin{split}
            \max(\|\bm{\Delta}_L\bm{\Sigma}^{-1/2}\|,\|\bm{\Delta}_R\bm{\Sigma}^{-1/2}\|) & \leq \epsilon,\\
            \|\bm{L}(\bm{L}^\top\bm{L})^{-1}\bm{\Sigma}^{1/2}\|  \leq 1+C\epsilon,~~
            \|\bm{R}(\bm{R}^\top\bm{R})^{-1}\bm{\Sigma}^{1/2}\| & \leq 1+C\epsilon,\\
            \|\bm{L}(\bm{L}^\top\bm{L})^{-1}\bm{\Sigma}^{1/2}-\bm{U}^*\| \leq C\epsilon,~~
            \|\bm{R}(\bm{R}^\top\bm{R})^{-1}\bm{\Sigma}^{1/2}-\bm{V}^*\| & \leq C\epsilon.
        \end{split}
    \end{equation}
    In addition, 
    \begin{equation}
        \|\bm{L}\bm{R}^\top-\bm{L}^*\bm{R}^{*\top}\|_\textup{F} \leq \left(1+\frac{1}{2}\max(\|\bm{\Delta}_L\bm{\Sigma}^{-1/2}\|,\|\bm{\Delta}_R\bm{\Sigma}^{-1/2}\|)\right)(\|\bm{\Delta}_L\bm{\Sigma}^{1/2}\|_\textup{F}+\|\bm{\Delta}_R\bm{\Sigma}^{1/2}\|_\textup{F}).
    \end{equation}
\end{lemma}

\section{Statistical Convergence Analysis}\label{append:C}

\subsection{Proof of Theorem \ref{thm:LinearModel}}

For simplicity, we consider the case $m=m_1=\cdots=m_n$, and the proof can be easily extended to the general case with $m\asymp m_1\asymp\cdots\asymp m_n$. For the linear trace regression
\begin{equation}
    Y_{ij} = \langle\bm{X}_{ij},\bm{B}_i^*\rangle + \varepsilon_{ij},\quad j=1,2,\dots,m\quad\text{and}\quad i=1,2,\dots,n,
\end{equation}
the loss function is
\begin{equation}
    \mathcal{L}_n(\bm{\Theta}) = \frac{1}{2}\sum_{i=1}^n\sum_{j=1}^{m}(Y_{ij} - \langle\bm{X}_{ij},\bm{C}\bm{L}_{1i}\bm{L}_{2i}^\top\bm{R}^\top\rangle)^2.
\end{equation}
Then, the de-scaled partial gradients are
\begin{equation}
    \begin{split}
        \bm{G}_C(\bm{\Theta}) & = \sum_{i=1}^n\sum_{j=1}^{m}(\langle\bm{X}_{ij},\bm{B}_i\rangle - Y_{ij})(\bm{e}_i(n)^\top\otimes\bm{X}_{ij})\widetilde{\bm{C}}(\widetilde{\bm{C}}^\top\widetilde{\bm{C}})^{-1/2}, \\
        \bm{G}_R(\bm{\Theta}) & = \sum_{i=1}^n\sum_{j=1}^{m}(\langle\bm{X}_{ij},\bm{B}_i\rangle - Y_{ij})(\bm{e}_i(n)^\top\otimes\bm{X}_{ij}^\top)\widetilde{\bm{R}}(\widetilde{\bm{R}}^\top\widetilde{\bm{R}})^{-1/2}, \\
        \bm{G}_{1i}(\bm{\Theta}) & = \sum_{j=1}^{m}(\langle\bm{X}_{ij},\bm{B}_i\rangle - Y_{ij})(\bm{C}^\top\bm{C})^{-1/2}\bm{C}^\top\bm{X}_{ij}\bm{R}_i(\bm{R}_i^\top\bm{R}_i)^{-1/2}, \\
        \bm{G}_{2i}(\bm{\Theta}) & = \sum_{j=1}^{m}(\langle\bm{X}_{ij},\bm{B}_i\rangle - Y_{ij})(\bm{R}^\top\bm{R})^{-1/2}\bm{R}^\top\bm{X}_{ij}^\top\bm{C}_i(\bm{C}_i^\top\bm{C}_i)^{-1/2}.
    \end{split}
\end{equation}
Let $\widetilde{\bm{C}}(\widetilde{\bm{C}}^\top\widetilde{\bm{C}})^{-1/2}=\breve{\bm{C}}$, $\widetilde{\bm{R}}(\widetilde{\bm{R}}^\top\widetilde{\bm{R}})^{-1/2}=\breve{\bm{R}}$, $\bm{R}_i(\bm{R}_i^\top\bm{R}_i)^{-1/2} = \breve{\bm{R}}_i$, $\bm{C}_i(\bm{C}_i^\top\bm{C}_i)^{-1/2} = \breve{\bm{C}}_i$, $\bm{C}(\bm{C}^\top\bm{C})^{-1/2} = \ddot{\bm{C}}$, and $\bm{R}(\bm{R}^\top\bm{R})^{-1/2} = \ddot{\bm{R}}$. They are matrices with orthonormal columns.

In the following, the proof consists of four steps. In the first two steps, we establish the stability of $\bm{G}_C(\bm{\Theta})$, $\bm{G}_R(\bm{\Theta})$, $\bm{G}_{1i}(\bm{\Theta})$, and $\bm{G}_{2i}(\bm{\Theta})$, respectively. In the third step, we extend the results to the initialization. In the last step, we verify that the data-driven rank selection is asymptotically consistent and conclude the proof. Therefore, in the first two steps, we assume that the ranks $(r,K_1,K_2)$ are known and correct.\\

\noindent \textit{Step 1} (Stability of $\bm{G}_C(\bm{\Theta})$ and $\bm{G}_R(\bm{\Theta})$)

\noindent We first prove the stability of $\bm{G}_C(\bm{\Theta})$ with $p_1=s_1$. Note that $\|\bm{\Delta}_C\|_\text{F}^2=\|\bm{G}_C(\bm{\Theta})-\mathbb{E}[\bm{G}_C(\bm{\Theta})]\|_\text{F}^2$. For the given $\bm{\Theta}$, denote $\breve{\bm{C}}=[\bm{v}_1,\bm{v}_2,\dots,\bm{v}_{K_1}]$, where each $\bm{v}_\ell$ is a $(p_2n)$-dimensional vector with $\bm{v}_{\ell}=(\bm{v}_{\ell1}^\top,\dots,\bm{v}_{\ell n}^\top)^\top$ and each $\bm{v}_{\ell i}$ is a $p_2$-dimensional vector. Since $\breve{\bm{C}}$ is an orthonormal matrix, $\sum_{i=1}^n\|\bm{v}_{\ell i}\|_2^2=1$.

Note that for $1\leq k\leq p_1$ and $1\leq \ell\leq K_1$, the $(k,\ell)$-th entry of $\bm{G}_C(\bm{\Theta})$ is
\begin{equation}
    \begin{split}
        \bm{e}_k^\top(p_1)\bm{G}_C(\bm{\Theta})\bm{e}_\ell(K_1) & = \sum_{i=1}^n\sum_{j=1}^{m_i}(\langle\bm{X}_{ij},\bm{B}_i\rangle-Y_{ij})\bm{e}_k^\top(p_1)(\bm{e}_i(n)^\top\otimes\bm{X}_{ij})\bm{v}_\ell \\
        & = \sum_{i=1}^n\sum_{j=1}^{m_i}(\langle\bm{X}_{ij},\bm{B}_i\rangle-Y_{ij})\bm{e}_k^\top(p_1)\bm{X}_{ij}\bm{v}_{\ell i}.
    \end{split}
\end{equation}
Therefore, the $(k,\ell)$-th entry of $\bm{G}_C(\bm{\Theta})-\mathbb{E}[\bm{G}_C(\bm{\Theta})]$ is
\begin{equation}
    \begin{split}    
        & \underbrace{\sum_{i=1}^n\sum_{j=1}^{m}(-\varepsilon_{ij})\bm{e}_k^\top(p_1)\bm{X}_{ij}\bm{v}_{\ell i}}_{T_{k\ell 1}}\\
        & + \underbrace{\sum_{i=1}^n\sum_{j=1}^{m}
        (\langle\bm{X}_{ij},\bm{B}_i-\bm{B}_i^*\rangle\bm{e}_k^\top(p_1)\bm{X}_{ij}\bm{v}_{\ell i}-\mathbb{E}[\langle\bm{X}_{ij},\bm{B}_i-\bm{B}_i^*\rangle\bm{e}_k^\top(p_1)\bm{X}_{ij}\bm{v}_{\ell i}])}_{T_{k\ell 2}}.
    \end{split}
\end{equation}

Hence, we first focus on $\sum_{i=1}^n\sum_{j=1}^{m}\varepsilon_{ij}\bm{e}_k^\top(p_1)\bm{X}_{ij}\bm{v}_{\ell i}=\sum_i\sum_j \varepsilon_{ij}z_{ijk\ell}$, where $\varepsilon_{ij}$ are independent $\sigma^2$-sub-Gaussian random variables. Given $\{z_{ijk\ell}\}$, $\sum_{i=1}^n\sum_{j=1}^{m}\varepsilon_{ij}z_{ijk\ell}$ is a $(\sigma^2\sum_{i=1}^n\sum_{j=1}^{m}z_{ijk\ell}^2)$-sub-Gaussian random variable. Thus, denote
\begin{equation}
    S_{k\ell} = \sum_{i=1}^n\sum_{j=1}^{m}\varepsilon_{ij}z_{ijk\ell},~~~R_{k\ell} = \sum_{i=1}^n\sum_{j=1}^{m}z_{ijk\ell}^2.
\end{equation}
Then, for all $\lambda>0$,
\begin{equation}
    \mathbb{E}[\exp(\lambda S_{k\ell})|R_{k\ell}] \leq \frac{\lambda^2\sigma^2 R_{k\ell}}{2}.
\end{equation}
Thus, by the standard Chernoff argument, for any $T_1,T_2>0$,
\begin{equation}
    \begin{split}
        & \mathbb{P}[\{S_{k\ell}\geq T_1\}\cap\{R_{k\ell}\leq T_2\}] \\
        = & \inf_{\lambda>0}\mathbb{P}[\{\exp(\lambda S_{k\ell})\geq \exp(\lambda T_1)\}\cap\{R_{k\ell}\leq T_2\}] \\
        = & \inf_{\lambda>0}\mathbb{P}[\{\exp(\lambda S_{k\ell})\mathbb{I}(R_{k\ell}\leq T_2) \geq \exp(\lambda T_1)] \\
        \leq & \inf_{\lambda>0}\exp(-\lambda T_1)\mathbb{E}[\exp(\lambda S_{k\ell})\mathbb{I}(R_{k\ell}\leq T_2)] \\
        = & \inf_{\lambda>0}\exp(-\lambda T_1 + \sigma^2\lambda^2 T_2/2)\mathbb{E}[\exp(\lambda S_{k\ell} - \sigma^2\lambda^2T_2/2)\mathbb{I}(R_{k\ell}\leq T_2)] \\
        \leq & \inf_{\lambda>0}\exp(-\lambda T_1 + \sigma^2\lambda^2 T_2/2)\mathbb{E}[\exp(\lambda S_{k\ell} - \sigma^2\lambda^2R_n/2)]\\
        \leq & \inf_{\lambda>0}\exp(-\lambda T_1 + \sigma^2\lambda^2 T_2/2) = \exp\left(-\frac{T_1^2}{2\sigma^2T_2}\right).
    \end{split}
\end{equation}
By definition, $z_{ijk\ell}$ is $(\kappa^2\|\bm{v}_{\ell i}\|_2^2)$-sub-Gaussian, and $z_{ijk\ell}^2$ is $(4\sqrt{2}\kappa^2\|\bm{v}_{\ell i}\|_2^2,4\kappa^2\|\bm{v}_{\ell i}\|_2^2)$-sub-exponential.
Hence, for any $t>0$,
\begin{equation}
    \mathbb{P}[R_{k\ell}-\mathbb{E}[R_{k\ell}]\geq t] \leq \exp\left(-\min\left(\frac{t^2}{32m\kappa^2},\frac{t}{4\kappa^2}\right)\right)
\end{equation}
and $\mathbb{E}[R_{k\ell}]\leq m\bar{\beta}_x^2$. Hence, we have that
\begin{equation}
    \begin{split}
        & \mathbb{P}[S_{k\ell}\geq C\sqrt{m\log(p_1)}\sigma\bar{\beta}_x]\\
        \leq & \mathbb{P}[\{S_{k\ell}\geq C\sqrt{m\log(p_1)}\sigma\bar{\beta}_x\}\cap\{R_n\leq 2m\bar{\beta}_x^2\}] + \mathbb{P}(R_n>2m\bar{\beta}_x^2)\\
        \leq & \exp\left(-C\log(p_1)\right) + \exp\left(-\frac{m\bar{\beta}_x^4}{32\kappa^2}\right) \leq \exp(-C\log(p_1))+\exp(-Cm\bar{\beta}_x^4\kappa^{-2}).
    \end{split}
\end{equation}
Therefore,
\begin{equation}
    \mathbb{P}\left[\max_{k,\ell}S_{k\ell}\geq C\sqrt{m\log(p_1)}\sigma\bar{\beta}_x\right]\leq C\exp(-C\log(p_1))+\exp(-Cm\bar{\beta}_x^4\kappa^{-2}),
\end{equation}
and with probability at least $1-C\exp(-C\log(p_1))-C\exp(-Cm\bar{\beta}_x^4\kappa^{-2})$,
\begin{equation}
    \sum_{1\leq k\leq p_1}\sum_{1\leq \ell\leq K_1}T_{k\ell 1}^2 \leq Cmp_1\log(p_1)\sigma^2\bar{\beta}_x^2
\end{equation}

Next, we consider the bound of $T_{k\ell2}$. Note that
\begin{equation}
    \langle\bm{X}_{ij},\bm{B}_i-\bm{B}_i^*\rangle\bm{e}_k^\top(p_1)\bm{X}_{ij}\bm{v}_{\ell i} = \text{vec}(\bm{B}_i-\bm{B}_i^*)^\top\text{vec}(\bm{X}_{ij})\text{vec}(\bm{X}_{ij})^\top(\bm{v}_{\ell i}\otimes\bm{e}_k(p_1)).
\end{equation}
Since $\bm{X}_{ij}$ is $\kappa^2$-sub-Gaussian, $\langle\bm{X}_{ij},\bm{B}_i-\bm{B}_i^*\rangle\bm{e}_k^\top(p_1)\bm{X}_{ij}\bm{v}_{\ell i}$ is sub-exponential with parameters $(4\sqrt{2}\kappa^2\|\bm{B}_i-\bm{B}_i^*\|_\text{F}^2\cdot\|\bm{v}_{\ell i}\|_2^2,4\kappa^2\|\bm{B}_i-\bm{B}_i^*\|_\text{F}^2\cdot\|\bm{v}_{\ell i}\|_2^2)$. Thus, we have that, for any $t>0$,
\begin{equation}
    \mathbb{P}(T_{k\ell 2}\geq t)\leq \exp\left(-\min\left(\frac{nt^2}{32m\|\cm{B}-\cm{B}^*\|_\text{F}^2\kappa^2},\frac{nt}{4\|\cm{B}-\cm{B}^*\|_\text{F}^2\kappa^2}\right)\right).
\end{equation}
Hence,
\begin{equation}
    \mathbb{P}\left[\max_{k,\ell}T_{k\ell 2}\geq \sqrt{m\log(p_1)/n}\kappa\|\cm{B}-\cm{B}^*\|_\text{F}\right]\leq \exp(-C\log(p_1)),
\end{equation}
and with probability at least $1-C\exp(-C\log(p_1))$,
\begin{equation}
    \sum_{1\leq k\leq p_1}\sum_{1\leq \ell\leq K_1}T_{k\ell 2}^2 \leq \frac{\kappa^2mp_1\log(p_1)}{n}\|\cm{B}-\cm{B}^*\|_\text{F}^2.
\end{equation}
Combining the results, with probability at least $1-C\exp(-C\log(p_1))-C\exp(-Cm\bar{\beta}_x^4\kappa^{-2})$,
\begin{equation}
    \|\bm{G}_C(\bm{\Theta})-\mathbb{E}[\bm{G}_C(\bm{\Theta})]\|_\text{F}^2 \lesssim \frac{\kappa^2mp_1\log(p_1)}{n}\|\cm{B}-\cm{B}^*\|_\text{F}^2 + mp_1\log(p_1)\sigma^2\bar{\beta}_x^2.
\end{equation}
Similarly, for $\bm{G}_R(\bm{\Theta})$, with probability at least $1-C\exp(-C\log(p_2))-C\exp(-Cm\bar{\beta}_x^4\kappa^{-2})$,
\begin{equation}
    \|\bm{G}_R(\bm{\Theta})-\mathbb{E}[\bm{G}_R(\bm{\Theta})]\|_\text{F}^2 \lesssim \frac{\kappa^2mp_2\log(p_2)}{n}\|\cm{B}-\cm{B}^*\|_\text{F}^2 + mp_2\log(p_2)\sigma^2\bar{\beta}_x^2.
\end{equation}

More generally, for a generic pair of sparsity levels $(s_1,s_2)$, as in the proof of Theorem \ref{thm:2}, the sparsity levels $(s_1+s_1^*,s_2+s_2^*)$ are directly used in the partial gradient submatrices. Therefore, we have
\begin{equation}
    \|\{\bm{G}_C(\bm{\Theta})-\mathbb{E}[\bm{G}_C(\bm{\Theta})]\}_{S_1}\|_\text{F}^2 \lesssim \frac{\kappa^2ms_1\log(p_1)}{n}\|\cm{B}-\cm{B}^*\|_\text{F}^2 + ms_1\log(p_1)\sigma^2\bar{\beta}_x^2,
\end{equation}
and
\begin{equation}
    \|\{\bm{G}_R(\bm{\Theta})-\mathbb{E}[\bm{G}_R(\bm{\Theta})]\}_{S_2}\|_\text{F}^2 \lesssim \frac{\kappa^2ms_2\log(p_2)}{n}\|\cm{B}-\cm{B}^*\|_\text{F}^2 + ms_2\log(p_2)\sigma^2\bar{\beta}_x^2.
\end{equation}~

\noindent \textit{Step 2} (Stability of $\bm{G}_{1i}(\bm{\Theta})$ and $\bm{G}_{2i}(\bm{\Theta})$)

\noindent Next, we prove the stability of $\bm{G}_{1i}(\bm{\Theta})$ and $\bm{G}_{2i}(\bm{\Theta})$. For the given $\bm{\Theta}$, denote $\ddot{\bm{C}}=[\bm{u}_1,\dots,\bm{u}_K]$ and $\breve{\bm{R}}_i=[\bm{w}_1,\dots,\bm{w}_r]$, where each $\bm{u}_k$ or $\bm{w}_\ell$ is a vector with unit Euclidean norm.

Note that for $1\leq k\leq K_1$ and $1\leq \ell\leq r$, the $(k,\ell)$-th entry of $\bm{G}_{1i}(\bm{\Theta})$ is
\begin{equation}
    \begin{split}
        \bm{e}_k^\top(K_1)\bm{G}_{1i}(\bm{\Theta})\bm{e}_\ell(r) = \sum_{j=1}^m(\langle\bm{X}_{ij},\bm{B}_i\rangle-Y_{ij})\bm{u}_k^\top\bm{X}_{ij}\bm{w}_\ell.
    \end{split}
\end{equation}
Therefore, the $(k,\ell)$-th entry of $\bm{G}_{1i}(\bm{\Theta})-\mathbb{E}[\bm{G}_{1i}(\bm{\Theta})]$ is
\begin{equation}
    \underbrace{\sum_{j=1}^m(-\varepsilon_{ij})\bm{u}_k^\top\bm{X}_{ij}\bm{w}_\ell}_{T_{ik\ell 1}} + \underbrace{\sum_{j=1}^m(\langle\bm{X}_{ij},\bm{B}_i-\bm{B}_i^*\rangle \bm{u}_k^\top\bm{X}_{ij}\bm{w}_\ell - \mathbb{E}[\langle\bm{X}_{ij},\bm{B}_i-\bm{B}_i^*\rangle \bm{u}_k^\top\bm{X}_{ij}\bm{w}_\ell]) }_{T_{ik\ell 2}}.
\end{equation}

For $T_{ik\ell1}$, we consider $\sum_{j=1}^m\varepsilon_{ij}\bm{u}_k^\top\bm{X}_{ij}\bm{w}_\ell=\sum_{j=1}^m\varepsilon_{ij}s_{ijk\ell}$. Given $\{s_{ijk\ell}\}$, $\sum_{j=1}^m\varepsilon_{ij}s_{ijk\ell}$ is a $(\sigma^2\sum_{j=1}^ms_{ijk\ell}^2)$-sub-Gaussian random variable. Thus, denote
\begin{equation}
    S_{ik\ell} = \sum_{j=1}^{m}\varepsilon_{ij}s_{ijk\ell},~~~R_{ik\ell} = \sum_{j=1}^{m}s_{ijk\ell}^2.
\end{equation}
Then, for all $\lambda>0$,
\begin{equation}
    \mathbb{E}[\exp(\lambda S_{ik\ell})|R_{ik\ell}] \leq \frac{\lambda^2\sigma^2 R_{ik\ell}}{2}.
\end{equation}
Thus, by the standard Chernoff argument, for any $T_1,T_2>0$,
\begin{equation}
    \begin{split}
        & \mathbb{P}[\{S_{ik\ell}\geq T_1\}\cap\{R_{ik\ell}\leq T_2\}] \\
        = & \inf_{\lambda>0}\mathbb{P}[\{\exp(\lambda S_{ik\ell})\geq \exp(\lambda T_1)\}\cap\{R_{ik\ell}\leq T_2\}] \\
        = & \inf_{\lambda>0}\mathbb{P}[\{\exp(\lambda S_{ik\ell})\mathbb{I}(R_{ik\ell}\leq T_2) \geq \exp(\lambda T_1)] \\
        \leq & \inf_{\lambda>0}\exp(-\lambda T_1)\mathbb{E}[\exp(\lambda S_{ik\ell})\mathbb{I}(R_{ik\ell}\leq T_2)] \\
        = & \inf_{\lambda>0}\exp(-\lambda T_1 + \sigma^2\lambda^2 T_2/2)\mathbb{E}[\exp(\lambda S_{ik\ell} - \sigma^2\lambda^2T_2/2)\mathbb{I}(R_{ik\ell}\leq T_2)] \\
        \leq & \inf_{\lambda>0}\exp(-\lambda T_1 + \sigma^2\lambda^2 T_2/2)\mathbb{E}[\exp(\lambda S_{ik\ell} - \sigma^2\lambda^2R_{ik\ell}/2)]\\
        \leq & \inf_{\lambda>0}\exp(-\lambda T_1 + \sigma^2\lambda^2 T_2/2) = \exp\left(-\frac{T_1^2}{2\sigma^2T_2}\right).
    \end{split}
\end{equation}
By definition, $s_{ijk\ell}$ is $\kappa^2$-sub-Gaussian, and $s_{ijk\ell}^2$ is $(4\sqrt{2}\kappa^2,4\kappa^2)$-sub-exponential.
Hence, for any $t>0$,
\begin{equation}
    \mathbb{P}[R_{ik\ell}-\mathbb{E}[R_{ik\ell}]\geq t] \leq \exp\left(-\min\left(\frac{t^2}{32m\kappa^2},\frac{t}{4\kappa^2}\right)\right)
\end{equation}
and $\mathbb{E}[R_{ik\ell}]\leq m\bar{\beta}_x^2$. Hence, we have that
\begin{equation}
    \begin{split}
        & \mathbb{P}[S_{ik\ell}\geq C\sqrt{m\log(nK_1)}\sigma\bar{\beta}_x]\\
        \leq & \mathbb{P}[\{S_{ik\ell}\geq C\sqrt{m\log(nK_1)}\sigma\bar{\beta}_x\}\cap\{R_{ik\ell}\leq 2m\bar{\beta}_x^2\}] + \mathbb{P}(R_{ik\ell}>2m\bar{\beta}_x^2)\\
        \leq & \exp\left(-C\log(nK_1)\right) + \exp\left(-\frac{m\bar{\beta}_x^4}{32\kappa^2}\right) \leq \exp(-C\log(n))+\exp(-Cm\bar{\beta}_x^4\kappa^{-2}).
    \end{split}
\end{equation}
Therefore,
\begin{equation}
    \mathbb{P}\left[\max_{i,k,\ell}S_{ik\ell}\geq C\sqrt{m\log(nK_1)}\sigma\bar{\beta}_x\right]\leq C\exp(-C\log(nK_1))+\exp(-Cm\bar{\beta}_x^4\kappa^{-2}),
\end{equation}
and with probability at least $1-C\exp(-C\log(nK_1))-C\exp(-Cm\bar{\beta}_x^4\kappa^{-2})$, for all $i=1,\dots,n$,
\begin{equation}
    \sum_{1\leq k\leq K_1}\sum_{1\leq \ell\leq r}T_{ik\ell 1}^2 \leq CmK_1\log(nK_1)\sigma^2\bar{\beta}_x^2.
\end{equation}

Next, we consider the bound of $T_{ik\ell2}$. Note that
\begin{equation}
    \langle\bm{X}_{ij},\bm{B}_i-\bm{B}_i^*\rangle\bm{u}_{k}^\top\bm{X}_{ij}\bm{w}_\ell = \text{vec}(\bm{B}_i-\bm{B}_i^*)^\top\text{vec}(\bm{X}_{ij})\text{vec}(\bm{X}_{ij})^\top(\bm{w}_\ell\otimes\bm{u}_k).
\end{equation}
Since $\bm{X}_{ij}$ is $\kappa^2$-sub-Gaussian, $\langle\bm{X}_{ij},\bm{B}_i-\bm{B}_i^*\rangle\bm{u}_{k}^\top\bm{X}_{ij}\bm{w}_\ell$ is sub-exponential with parameters $(4\sqrt{2}\kappa^2\|\bm{B}_i-\bm{B}_i^*\|_\text{F}^2,4\kappa^2\|\bm{B}_i-\bm{B}_i^*\|_\text{F}^2)$. Thus, we have that, for any $t>0$,
\begin{equation}
    \mathbb{P}(T_{ik\ell 2}\geq t)\leq \exp\left(-\min\left(\frac{t^2}{32m\|\bm{B}_i-\bm{B}_i^*\|_\text{F}^2\kappa^2},\frac{t}{4\|\bm{B}_i-\bm{B}_i^*\|_\text{F}^2\kappa^2}\right)\right).
\end{equation}
Hence,
\begin{equation}
    \mathbb{P}\left[\max_{i,k,\ell}T_{ik\ell2}\geq\sqrt{m\log(nK_1)}\kappa\|\bm{B}_i-\bm{B}_i^*\|_\text{F}\right]\leq \exp(-C\log(nK_1)),
\end{equation}
and with probability at least $1-C\exp(-C\log(nK_1))$, for any $i=1,\dots,n$,
\begin{equation}
    \sum_{1\leq k\leq K_1}\sum_{1\leq \ell\leq r}T_{ik\ell 2}^2\leq C\kappa^2mK_1\log(nK_1)\|\bm{B}_i-\bm{B}_i^*\|_\text{F}^2.
\end{equation}

Hence, with probability at least $1-C\exp(-C\log(nK_1))-C\exp(-Cm\bar{\beta}_x^4\kappa^{-2})$, for all $1\leq i\leq n$,
\begin{equation}
    \|\bm{G}_{1i}(\bm{\Theta})-\mathbb{E}[\bm{G}_{1i}(\bm{\Theta})]\|_\text{F}^2 \lesssim \kappa^2mK_1\log(nK_1)\|\bm{B}_i-\bm{B}_i^*\|_\text{F}^2 + mK_1\log(nK_1)\sigma^2\bar{\beta}_x^2.
\end{equation}
Similarly, for $\bm{G}_{i2}$, with probability at least $1-C\exp(-C\log(nK_2))-C\exp(-Cm\bar{\beta}_x^4\kappa^{-2})$,
\begin{equation}
    \|\bm{G}_{2i}(\bm{\Theta})-\mathbb{E}[\bm{G}_{2i}(\bm{\Theta})]\|_\text{F}^2 \lesssim \kappa^2mK_2\log(nK_2)\|\bm{B}_i-\bm{B}_i^*\|_\text{F}^2 + mK_2\log(nK_2)\sigma^2\bar{\beta}_x^2.
\end{equation}~

\noindent\textit{Step 3} (Initial error)

\noindent For initialization, the gradient stability in Algorithm \ref{alg:2} can be directly extended from Step 2 by setting $K_1=s_1$ and $K_2=s_2$. In other words, with probability at least $1-C\exp(-C\log(n\bar{s}))-C\exp(-Cm\bar{\beta}_x^4\kappa^{-2})$, for any $i=1,\dots,n$ and generic index sets $S_1$ and $S_2$,
\begin{equation}
    \begin{split}
        \|\{\bm{G}_{C_i}(\bm{\Theta}) - \mathbb{E}[\bm{G}_{C_i}(\bm{\Theta})]\}_{S_1}\|_\text{F}^2 & \lesssim \kappa^2m\bar{s}\log(n\bar{s})\|\bm{B}-\bm{B}_i\|_\text{F}^2 + m\bar{s}\log(n\bar{s})\sigma^2\bar{\beta}_x^2,\\
        \|\{\bm{G}_{R_i}(\bm{\Theta}) - \mathbb{E}[\bm{G}_{R_i}(\bm{\Theta})]\}_{S_2}\|_\text{F}^2 & \lesssim \kappa^2m\bar{s}\log(n\bar{s})\|\bm{B}-\bm{B}_i\|_\text{F}^2 + m\bar{s}\log(n\bar{s})\sigma^2\bar{\beta}_x^2.
    \end{split}
\end{equation}

Therefore, by Corollary \ref{cor:2}, when $m\gtrsim\ubar{\alpha}_x^{-2}\kappa^2\bar{s}\log(n\bar{s})$, for any $i=1,\dots,n$,
\begin{equation}
    \|\widecheck{\bm{B}}_i-\bm{B}_i^*\|_\text{F}^2 \lesssim \frac{\bar{s}\log(n\bar{s})\sigma^2\ubar{\alpha}_x^{-2}\bar{\beta}_x^2}{m}
\end{equation}
and when $m\gtrsim \bar{s}\bar{\beta}_x^3\ubar{\alpha}_x^{-3}$,
\begin{equation}
    \text{dist}(\widecheck{\bm{\Theta}},\bm{\Theta}^*)^2 \lesssim nm^{-1}\bar{s}\log(n\bar{s})\sigma^2\ubar{\alpha}_x^{-2}\bar{\beta}_x^2 \lesssim \ubar{\alpha}_x\bar{\beta}_x^{-1}\ubar{\sigma}^2.
\end{equation}
Therefore, the intial error bound requirement is satisfied if $m\gtrsim\bar{s}\log(n\bar{s})\bar{\beta}_x^3\ubar{\alpha}_x^{-3}$.\\

\noindent\textit{Step 4} (Rank selection consistency)

\noindent By Step 3, for a finite $\bar{r}$, we have that for any $i=1,\dots,n$,
\begin{equation}
    \|\widecheck{\bm{B}}_i(\bar{r})-\bm{B}_i^*\|_\text{F}^2 \lesssim \frac{\bar{s}\log(n\bar{s})\sigma^2\ubar{\alpha}_x^{-2}\bar{\beta}_x^2}{m}.
\end{equation}
Obviously, $\text{rank}(\widecheck{\bm{B}}_i(\bar{r}) - \bm{B}_i^*)\leq\bar{r}+r$. By definition,
\begin{equation}
    \|\widecheck{\bm{B}}_i(\bar{r})-\bm{B}_i^*\|_\text{F}^2 = \sum_{j=1}^{\bar{r}+r}\sigma_j^2(\widecheck{\bm{B}}_i(\bar{r})-\bm{B}_i^*).
\end{equation}
By Mirsky's singular value inequality,
\begin{equation}
    \sum_{j=1}^{\bar{r}+r}[\sigma_j(\widecheck{\bm{B}}_i(\bar{r}))-\sigma_j(\bm{B}_i^*)]^2 \leq \sum_{j=1}^{\bar{r}+r}\sigma_j^2(\widecheck{\bm{B}}_i(\bar{r})-\bm{B}_i^*) = \|\widecheck{\bm{B}}_i(\bar{r})-\bm{B}_i^*\|_\text{F}^2.
\end{equation}
It is followed by
\begin{equation}
    \max_{1\leq j\leq \bar{r}+r}|\sigma_j(\widecheck{\bm{B}}_i(\bar{r}))-\sigma_j(\bm{B}_i^*)| \leq \left\{\sum_{j=1}^{\bar{r}+r}[\sigma_j(\widecheck{\bm{B}}_i(\bar{r}))-\sigma_j(\bm{B}_i^*)]^2\right\}^{1/2} = \|\widecheck{\bm{B}}_i(\bar{r})-\bm{B}_i^*\|_\text{F}.
\end{equation}
For any $1\leq j\leq \bar{r}$, as $\delta_1(n,m,\bar{s}) \asymp n\bar{s}^{1/4}m^{-1/4}$ and $m\gtrsim\bar{s}$,
\begin{equation}
    \begin{split}
        & \sum_{i=1}^n\sigma_j(\widecheck{\bm{B}}_i(\bar{r}))+\delta_1(n,m,\bar{s})\\
        & = \sum_{i=1}^n\sigma_j(\bm{B}_i^*) + \sum_{i=1}^n[\sigma_j(\widecheck{\bm{B}}_i(\bar{r}))-\sigma_j(\bm{B}_i^*)] + \delta_1(n,m,\bar{s}).
    \end{split}
\end{equation}
For $j>r$, $\sum_{i=1}^n\sigma_j(\bm{B}_i^*)=0$ and $|\sum_{i=1}^n[\sigma_j(\widecheck{\bm{B}}_i(\bar{r}))-\sigma_j(\bm{B}_i^*)]|\lesssim n[m^{-1}\bar{s}\log(n\bar{s})\sigma^2\ubar{\alpha}_x^{-2}\bar{\beta}_x^2]^{1/2}=o(\delta_1(n,m,\bar{s}))$. For $j\leq r$, $\delta_1(n,m,\bar{s})=o(\sum_{i=1}^n\sigma_r(\bm{B}_i^*))=O(n)$. Therefore, for $j>r$, as $m\to\infty$,
\begin{equation}
    \frac{\sum_{i=1}^n\sigma_{j+1}(\widecheck{\bm{B}}_i(\bar{r}))+\delta_1(n,m,\bar{s})}{\sum_{i=1}^n\sigma_{j}(\widecheck{\bm{B}}_i(\bar{r}))+\delta_1(n,m,\bar{s})} \to \frac{\delta_1(n,m,\bar{s})}{\delta_1(n,m,\bar{s})} = 1.
\end{equation}
For $j<r$,
\begin{equation}
    \frac{\sum_{i=1}^n\sigma_{j+1}(\widecheck{\bm{B}}_i(\bar{r}))+\delta_1(n,m,\bar{s})}{\sum_{i=1}^n\sigma_{j}(\widecheck{\bm{B}}_i(\bar{r}))+\delta_1(n,m,\bar{s})} \to \frac{\sum_{i=1}^n\sigma_{j+1}(\widecheck{\bm{B}}_i^*)}{\sum_{i=1}^n\sigma_{j}(\widecheck{\bm{B}}_i^*)}.
\end{equation}
For $j=r$,
\begin{equation}
    \frac{\sum_{i=1}^n\sigma_{j+1}(\widecheck{\bm{B}}_i(\bar{r}))+\delta_1(n,m,\bar{s})}{\sum_{i=1}^n\sigma_{j}(\widecheck{\bm{B}}_i(\bar{r}))+\delta_1(n,m,\bar{s})} \to \frac{\delta_1(n,m,\bar{s})}{\sum_{i=1}^n\sigma_{r}(\widecheck{\bm{B}}_i^*)}=o\left(\min_{1\leq j\leq r-1}\frac{\sum_{i=1}^n\sigma_{j+1}(\widecheck{\bm{B}}_i^*)}{\sum_{i=1}^n\sigma_{j}(\widecheck{\bm{B}}_i^*)}\right).
\end{equation}
The selection consistency of $\widehat{r}$ stems from these results.

Next, by the approximate equivalence of $\|\widecheck{\bm{B}}_i-\bm{B}_i^*\|_\text{F}^2$ and $\|(\widecheck{\bm{C}}_i-\bm{C}_i^*)\bm{\Sigma}_i^{1/2}\|_\textup{F}^2+\|(\widecheck{\bm{R}}_i-\bm{R}_i^*)\bm{\Sigma}_i^{1/2}\|_\textup{F}^2$, we have that
\begin{equation}
    \left\|\bm{M}_C-\sum_{i=1}^n\bm{C}_i^*\bm{C}_i^{*\top}\right\|^2 \leq nm^{-1}\bar{s}\log(n\bar{s})\sigma^2\ubar{\alpha}_x^{-2}\bar{\beta}_x^2.
\end{equation}
In addition, by Assumption \ref{asmp:LinearModel3}, $\|\sum_{i=1}^n\bm{C}_i^*\bm{C}_i^{*\top}\|^2=n$. Therefore, by the similar arguments for the selection of $r$, when $\delta_2\asymp nm^{-1/2}\bar{s}^{1/2}$, the selection consistency of $K_1$ can be established. Similarly, the rank selection consistency of $K_2$ can be derived.

\subsection{Generalized Linear Model}

Similar to Theorem \ref{thm:LinearModel}, we also consider the simple case with $m=m_1=\cdots=m_n$, and the proof can be easily extended to $m\asymp m_1\asymp\cdots m_n$. We consider the generalized linear model with loss function
\begin{equation}
    \mathcal{L}(\bm{B}_i;\bm{X}_{ij},Y_{ij})=g(\langle\bm{X}_{ij},\bm{B}_i\rangle)-Y_{ij}\langle\bm{X}_{ij},\bm{B}_i\rangle,\quad j=1,\dots,m\quad\text{and}\quad i=1,\dots,n.
\end{equation}
Then, the de-scaled partial gradients are
\begin{equation}
    \begin{split}
        \bm{G}_C(\bm{\Theta}) & = \sum_{i=1}^n\sum_{j=1}^m [g'(\langle\bm{X}_{ij},\bm{B}_i\rangle)-Y_{ij}](\bm{e}_i(n)^\top\otimes\bm{X}_{ij})\widetilde{\bm{C}}(\widetilde{\bm{C}}^\top\widetilde{\bm{C}})^{-1/2},\\
        \bm{G}_R(\bm{\Theta}) & = \sum_{i=1}^n\sum_{j=1}^m [g'(\langle\bm{X}_{ij},\bm{B}_i\rangle)-Y_{ij}](\bm{e}_i(n)^\top\otimes\bm{X}_{ij}^\top)\widetilde{\bm{R}}(\widetilde{\bm{R}}^\top\widetilde{\bm{R}})^{-1/2},\\
        \bm{G}_{1i}(\bm{\Theta}) & = \sum_{j=1}^m [g'(\langle\bm{X}_{ij},\bm{B}_i\rangle)-Y_{ij}](\bm{C}^\top\bm{C})^{-1/2}\bm{C}^\top\bm{X}_{ij}\bm{R}_i(\widetilde{\bm{R}}^\top\widetilde{\bm{R}})^{-1/2},\\
        \bm{G}_{2i}(\bm{\Theta}) & = \sum_{j=1}^m [g'(\langle\bm{X}_{ij},\bm{B}_i\rangle)-Y_{ij}](\bm{R}^\top\bm{R})^{-1/2}\bm{R}^\top\bm{X}_{ij}^\top\bm{C}_i(\widetilde{\bm{C}}^\top\widetilde{\bm{C}})^{-1/2},
    \end{split}
\end{equation}
where $g'(\langle\bm{X}_{ij},\bm{B}_i\rangle)-Y_{ij}$ is a bounded variable due to the $L$-Lipschitz condition of $g(\cdot)$ and $|Y_{ij}|\leq B$.

Therefore, by the same argument as in the proof of Theorem \ref{thm:LinearModel}, with a high probability,
\begin{equation}
    \begin{split}    
        \|\{\bm{G}_C(\bm{\Theta})-\mathbb{E}[\bm{G}_C(\bm{\Theta})]\}_{S_1}\|_\text{F}^2 & \lesssim \frac{\kappa^2ms_1\log(p_1)}{n}\|\cm{B}-\cm{B}^*\|_\text{F}^2 + ms_1\log(p_1)(L+B)^2\bar{\beta}_x^2,\\
        \|\{\bm{G}_R(\bm{\Theta})-\mathbb{E}[\bm{G}_R(\bm{\Theta})]\}_{S_2}\|_\text{F}^2 & \lesssim \frac{\kappa^2ms_2\log(p_2)}{n}\|\cm{B}-\cm{B}^*\|_\text{F}^2 + ms_2\log(p_2)(L+B)^2\bar{\beta}_x^2,
    \end{split}
\end{equation}
for any index sets $S_1$ and $S_2$ such that $\text{card}(S_1)\leq s_1$ and $\text{card}(S_2)\leq s_2$. In addition, with a high probability,
\begin{equation}
    \|\bm{G}_{1i}(\bm{\Theta})-\mathbb{E}[\bm{G}_{1i}(\bm{\Theta})]\|_\text{F}^2 \lesssim \kappa^2mK_1\log(nK_1)\|\bm{B}_i-\bm{B}_i^*\|_\text{F}^2 + mK_1\log(nK_1)(L+B)^2\bar{\beta}_x^2,
\end{equation}
and
\begin{equation}
    \|\bm{G}_{2i}(\bm{\Theta})-\mathbb{E}[\bm{G}_{2i}(\bm{\Theta})]\|_\text{F}^2 \lesssim \kappa^2mK_2\log(nK_2)\|\bm{B}_i-\bm{B}_i^*\|_\text{F}^2 + mK_2\log(nK_2)(L+B)^2\bar{\beta}_x^2.
\end{equation}

For initialization, the gradient stability in Algorithm \ref{alg:2} can be directly extended from Step 2 by setting $K_1=s_1$ and $K_2=s_2$. Hence, with a high probability, for $i=1,\dots,n$ and generic index sets $S_1$ and $S_2$,
\begin{equation}
    \begin{split}
        \|\{\bm{G}_{C_i}(\bm{\Theta})-\mathbb{E}[\bm{G}_{C_i}(\bm{\Theta})]\}_{S_1}\|_\text{F}^2 & \lesssim\kappa^2m\bar{s}\log(n\bar{s})\|\bm{B}_i-\bm{B}_i^*\|_\text{F}^2 + m\bar{s}\log(n\bar{s})(L+B)^2\bar{\beta}_x^2,\\
        \|\{\bm{G}_{R_i}(\bm{\Theta})-\mathbb{E}[\bm{G}_{R_i}(\bm{\Theta})]\}_{S_2}\|_\text{F}^2 & \lesssim\kappa^2m\bar{s}\log(n\bar{s})\|\bm{B}_i-\bm{B}_i^*\|_\text{F}^2 + m\bar{s}\log(n\bar{s})(L+B)^2\bar{\beta}_x^2.
    \end{split}
\end{equation}
Therefore, the initial error bound requirement is satisfied if $m\gtrsim\bar{s}\log(n\bar{s})\bar{\beta}_x^3\ubar{\alpha}_x^{-3}$. Finally, the rank selection consistency for $\widehat{r}$, $\widehat{K}_1$ and $\widehat{K}_2$ are established in the same manner as Theorem \ref{thm:LinearModel}.

\section{Minimax Optimality Results}\label{append:D}

\subsection{Proof of Theorem \ref{thm:LinearModelLower}}\label{sec:C1}

Let $\mathcal{T}=\{\cm{B}\in\mathbb{R}^{p_1\times p_2\times n}:\text{rank}(\cm{B}_{(1)})\leq r_1,~\text{rank}(\cm{B}_{(1)})\leq r_2\}$. Let $\text{vec}(\bm{X}_{ij})$ follow the multivariate Gaussian distribution with covariance $\bm{\Sigma}_x$ satisfying $\lambda_{\max}(\bm{\Sigma}_x)\leq\bar{\beta}_x$. Consider the augmented covariate tensor $\cm{X}_{ij}\in\mathbb{R}^{p_1\times p_2\times n}$ whose $i$-th slice is $\bm{X}_{ij}$ and others are zero matrices.

Suppose that there exists a finite set $\{\cm{B}^1,\cm{B}^2,\dots,\cm{B}^K\}\in\mathcal{T}$ of tensors such that $\log(K)\gtrsim\bar{s}^*\log(\bar{p})$, such that
\begin{equation}
    \sigma^2\bar{\beta}_x^{-2}\delta^2 \leq \|\cm{B}^{\ell_1}-\cm{B}^{\ell_2}\|_\text{F}^2 \leq 8\sigma^2\bar{\beta}_x^{-2}\delta^2.
\end{equation}
Let $\widetilde{K}$ be a random variable uniformly distributed over the index set $[K]$.
Then, using the standard argument in terms of multple hypothesis testing problem, it yields the lower bound
\begin{equation}
    \inf_{\widetilde{\scalebox{0.7}{\cm{B}}}}\sup_{\scalebox{0.7}{\cm{B}}\in\mathcal{T}}\mathbb{P}\left\{\|\widetilde{\cm{B}}-\cm{B}\|_\text{F}^2\geq \frac{\sigma^2\bar{\beta}_x^{-2}\delta^2}{2}\right\} \geq \inf\mathbb{P}(\widetilde{\cm{B}}\neq \cm{B}^{\widetilde{K}})
\end{equation}
where the infimum is taken over all estimators $\widetilde{\cm{B}}$ that are measurable functions of data.

Let $X=\{\bm{X}_{ij},1\leq i\leq n,1\leq j\leq m\}$ and $Y=\{Y_{ij},1\leq i\leq n,1\leq j\leq m\}$. Using Fano's inequality, for any estimator $\widetilde{\cm{B}}$, we have
\begin{equation}
    \mathbb{P}[\widetilde{\cm{B}} \neq \cm{B}^{\widetilde{m}} | X] \geq 1-\frac{I_X(\cm{B}^{\widetilde{K}};Y)+\log2}{\log K}.
\end{equation}
Taking expecatations over $X$ on both sides, we have
\begin{equation}
    \mathbb{P}[\widetilde{\cm{B}} \neq \cm{B}^{\widetilde{m}}] \geq 1-\frac{\mathbb{E}_X[I_X(\cm{B}^{\widetilde{K}};Y)] + \log2}{\log K}.
\end{equation}

From the convexity of mutual information, we have the upper bound
\begin{equation}
    I_X(\cm{B};Y) \leq \frac{1}{\binom{K}{2}}\sum_{\ell_1,\ell_2=1}^K D_{\text{KL}}(\mathbb{Q}^{\ell_1}|\mathbb{Q}^{\ell_2}),
\end{equation}
where $D_{\text{KL}}(\mathbb{Q}^{\ell_1}|\mathbb{Q}^{\ell_2})$ is the Kulback-Liebler devergence between $\mathbb{Q}^{\ell_1}$ and $\mathbb{Q}^{\ell_2}$ with
\begin{equation}
    D_{\text{KL}}(\mathbb{Q}^{\ell_1}|\mathbb{Q}^{\ell_2}) = \frac{1}{2\sigma^2}\sum_{i=1}^n\sum_{j=1}^m\left(\langle\cm{B}^{\ell_1},\cm{X}_{ij}\rangle - \langle\cm{B}^{\ell_2},\cm{X}_{ij}\rangle\right)^2.
\end{equation}

Taking the expectation over $X$, we have
\begin{equation}
    \mathbb{E}_X[I_X(\cm{B};Y)] \leq \frac{\bar{\beta}_x^2m}{2\sigma^2\binom{K}{2}}\sum_{\ell_1\neq \ell_2}\|\cm{B}^{\ell_1} - \cm{B}^{\ell_2}\|_\text{F}^2 \leq m\delta^2.
\end{equation}
Pluggin it into the previous bound and taking $\delta^2=\bar{s}^*\log(\bar{p})m^{-1} + \log n$, as 
\begin{equation}
    \mathbb{P}[\widetilde{\cm{B}} \neq \cm{B}^{\widetilde{K}}] \geq 1-\frac{\bar{s}\log(\bar{p}) + \log n + \log2}{\log K} \geq \frac{1}{2}.
\end{equation}
Hence, with probability at least $1/2$,
\begin{equation}
    \inf_{\widetilde{\scalebox{0.7}{\cm{B}}}}\sup_{\scalebox{0.7}{\cm{B}}\in\mathcal{T}}\|\widetilde{\cm{B}}-\cm{B}\|_\text{F}^2 \gtrsim \frac{\sigma^2\bar{s}^*\log\bar{p} + \log n}{\bar{\beta}_x^2m}.
\end{equation}

Finally, it suffices to construct a suitable packing for $\mathcal{T}$. As we consider row-wise sparsity in $\bm{C}$ and $\bm{R}$, as well as $\cm{B}_{(1)}$ and $\cm{B}_{(2)}$. The first two packing sets we consider involves select the ${s}^*_1$-dimensional nonzero row index set for $\cm{B}_{(1)}$ and $s^*_2$-dimensional nonzero row index set for $\cm{B}_{(2)}$. After selecting the nonzero index sets, we consider the packing set for the low-rank tensors. By Lemmas \ref{lemma:packing1} and \ref{lemma:packing2}, there exists a set of $\{\cm{B}^1,\cm{B}^2,\dots,\cm{B}^K\}$ such that $\log K\gtrsim \bar{s}^*\log(\bar{p})$ and
\begin{equation}
    \frac{\delta^2}{8} \leq \|\cm{B}^{\ell_1} - \cm{B}^{\ell_2}\|_\text{F}^2 \leq \delta^2,
\end{equation}
for any $\ell_1\neq\ell_2$ and $\delta>0$. It concludes the proof

\subsection{Proof of Theorem \ref{thm:GLMLower}}

Let $\mathcal{T}=\{\cm{B}\in\mathbb{R}^{p_1\times p_2\times n}:\text{rank}(\cm{B}_{(1)})\leq r_1,~\text{rank}(\cm{B}_{(1)})\leq r_2\}$. Let $\text{vec}(\bm{X}_{ij})$ follow the multivariate Gaussian distribution with covariance $\bm{\Sigma}_x$ satisfying $\lambda_{\max}(\bm{\Sigma}_x)\leq\bar{\beta}_x$. Consider the augmented covariate tensor $\cm{X}_{ij}\in\mathbb{R}^{p_1\times p_2\times n}$ whose $i$-th slice is $\bm{X}_{ij}$ and others are zero matrices. Additionally,
\begin{equation}
    \mathbb{P}(Y_{ij}=1|\bm{X}_{ij},\bm{B}_i) = \frac{\exp(\langle\bm{X}_{ij},\bm{B}_i\rangle)}{1+\exp(\langle\bm{X}_{ij},\bm{B}_i\rangle)},\quad\mathbb{P}(Y_{ij}=0|\bm{X}_{ij},\bm{B}_i) = \frac{1}{1+\exp(\langle\bm{X}_{ij},\bm{B}_i\rangle)}.
\end{equation}

Hence, for different $\cm{B}^1$ and $\cm{B}^2$, with corresponding joint distributions $\mathbb{Q}^1$ and $\mathbb{Q}^2$, by \citet{cover1999elements}, the KL divergence between $\mathbb{Q}^1$ and $\mathbb{Q}^2$ is bounded by 
\begin{equation}
    D_{\text{KL}}(\mathbb{Q}^1|\mathbb{Q}^2)\leq \frac{1}{8}(\langle\bm{X}_{ij},\cm{B}^1-\cm{B}^2\rangle)^2.
\end{equation}
Therefore, by the same argument as in the proof of Theorem \ref{thm:LinearModelLower}, we can construct the desired lower bound. We omit the details and refer the readers to Appendix \ref{sec:C1}.

\subsection{Auxiliary Lemmas}

The first lemma is a hypercube packing set for the sparse subset of vectors. That is the set
\begin{equation}
    \mathcal{V}(s) = \{\bm{v}\in\mathbb{R}^p: \|\bm{v}\|_0\leq s\}.
\end{equation}
It follows from Lemma 4 in \citet{raskutti2011minimax}

\begin{lemma}\label{lemma:packing1}
    Let $\mathcal{C}=\{-1,1\}^p$ where $p\geq6$. Then there exists a discrete subset $\{\bm{v}^1,\bm{v}^2,\dots,\bm{v}^J\}\subset\mathcal{V}(s)\cap\mathcal{C}$, such that $\log J\geq Cs\log(p/s)$ for some $C>0$, and for any $\ell_1\neq\ell_2$,
    \begin{equation}
        \frac{\delta^2}{8}\leq \|\bm{v}^{\ell_1}-\bm{v}^{\ell_2}\|_2^2 \leq \delta^2,
    \end{equation}
    for any $\delta>0$.
\end{lemma}

The second lemma is a packing set result from Lemma A.1 in \citet{wang2020compact} that packs into the set of $p_1\times p_2\times p_3$ tensors with Tucker ranks $(r_1,r_2,r_3)$.

\begin{lemma}\label{lemma:packing2}
    Let $\min(p_1,p_2,p_3)\geq10$, and let $\delta>0$. Then for $1\leq r_i\leq p_i$ and $i=1,2,3$, there exists a set of $p_1\times p_2\times p_3$ tensors $\{\cm{A}^1,\cm{A}^2,\dots,\cm{A}^J\}$ with cardinality $\log J\geq C(p_1r_1+p_2r_2+p_3r_3)$ for some constant $C>0$ such that
    \begin{equation}
        \frac{\delta^2}{4}\leq \|\cm{A}^{\ell_1}-\cm{A}^{\ell_2}\|_\textup{F}^2 \leq \delta^2,
    \end{equation}
    for all $\ell_1\neq \ell_2$.
\end{lemma}

\end{appendix}

\bibliographystyle{imsart-nameyear}
\bibliography{mybib.bib}

\end{document}